\documentclass{article}
\usepackage{fullpage}

\usepackage[T1]{fontenc}
\usepackage[utf8]{inputenc}

\usepackage{cite}
\usepackage{mathrsfs}
\usepackage{amstext,amsfonts,amsmath,amssymb}
\usepackage{amsthm}

\usepackage{graphicx,tikz}
\usetikzlibrary{snakes,shapes}
\usetikzlibrary{calc}
\usetikzlibrary{decorations.markings}
\usetikzlibrary{patterns}
\usepackage{comment}
\usepackage{url}
\usepackage{xspace}
\usepackage{hyperref}
\usepackage[absolute]{textpos}
\usepackage[textsize=footnotesize,backgroundcolor=white]{todonotes}
\usepackage[shortlabels]{enumitem}
\usepackage[capitalize]{cleveref}
\usepackage{wrapfig}
\usepackage{rotating}
\usepackage{pdflscape}

\newenvironment{myfig}{\begin{figure*}[tb]}{\end{figure*}}

\def\cqedsymbol{\ifmmode$\lrcorner$\else{\unskip\nobreak\hfil
\penalty50\hskip1em\null\nobreak\hfil$\lrcorner$
\parfillskip=0pt\finalhyphendemerits=0\endgraf}\fi} 

\newcommand{\cqed}{}

\newcommand{\Oh}{\mathcal{O}}
\newcommand{\Q}{\mathbb{Q}}

\newcommand{\NP}[0]{\ensuremath{\mathrm{NP}}\xspace}

\newcommand{\APX}[0]{\ensuremath{\mathrm{APX}}\xspace}

\newcommand{\leqx}{\leq_{\mathrm{x}}}
\newcommand{\leqy}{\leq_{\mathrm{y}}}
\newcommand{\cell}{\mathsf{cell}}
\newcommand{\apxcell}{\mathsf{apxcell}}
\newcommand{\area}{\mathsf{area}}
\newcommand{\apx}{\textsf{apx}}
\newcommand{\Xapx}{X_0}
\newcommand{\Yapx}{Y_0}
\newcommand{\Cells}{\mathsf{Cells}}
\newcommand{\opt}{\textsf{opt}}
\newcommand{\optcell}{\mathsf{optcell}}
\newcommand{\lead}{\phi}
\newcommand{\cont}{\delta}

\newcommand{\leader}{\mathsf{leader}}

\newcommand{\situ}{\sigma}
\newcommand{\btlines}{L}
\newcommand{\alllines}{\btlines'}
\newcommand{\altlines}{\widetilde{\btlines}}
\newcommand{\stlines}{\altlines'}
\newcommand{\seq}{S}
\newcommand{\stseq}{\widetilde{\seq}}

\newcommand{\segment}[3]{#1 [#2,#3 ]}

\newcommand{\Xlines}{X_{\mathsf{lin}}}
\newcommand{\Xapxlines}{\Xlines^\mathsf{apx}}
\newcommand{\Xoptlines}{\Xlines^\mathsf{opt}}
\newcommand{\Xapxsol}{\zeta^{\mathrm{x},\mathsf{apx}}}
\newcommand{\Xoptsol}{\zeta^{\mathrm{x},\mathsf{opt}}}
\newcommand{\Xsol}{\zeta^\mathrm{x}}
\newcommand{\Ylines}{Y_{\mathsf{lin}}}
\newcommand{\Yapxlines}{\Ylines^\mathsf{apx}}
\newcommand{\Yoptlines}{\Ylines^\mathsf{opt}}
\newcommand{\Yapxsol}{\zeta^{\mathrm{y},\mathsf{apx}}}
\newcommand{\Yoptsol}{\zeta^{\mathrm{y},\mathsf{opt}}}
\newcommand{\Ysol}{\zeta^\mathrm{y}}
\newcommand{\apxsol}{\zeta^\mathsf{apx}}
\newcommand{\points}{\mathsf{pts}}

\newcommand{\leadpos}{f}
\newcommand{\leadup}{f^\uparrow}
\newcommand{\leaddown}{f^\downarrow}
\newcommand{\leadact}{\mathsf{active}}
\newcommand{\parent}{\mathsf{parent}}
\newcommand{\treeroot}{\mathsf{root}}
\newcommand{\leaves}{\mathsf{leaves}}
\newcommand{\flf}{\mathsf{type}}
\newcommand{\typeinc}{\mathsf{inc}}
\newcommand{\typedec}{\mathsf{dec}}

\newcommand{\epoch}{\epsilon}

\newcommand{\trees}{\mathsf{trees}}
\newcommand{\tree}{\mathsf{tree}}

\newtheorem{lemma}{Lemma}[section]

\newtheorem{observation}{Observation}

\newtheorem{theorem}[lemma]{Theorem}
\newtheorem{claim}{Claim}
\theoremstyle{definition}
\newtheorem{definition}{Definition}

\newcommand{\optdis}{\textsc{Optimal Discretization}\xspace}

\begin{document}

\title{Optimal Discretization is Fixed-parameter Tractable%
  \thanks{
      This research is a part of a project that has received funding from the European Research Council (ERC) under the European Union's Horizon 2020 research and innovation programme
      Grant Agreement no.~714704 (TM, IM, MP, and MS) and 648527 (IM).
    \newline%
    TM also supported by Charles University, student grant number SVV–2017–260452.
    MS also supported by Alexander von Humboldt-Stiftung.
    MS' main work done while with University of Warsaw and funded by the ERC.
    \newline%
    An extended abstract of this manuscript appeared at ACM-SIAM Symposium on Discrete Algorithms (SODA 2021)~\cite{soda-version}.
  }%
}%

\author{
   Stefan Kratsch\thanks{Institut f\"{u}r Informatik, Humboldt-Universit\"{a}t zu Berlin, Germany, \texttt{kratsch@informatik.hu-berlin.de}}
   \and
  Tom\'{a}\v{s} Masa\v{r}\'{i}k\thanks{Faculty of Mathematics and Physics, Charles University, Czech Republic \& Faculty of Mathematics, Informatics and Mechanics, University of Warsaw, Poland, \texttt{masarik@kam.mff.cuni.cz}} 
  \and Irene Muzi\thanks{Technische Universit\"{a}t Berlin, Germany, \texttt{irene.muzi@tu-berlin.de}} 
    \and Marcin Pilipczuk\thanks{Faculty of Mathematics, Informatics and Mechanics, University of Warsaw, Poland, \texttt{malcin@mimuw.edu.pl}}
    \and Manuel Sorge\thanks{Faculty of Mathematics, Informatics and Mechanics, University of Warsaw, Poland and
      Institute of Logic and Computation, TU Wien, Austria, \texttt{manuel.sorge@ac.tuwien.ac.at}}
}

\date{}

\maketitle

\begin{textblock}{20}(0, 12.0)
\includegraphics[width=40px]{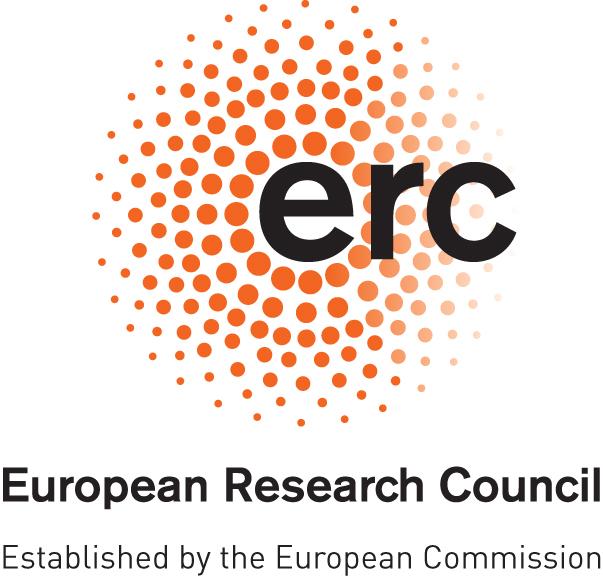}%
\end{textblock}
\begin{textblock}{20}(0, 12.9)
\includegraphics[width=40px]{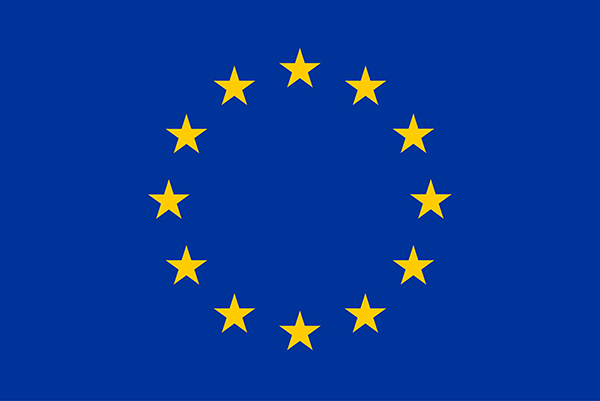}%
\end{textblock}

\begin{abstract}
Given two disjoint sets $W_1$ and $W_2$ of points in the plane, the
\optdis{} problem asks for the minimum size of a family of horizontal and vertical lines
that separate $W_1$ from $W_2$, that is, in every region into which the lines partition the plane
there are either only points of $W_1$, or only points of $W_2$, or the region is empty.
Equivalently, \optdis{} can be phrased as a task of discretizing continuous variables:
We would like to discretize the range of $x$-coordinates and the range of $y$-coordinates
into as few segments as possible, maintaining that no pair of points
from $W_1 \times W_2$ are projected onto the same pair of segments under this discretization. 

We provide a fixed-parameter algorithm for the problem, parameterized by the number of lines
in the solution. Our algorithm works in time $2^{\Oh(k^2 \log k)} n^{\Oh(1)}$, where
$k$ is the bound on the number of lines to find and $n$ is the number of points in the input. 

Our result answers in positive a question of Bonnet, Giannopolous, and Lampis [IPEC 2017]
and of Froese (PhD thesis, 2018) and is in contrast with the known intractability
of two closely related generalizations: the \textsc{Rectangle Stabbing} problem
and the generalization in which the selected lines are not required to be axis-parallel.
\end{abstract}

\section{Introduction}
Separating and breaking geometric objects or point sets, often into clusters, is a common task in computer science.
For example, it is a subtask in divide and conquer algorithms~\cite{jones_search_2021} and appears in machine-learning contexts~\cite{liu_discretization_2002}.
A fundamental problem in this area is to separate two given sets of points by introducing the smallest number of axis-parallel hyperplanes~\cite{Megiddo88}.
This is a classical problem that is even challenging in two dimensions as it is \NP-complete and \APX-hard
already in this case~\cite{Megiddo88,CalinescuDKW05}.
We call the two-dimensional variant of this problem \optdis.
\optdis\ and related problems have been continually studied, in particular with respect to their parameterized complexity~\cite{DomFRS12,giannopoulos_fixedparameter_2013,heggernes_fixedparameter_2013,BonnetGL17,froese_finegrained_2018,banik_parameterized_2020,MisraMS20,dev_redblue_2022}.
Nevertheless, the parameterized complexity status of \optdis\ when parameterized by the number of hyperplanes to introduce remained open~\cite{BonnetGL17,froese_finegrained_2018}.
In this work we show that \optdis\ is fixed-parameter tractable.

Formally, \optdis\ is defined as follows.
For three numbers $a,b,c \in \Q$, we say that \emph{$b$ is between $a$ and $c$} if $a < b< c$ or $c < b< a$. 
The input to \optdis{} consists of two sets $W_1,W_2 \subseteq \Q \times \Q$ and an integer $k$.
A pair $(X,Y)$ of sets $X, Y \subseteq \Q$ is called a \emph{separation} (of $W_1$ and $W_2$) if for every $(x_1,y_1) \in W_1$
and $(x_2,y_2) \in W_2$ there exists an element of $X$ between $x_1$ and $x_2$ or an element of $Y$ between $y_1$ and $y_2$.
We also call the elements of $X \cup Y$ \emph{lines}, for the geometric interpretation of a separation~$(X, Y)$ is as follows.
We draw $|X|$ vertical lines at $x$-coordinates from $X$ and $|Y|$ horizontal lines at $y$-coordinates from $Y$
and focus on partitioning the plane into $(|X|+1)(|Y|+1)$ regions given by the drawn lines.
We require that the closure of every such region does not contain both a point from $W_1$ and a point from $W_2$.
The optimization version of \optdis{} asks for a separation $(X,Y)$ minimizing $|X|+|Y|$; the decision version
takes also an integer $k$ as an input and looks for a separation $(X,Y)$ with $|X|+|Y| \leq k$.

Here we establish fixed-parameter tractability of \optdis{} by showing the following.
\begin{theorem}\label{thm:main}
  \textsc{Optimal Discretization} can be solved in time $2^{\Oh(k^2 \log k)} n^{\Oh(1)}$, where $k$ is the upper bound on the number of lines and $n$ is the number of input points.  
\end{theorem}

\paragraph{Motivation and related work.}
Studying \optdis\ is motivated from three contexts: machine learning, geometric covering problems, and the theory of CSPs.

First, discretization is a preprocessing technique in machine learning in which continuous features of the elements of a data set are discretized.
This is done in order to make the data set amenable to classification by clustering algorithms that work only with discrete features, to speed up algorithms whose running time is sensitive to the number of different feature values, or to improve the interpretability of learning results~\cite{liu_discretization_2002,yang_discretization_2010,witten_data_2011,garcia_survey_2013}.
Various discretization techniques have been studied and are implemented in standard machine learning frameworks~\cite{witten_data_2011,garcia_survey_2013}.
\optdis\ is a formalization of the so-called supervised discretization~\cite{yang_discretization_2010} for two features and two classes; herein, we are given a data set labeled with classes and want to discretize each continuous feature into a minimum number of distinct values so as to avoid mapping two data points with distinct classes onto the same discretized values~\cite{chlebus_finding_1998,froese_finegrained_2018}.
Within this context, fixed-parameter tractability of \optdis{} was posed as an open question by Froese~\cite[Section 5.5]{froese_finegrained_2018}.

Second, being a fundamental geometric problem, \optdis\ and related problems have been studied for a long time in this context.
Indeed, Megiddo~\cite[Proposition 4]{Megiddo88} showed \optdis\ to be \NP-complete in 1988, preceding a later independent \NP-completeness proof by Chlebus and Nguyen~\cite[Corollary 1]{chlebus_finding_1998}.
\optdis\ can be seen as a geometric set covering problem:
A set covering problem is, given a universe $U$ and a family $\mathcal{F}$ of subsets of this universe, to cover the universe~$U$ with the smallest number of subsets from $\mathcal{F}$.
In a geometric covering problem the universe $U$ and the family $\mathcal{F}$ have some geometric relation.
For instance, in \optdis\ the universe $U$ is the set of lines we may select and the family $\mathcal{F}$ is the set of all rectangles that are defined by taking each pair of points $w_1 \in W_1$ and $w_2 \in W_2$ as two antipodal vertices of a rectangle.

Geometric covering problems arise in many different applications such as reducing interference in cellular networks, facility location, or railway-network maintenance and are subject to intensive research (see e.g.\ \cite{agarwal_efficient_1998,giannopoulos_parameterized_2008,DomFRS12,heggernes_fixedparameter_2013,giannopoulos_fixedparameter_2013,bringmann_hitting_2016,ashok_exact_2018,agrawal_parameterized_2020,bonnet_parameterized_2020,banik_parameterized_2020}).
Focusing on the parameterized complexity of geometric covering problems, one can get the impression that they are fixed-parameter tractable when the elements of the universe are pairwise disjoint~\cite{langerman_covering_2005,giannopoulos_parameterized_2008,giannopoulos_fixedparameter_2013,heggernes_fixedparameter_2013} but W[1]-hard when they may non-trivially overlap~\cite{giannopoulos_fixedparameter_2013,DomFRS12,BonnetGL17}.
A particular example from the latter category is the well-studied \textsc{Rectangle Stabbing} problem, wherein the universe $U$ is a set of axis-parallel lines and the family $\mathcal{F}$ a set of axis-parallel rectangles~\cite{hassin_approximation_1991,even_algorithms_2008,DomFRS12,giannopoulos_fixedparameter_2013}.
Similarly closely related but also W[1]-hard is the variant of \optdis\ in which the lines are allowed to have arbitrary slopes, as shown by Bonnet, Giannopolous, and Lampis~\cite{BonnetGL17}.
They also proved that \optdis\ is fixed-parameter tractable under a larger parameterization, namely the cardinality of the smaller of the sets~$W_1$ and $W_2$.
They conjectured that \optdis\ is fixed-parameter tractable with respect to~$k$, which we confirm here.

Approximation algorithms for geometric covering problems have been studied intensively, see the overview by Agarwal and Pan~\cite{agarwal_nearlinear_2020}.
For \optdis, C{\u{a}}linescu, Dumitrescu, Karloff, and Wan~\cite{CalinescuDKW05} obtained a factor-2 approximation in polynomial time, which we use as a subprocedure in our algorithm.
However, they also showed that \optdis\ is \APX-hard.

Third, it turns out that our work is relevant in the framework of determining efficient algorithms for special cases of the \textsc{Constraint Satisfaction Problem} (CSP).
In the proof of Theorem~\ref{thm:main} we ultimately reduce to special form of a CSP.
This special form is over ordered domains of unbounded size and the constraints are binary, that is, they each bind two variables, and they have some restricted form.
In general, such CSPs capture the \textsc{Multicolored Clique} problem~\cite{fellows_parameterized_2009} and are therefore W[1]-hard to solve.
In contrast, we show that our special case is fixed-parameter tractable with some explicit algorithm and an efficient running time.
To the best of our knowledge, this is one of the first approaches that showed a particular problem to be fixed-parameter tractable by a reformulation as a hand-crafted CSP over a large ordered domain, which is then shown to be tractable.
We believe this methodology has the potential to have broader impact. 
Indeed, a more recent work~\cite{multicut-3-term} shows an alternative and more general way how to solve a wider class of binary CSPs, but it comes with the price of a large non-explicit running-time.
The high (but still tractable) running time is caused by use of a meta-theorem: the framework of first-order model checking for structures of bounded twin-width~\cite{BonnetKTW22,tww4}.
In contrast, our algorithm is combinatorial and singly exponential.

\paragraph{Our approach.}
In the proof of Theorem~\ref{thm:main}, we proceed as follows.
Let $(X_0,Y_0)$ be an approximate solution (that can be obtained via, e.g., the iterative compression 
technique or the known polynomial-time $2$-approximation algorithm~\cite{CalinescuDKW05}). 
Let $(X,Y)$ be an optimal solution.
For every two consecutive elements of~$X_0$, we guess (by trying all possibilities) how many (if any) elements of $X$ are between
them and similarly for every two consecutive elements of $Y_0$.
This gives us a general picture of the \emph{layout} of the lines of $X$, $X_0$, $Y$, and $Y_0$.

Consider all $\Oh(k^2)$ \emph{cells} in which the vertical lines with $x$-coordinates from $X_0 \cup X$
and the horizontal lines with $y$-coordinates from $Y_0 \cup Y$ partition the plane.
Similarly, consider all $\Oh(k^2)$ \emph{supercells} in which the vertical lines with $x$-coordinates from $X_0$
and the horizontal lines with $y$-coordinates from $Y_0$ partition the plane.
Every cell is contained in exactly one supercell.
For every cell, guess whether it is empty or contains a point of $W_1 \cup W_2$. 
Note that the fact that $(X_0,Y_0)$ is a solution implies that every supercell contains only 
points from $W_1$, only points from $W_2$, or is empty.
Hence, for each nonempty cell, we can deduce
whether it contains only points of $W_1$ or only points of $W_2$.
Check Figure~\ref{fig:basic} for an example of such a situation.

\begin{myfig}
\begin{center}
\includegraphics{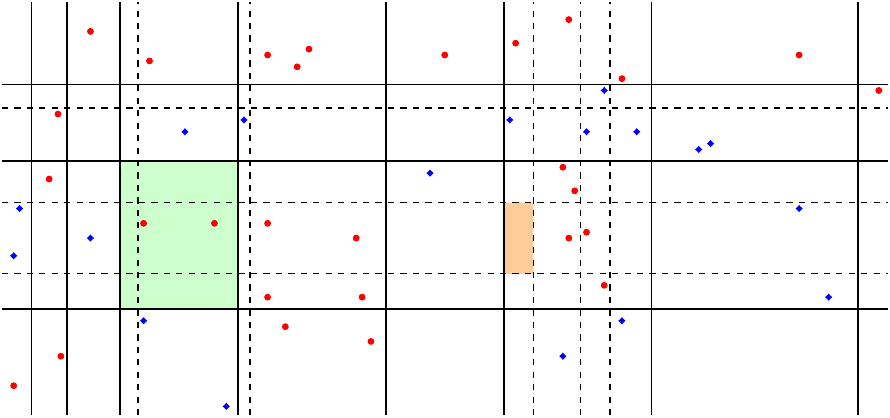}
\caption{Example of a basic situation.
An approximate solution $(X_0,Y_0)$ is denoted by solid lines, an optimal solution $(X,Y)$ by dashed lines.
A supercell is marked in green and a cell in orange.}
\label{fig:basic}
\end{center}
\end{myfig}

We treat every element of $X \cup Y$ as a variable with a domain being all rationals between
the closest lines of $X_0$ or $Y_0$, respectively.

If we know that there exists an optimal solution $(X,Y)$ such that between every two consecutive
elements of $X_0$ there is at most one element of $X$ and between
every two consecutive elements of $Y_0$ there is at most one element of $Y$, we can proceed as above.
For every two consecutive elements of $X_0$, we guess (trying both possibilities) whether there is an element of $X$ between them and similarly for every two consecutive elements of $Y_0$.
This ensures that every cell has at most two borders coming from $X \cup Y$. 
Also, as before, for every cell, we guess whether it is empty.
Thus, for every cell $\mathcal{C}$ that is guessed to be empty and every point $p$ in the supercell containing $\mathcal{C}$ we add a constraint binding
the at most two borders of $\mathcal{C}$ from $X \cup Y$, asserting that $p$ does not land in $\mathcal{C}$.

The crucial observation is that the instance of CSP constructed in this manner admits the median as 
a so-called majority polymorphism, and such CSPs are polynomial-time solvable (for more on majority polymorphisms, which are ternary near-unanimity polymorphisms, see e.g.\ \cite{BartoKW17} or~\cite{CarbonnelC16}).
We remark that Agrawal et al.~\cite{agrawal_parameterized_2020} recently obtained a fixed-parameter algorithm for the \textsc{Art Gallery} problem by reducing it to an equivalent CSP variant, which they called \textsc{Monotone 2-CSP} and directly proved to be polynomial-time solvable.

However, the above approach breaks down if there are multiple lines of $X$ between two consecutive
elements of $X_0$. One can still construct a CSP instance with variables corresponding
to the lines of $X \cup Y$ and constraints asserting that the content of the cells is as we guessed
it to be. Nonetheless, it is possible to show that the constructed CSP instance no longer admits a majority polymorphism.

To cope with that, we perform an involved series of branching and color-coding steps on the instance to clean up the structure of the constructed constraints and obtain a tractable CSP instance.
In Section~\ref{sec:csp}, we introduce the corresponding special CSP variant and prove its tractability via an explicit yet nontrivial branching algorithm.
As already mentioned, the tractability of the obtained CSP instance follows also from recent arguments based on the twin-width; see~\cite{multicut-3-term}. 
We provide a more detailed explanation of this approach later in a paragraph in \cref{sec:overview}.

\medskip
\paragraph{Organization of the paper.}
First, we give a detailed overview of the algorithm in Section~\ref{sec:overview}.
The overview repeats a number of definitions and statements from the full proof. 
Thus, the reader may choose to read only the overview, to get some insight into the main ideas of the proof, or choose to skip the
overview entirely and start with Section~\ref{sec:segments} directly.
A full description of the algorithm for the auxiliary CSP problem is in Sections~\ref{sec:segments} and~\ref{sec:csp}.
We conclude the proof of Theorem~\ref{thm:main} in Section~\ref{sec:redblue}.
There we give an algorithm that constructs a branching tree such that at each branch, the algorithm tries a limited number of options for some property of the solution.
At the leaves we will then assume that the chosen options are correct and reduce the resulting restricted instance of \optdis\ to the auxiliary CSP from Section~\ref{sec:csp}. 
Finally, a discussion of future research directions is in Section~\ref{sec:conclusions}.

\section{Overview}\label{sec:overview}
Let $(W_1,W_2,k)$ be an \optdis{} instance given as input.

\medskip
\paragraph{Layout and cell content.}
As discussed in the introduction, we start by computing a $2$-approximate solution $(X_0,Y_0)$ 
and, in the first branching step, guess the \emph{layout} of $(X_0,Y_0)$ and the sought solution $(X,Y)$, that is,
how many elements of $X$ are between two consecutive elements of $X_0$ and how many elements of $Y$
are between two consecutive elements of~$Y_0$.
By adding a few artificial elements to $X_0$ and $Y_0$ that bound the picture, we can assume that
the first and the last element of $X_0 \cup X$ is from $X_0$
and the first and the last element of $Y_0 \cup Y$ is from $Y_0$. 
Also, simple discretization steps ensure that all elements of $X_0$, $Y_0$, $X$, and $Y$
are integers and $X_0 \cap X = \emptyset$, $Y_0 \cap Y = \emptyset$. 

In the layout, we have $\Oh(k^2)$ \emph{cells} in which the vertical lines 
with $x$-coordinates from $X_0 \cup X$
and the horizontal lines with $y$-coordinates from $Y_0 \cup Y$ partition the plane.
If we only look at the way in which the lines from $X_0$ and $Y_0$ partition the plane, 
we obtain $\Oh(k^2)$ \emph{\apx-supercells}.
If we only look at the way in which the lines from $X$ and $Y$ partition the plane, 
we obtain $\Oh(k^2)$ \emph{\opt-supercells}.
Every cell is in exactly one \apx-supercell and exactly one \opt-supercell.

A second branching step is to guess, for every cell, whether it is empty or not. 
Note that, since $(X_0,Y_0)$ is an approximate solution, a nonempty cell
contains points only from $W_1$ or only from $W_2$, and we can deduce which one is
the case from the instance.
At this moment we verify whether the guess indeed leads to a solution $(X,Y)$:
We reject the current guess if there is an \opt-supercell containing both a cell guessed
to have an element of $W_1$ and a cell guessed to have an element of $W_2$.
Consequently, if we choose $X$ and $Y$ to
ensure that the cells guessed to be empty are indeed empty,
$(X,Y)$ will be a solution to the input instance.

\medskip
\paragraph{CSP formulation.}
We phrase the problem resulting from adding the information guessed above as a CSP instance as follows.
For every sought element $x$ of~$X$,
    we construct a variable whose domain is the set of all integers between the (guessed) closest elements of $X_0$.
Similarly, for every sought element $y$ of $Y$,
we construct a variable whose domain
is the set of all integers between the (guessed) closest elements of $Y_0$.
If between the same two elements of $X_0$ there are multiple elements of $X$, 
we add binary constraints between them that force them to be ordered as we planned in the layout, 
and similarly for $Y_0$ and~$Y$. 
Furthermore, for every cell $\cell$ guessed to be empty and for every point $p$
in the \apx-supercell containing $\cell$, we add a constraint that binds the 
borders of $\cell$ from $X \cup Y$ asserting that $p$ is not in $\cell$. 

Clearly, the constructed CSP instance is equivalent to choosing the values of $X \cup Y$
such that the layout is as guessed and every cell that is guessed to be empty is indeed
empty. This ensures that a satisfying assignment of the constructed CSP instance
yields a solution to the input \optdis{} instance and, in the other direction,
if the input instance is a yes-instance and the guesses were correct, the constructed
CSP instance is a yes-instance. It ``only'' remains to study the tractability of the class
of constructed CSP instances.

\begin{myfig}
\begin{center}
\includegraphics{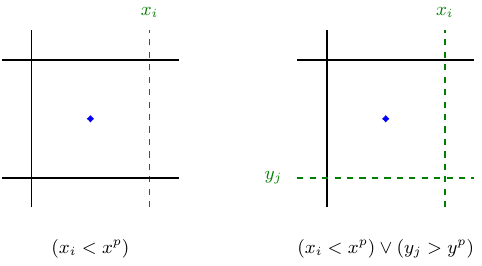}
\caption{Constraints for blocks with one or two borders from $X \cup Y$.}
\label{fig:over:constraints}
\end{center}
\end{myfig}

\medskip
\paragraph{The instructive polynomial-time solvable case.}
As briefly argued in the introduction, if no two elements of $X$ are between two consecutive elements of $X_0$ and no two elements of $Y$ are between two consecutive elements of $Y_0$, then the CSP instance admits the median as a majority polymorphism and therefore is polynomial-time solvable~\cite{BartoKW17}.%
\footnote{We do not define majority polymorphisms here, since they are only instructive in getting an intuition for the difficulties of the problem and will not be needed in the formal proof.
  Intuitively, a CSP with ordered domains admits the median as a majority polymorphism if for any three satisfying assignments, taking for each variable the median of the three assigned values yields another satisfying assignment.
  A polynomial-time algorithm for the special case of CSP considered in this paragraph can also be obtained by a reduction to satisfiability of 2-CNF SAT formulas.}
Let us have a more in-depth look at this argument.

In the above described case, every cell has at most two borders from $X \cup Y$ and thus every introduced
constraint is of arity at most $2$.
As an example, consider
a cell $\cell$ between $x \in X_0$ and $x_i \in X$ and between $y_j \in Y$
and $y \in Y_0$ that is guessed to be empty. For every point $p = (x^p,y^p)$ in the \apx-supercell
containing $\cell$, we add a constraint that $p$ is not in $\cell$.
Observe that this constraint is indeed equivalent to $(x_i < x^p) \vee (y_j > y^p)$, 
        see the right panel of Figure~\ref{fig:over:constraints}.
It is straightforward to verify that a median of three satisfying assignments in this constraint
yields a satisfying assignment as well. 
That is, the median is a majority polymorphism.
Observe also that the constraints yielded for cells $\cell$ with different configurations 
of exactly two borders from $X_0 \cup Y_0$ and exactly two borders from $X \cup Y$ yield similar constraints. 

As a second example, consider a cell $\cell$ with exactly one border from $X \cup Y$,
say between $x \in X$ and $x_i \in X_0$ and between $y,y' \in Y_0$; see the left
panel of Figure~\ref{fig:over:constraints}.
If $\cell$ is guessed to be empty, then for every $p$ in the \apx-supercell containing $\cell$
we add a constraint that $p = (x^p,y^p)$ is not in $\cell$. This constraint
is a unary constraint on $x_i$, in this case $(x_i < x^p)$. 
We can replace this constraint with a filtering step that removes from the domain of $x_i$
the values that do not satisfy it. 

This concludes the sketch why the problem is tractable if there are no two elements
of $X$ between two consecutive elements of $X_0$ and 
no two elements of $Y$ between two consecutive elements of $Y_0$. 

\begin{myfig}
\begin{center}
\includegraphics{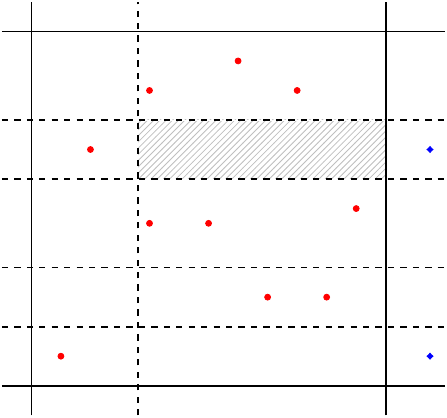}
\caption{Problematic empty cell with three borders from $X \cup Y$.}\label{fig:over:nocorner}
\end{center}
\end{myfig}

\paragraph{Difficult constraints.}
Let us now have a look at what breaks down in the general picture.

First, observe that monotonicity constraints, constraints ensuring that
the lines are in the correct order, are constraints of the form $(x_i < x_{i+1})$
and $(y_i < y_{i+1})$, and thus are simple.
To see this, observe either that the median is again a majority polymorphism for them,
or that $(x_i < x_{i+1})$ can be expressed as a conjunction of constraints
$(x_i \leq a) \vee (x_{i+1} > a)$ over all $a$ in the domain of $x_i$ (thus, we get an instance of a 2-CNF formula).

Second, consider a cell $\cell$ that has all four borders from $X \cup Y$. 
This cell is actually an \opt-supercell as well, contained in a single \apx-supercell. 
Observe that, since $(X_0,Y_0)$ is a solution, this \opt-supercell will never contain 
both a point of $W_1$ and a point of $W_2$, regardless of the choices of the values
of the borders of $\cell$ within their domains. 
Thus, we may ignore the constraints for such cells.

The only remaining cells are cells with three borders from $X \cup Y$. 
As shown on Figure~\ref{fig:over:nocorner}, they can be problematic: In this particular example,
we want the striped cell to be empty of red points, to allow its neighbor to the right
to contain a blue point.
In the construction above, such a cell yields complicated ternary constraints 
on its three borders in~$X \cup Y$. 

\begin{myfig}
\begin{center}
\includegraphics{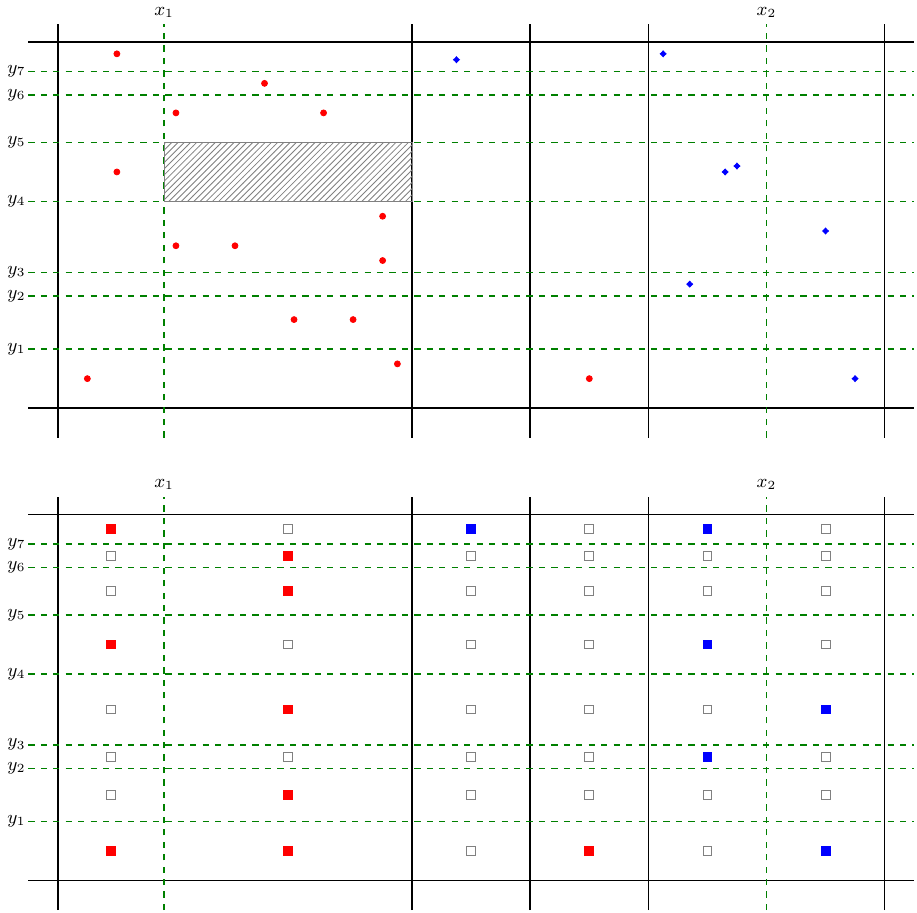}%
\caption{A more complicated scenario where a problematic cell (striped) occurs.
The panel below represents the guessed cell content.}\label{fig:over:binary}%
\end{center}%
\end{myfig}

\medskip
\paragraph{Reduction to binary constraints.}
Our first step to tackle cells with three borders from $X \cup Y$ (henceforth called \emph{ternary cells}) is to break down 
the constraints imposed by them into binary constraints. 

Consider the striped empty cell as in Figure~\ref{fig:over:nocorner}.
It has three
borders from $X \cup Y$: two from $Y$ and one from $X$. 
The two borders from $Y$ are two consecutive elements from $Y$ between the same two elements of $Y_0$ and the border from $X$ is the \emph{last} element of $X$ in the segment\footnote{Throughout, by segment we mean a subset of consecutive elements.} between two elements of $X_0$.
Thus, a constraint imposed by a ternary cell always involves two consecutive elements
of $X$ or $Y$ and first or last element of $Y$ resp. $X$, where first/last refers to a
segment between two consecutive elements of $Y_0$ resp.~$X_0$. 

Consider now an example in Figure~\ref{fig:over:binary}: The top panel represents
the solution, whereas the bottom panel represents the (correctly guessed) layout and cell content.
Assume now that, apart from the layout and cell content, we have somehow learned the value
of $x_1$ (the position of the first vertical green line). 
Consider the area between $x_1$ and $x_2$ in this figure between the top and bottom black line.
Then the set of red points between $x_1$ and $x_2$ is fixed;
moving $x_2$ around only changes the set of blue points. 
Furthermore, the guessed information about layout and cell content determines, if one scans
the area in question from top to bottom, how many \emph{alternations} of \emph{blocks} 
of red and blue points one should
encounter.
For example, scanning the area between $x_1$ and $x_2$ in Figure~\ref{fig:over:binary} from top to bottom
one first encounters a block of blue points, then a block of red points,
then blue, red, blue, and red in the end.

A crucial observation is that, for a fixed value of $x_1$, the set of values of $x_2$
that give the correct alternation (as the guessed cell content has predicted) 
is a segment: Putting $x_2$ too far to the left gives too few alternations due to too few blue points 
and putting $x_2$ too far to the right gives too many alternations due to too many blue points
in the area of interest. 
Furthermore, as the value of $x_1$ moves from left to right, the number of red points decrease,
and the aforementioned ``allowed interval'' of values of $x_2$ moves to the right as well
(one needs more blue points to give the same alternation in the absence of some red points). 

\begin{myfig}
\begin{center}
\includegraphics{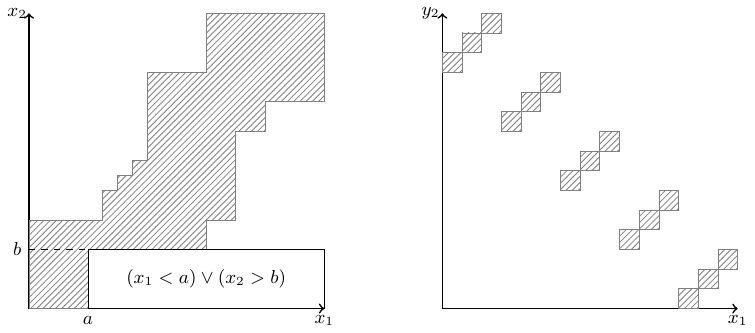}
\caption{Left panel: An example of the space of allowed pairs of values (striped)
  for a constraint binding $x_1$ and $x_2$, asserting that the alternation
    of the area in between is as guessed.
  Such a constraint can be expressed as a conjunction of simple constraints
  as in the figure.
  Right panel: An example of the constraint binding $x_1$ (left vertical green line)
  and $y_2$ (middle horizontal green line) in Figure~\ref{fig:over:reversion},
      which is not that simple.}\label{fig:over:binary2}
\end{center}
\end{myfig}

In the scenario as in Figure~\ref{fig:over:binary}, let us introduce a binary constraint
binding $x_1$ and $x_2$, asserting that the alternation in the discussed area 
is as the cell-content guess predicted. 
From the discussion above we infer that this constraint is of the same type
as the constraints for cells with two borders of $X \cup Y$ and thus simple: It
can be expressed as a conjunction of a number of clauses of the type
$(x_1 < a) \vee (x_2 > b)$ and $(x_1 > a) \vee (x_2 < b)$ for constants $a$ and $b$.
See Figure~\ref{fig:over:binary2}. 

The second crucial observation is that, if one fixes the value of $x_1$, 
then not only does this determine the set of red points in the discussed area, but also how they are partitioned into blocks in the said alternation. Indeed, moving $x_2$ to the right only adds more blue points,
but since the alternation is fixed, more blue points join existing blue blocks instead
of creating new blocks (as this would increase the number of alternations). 
Hence, the value of $x_1$ itself implies the partition of the red points into blocks. 

Now consider the striped cell in Figure~\ref{fig:over:binary}, between lines $y_4$ and $y_5$. 
It is guessed to be empty. 
Instead of handling it with a ternary constraint as before, we handle it as follows:
Apart from the constraint binding $x_1$ and $x_2$ asserting that the alternation is as guessed,
we add
a constraint binding $x_1$ and $y_4$, asserting that, for every fixed value of $x_1$,
   line~$y_4$ is above the second (from the top) red block, and
a constraint binding $x_1$ and $y_5$, asserting that for every fixed value of $x_1$,
   line~$y_5$ is below the first (from the top) red block.
It is not hard to verify that the introduced constraints, together with monotonicity constraints
$(y_i<y_{i+1})$, imply that the striped cell is empty. 

\begin{myfig}
\begin{center}
\includegraphics{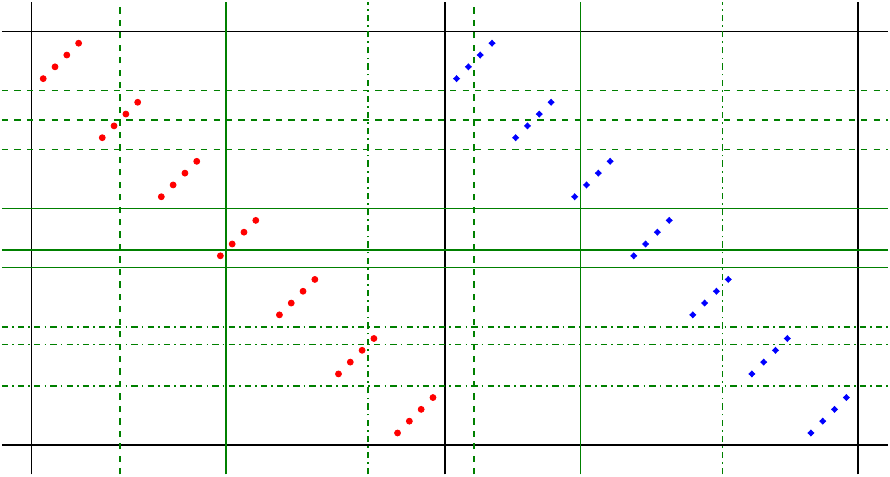}
\caption{An example where the constraint between the left vertical solution line in~$X$ and
  the middle horizontal solution line in~$Y$ is complicated (both green). 
  The three different line styles represent three different possible solutions
  (among many others). 
  Moreover, each of the three depicted solutions has the same alternation (from top to bottom: Blue, Red, Blue, Red).
}\label{fig:over:reversion}
\end{center}
\end{myfig}

Thus, it remains to understand how complicated the constraints binding $x_1$ and $y_4/y_5$ are.
Unfortunately, they may not have the easy ``tractable'' form as the constraints described
so far (e.g., like in Figure~\ref{fig:over:binary2}). 
Consider an example in Figure~\ref{fig:over:reversion}, where a number of possible positions of
the solution lines in $X \cup Y$ (green) have been depicted with various line styles.
The constraint between the left vertical line and the middle horizontal line 
has been depicted in the right panel of Figure~\ref{fig:over:binary2}.
For such a constraint, the median is not necessarily a majority polymorphism.
In particular, such a constraint cannot be expressed in the style in which we expressed all other constraints so far. 

Our approach is now as follows:
\begin{enumerate}
\item Introduce a class of CSP instances that allow constraints both
as in the left and right panel of Figure~\ref{fig:over:binary2} and show that
the problem of finding a satisfying assignment is fixed-parameter tractable
when parameterized by the number of variables.
\item By a series of involved branching and color coding steps, reduce
the \optdis{} instance at hand to a CSP instance from the aforementioned tractable class.
In some sense, our reduction shows that 
the example of Figure~\ref{fig:over:reversion} is the 
most complicated picture one can encode in an \optdis{} instance. 
\end{enumerate}
In the remainder of this overview we focus on the first part above. 
As we shall see in a moment, there is strong resemblance of the introduced class to the constraints
of Figure~\ref{fig:over:binary2}.
The highly technical second part, spanning over most of Section~\ref{sec:redblue}, takes
the analysis of Figure~\ref{fig:over:binary} as its starting point and investigates deeper
how the red/blue blocks change if one moves the lines $x_1$ and $x_2$ around. 

\medskip
\paragraph{Tractable CSP class.}
An instance of \textsc{Forest CSP} consists of a forest $G$, where the vertices $V(G)$
are variables, a domain $[n_T] = \{1,2,\ldots,n_T\}$ for every connected component $T$ of $G$,
shared among all variables of $T$, and a number of constraints, split into two families: segment-reversion and downwards-closed constraints.

A permutation $\pi$ of $[n]$ is a \emph{segment reversion}, if its matrix representation looks for example like this:
$$\begin{matrix}
0 & 0 & 1 & 0 & 0 & 0 & 0 & 0 & 0 & 0  \\
0 & 1 & 0 & 0 & 0 & 0 & 0 & 0 & 0 & 0  \\
1 & 0 & 0 & 0 & 0 & 0 & 0 & 0 & 0 & 0  \\
0 & 0 & 0 & 0 & 0 & 0 & 1 & 0 & 0 & 0  \\
0 & 0 & 0 & 0 & 0 & 1 & 0 & 0 & 0 & 0  \\
0 & 0 & 0 & 0 & 1 & 0 & 0 & 0 & 0 & 0  \\
0 & 0 & 0 & 1 & 0 & 0 & 0 & 0 & 0 & 0  \\
0 & 0 & 0 & 0 & 0 & 0 & 0 & 1 & 0 & 0  \\
0 & 0 & 0 & 0 & 0 & 0 & 0 & 0 & 0 & 1  \\
0 & 0 & 0 & 0 & 0 & 0 & 0 & 0 & 1 & 0  
\end{matrix}$$
Formally, $\pi$ is a segment reversion if there exist integers $1 = a_1 < a_2 < \ldots < a_r = n+1$
such that for every $x \in [n]$, if $i \in [r-1]$ is the unique index such that $a_i \leq x < a_{i+1}$, then 
$\pi(x) = a_{i+1} - 1 - (x - a_i)$. That is, $\pi$ reverses a number of disjoint segments in the domain~$[n]$. 
Note the resemblance of the matrix above and the right panel of Figure~\ref{fig:over:binary2}.

With every edge $e = uv \in E(G)$ in a component $T$, the \textsc{Forest CSP} instance contains a segment reversion~$\pi_e$ of $[n_T]$ and a constraint asserting
that $\pi_e(u) = v$. Note that segment reversions are involutions, so $\pi_e(u) = v$ is equivalent to $\pi_e(v) = u$.
One can think of the whole component $T$ of $G$ as a single super-variable: Setting the value of a single variable in a component $T$
propagates the value over the segment reversions on the edges to the entire tree. Thus, every tree $T$ has $n_T$ different allowed assignments.

\looseness=-1
A relation $R \subseteq [n_1] \times [n_2]$ is \emph{downwards-closed} if $(x,y) \in R$ and $(x' \leq x) \wedge (y' \leq y)$ implies $(x',y') \in R$. 
For every pair of two distinct vertices $u,v \in V(G)$, a \textsc{Forest CSP} instance may contain a downwards-closed relation $R_{u,v}$ and a constraint binding $u$ and $v$ asserting that $(u,v) \in R_{u,v}$. Such a constraint is henceforth called a \emph{downwards-closed constraint}. 
Note that an intersection of two downwards-closed relations is again downwards-closed; thus it would not add more expressive power to the problem to allow multiple downwards-closed constraints
between the same pair of variables.

\looseness=-1
Observe that if one for every $u \in V(G)$ adds a clone $u'$, connected to $u$ with an edge $uu'$ with a segment reversion $\pi_{uu'}$ that reverses the whole domain, 
then with the four downwards closed constraints in $\{u,u'\} \times \{v,v'\}$ one can express any constraint as in the left panel of Figure~\ref{fig:over:binary2}. 

This concludes the description of the \textsc{Forest CSP} problem that asks for a satisfying assignment to the input instance. 

\paragraph{Solving the CSP formulation via twin-width.}
After this work appeared at SODA 2021~\cite{soda-version}, a new point of view on such CSP problems has been developed~\cite{multicut-3-term} which can be used to obtain a fixed-parameter algorithm for \textsc{Forest CSP} as follows.
In this framework, we consider CSPs with domain $[n]$, binary constraints, and parameterized by both the number of variables and constraints.
Every constraint is given as a binary $n \times n$ matrix $M$ and a first-order formula $\phi$ that has two free variables, may reference values in the matrix $M$, and may include comparisons over integers.
A constraint is satisfied by values $(x,y)$ if $\phi(x,y)$ is satisfied.

For example, a constraint as in the left panel of Figure~\ref{fig:over:binary2} can be expressed as follows, that is, a constraint that is a conjuction of an arbitrary number of constraints $\mathcal{C}$ of the form $(x < a) \vee (y > b)$.
First, we select a subset $\mathcal{C}' \subseteq \mathcal{C}$ of \emph{maximal} constraints, that is, constraints
$(x < a) \vee (y > b)$ for which there is no other constraint $(x < a') \vee (y > b')$ with $a' \leq a$ and $b' \geq b$ 
(note that the latter constraint implies the former). 
We initially set all values of $M$ to $0$ and for every constraint $(x < a) \vee (y < b)$ in $\mathcal{C}'$ we set $M[a,b]=1$.
Finally, we express the conjuction of the constraints of $\mathcal{C}$ as the formula
\[ \phi(x,y) = \neg \exists_a \exists_b M[a,b]=1 \wedge x \geq a \wedge y \leq b. \]

We can solve a CSP as above with twin-width related machinery; to explain it we need the following notions: 
A \emph{$k$-partition} of $[n]$ is a partition of $[n]$ into $k$ pairwise disjoint nonempty intervals.
A \emph{$k$-grid minor} of a $0$-$1$ $n \times n$ matrix $M$ is a pair $((I_i)_{i=1}^k, (J_j)_{j=1}^k)$ of $k$-partitions
of $[n]$ such that for every $1 \leq i,j \leq k$ there exists $x \in I_i$ and $y \in J_j$ such that $M[x,y] = 1$. 
The \emph{maximum grid minor size} of $M$ is the maximum $k$ such that $M$ admits a $k$-grid minor. 

An insight of Hatzel et al.~\cite[Theorem 3.1]{multicut-3-term} is that the twin-width machinery of Bonnet et al.~\cite{tww4} can be used
to prove fixed-parameter tractability of the discussed class of CSP instances, when not only the number of variables
and constraints is bounded in parameter, but also the sizes of the used formulae and the maximum grid minor size of the used
matrices~$M$ are bounded.%
\footnote{The article~\cite{multicut-3-term} formally states the result for a more restricted class of CSPs (that still contains our \textsc{Forest CSP})
  but their arguments actually prove tractability of CSPs as described here.}
One can observe that the permutation matrix of a segment reversion has grid minor size at most $2$, 
while the matrix $M$ in the aforementioned example of the encoding of a conjunction of constraints
of the form $(x < a) \vee (y > b)$ has grid minor size at most $1$ (thanks to the subselection
of the set $\mathcal{C}'$).
Combining these observations with Hatzel et al.s' insight we thus immediately obtain fixed-parameter tractability of \textsc{Forest CSP}, parameterized by the number of variables and constraints.
However, the usage of the meta-theorem of~\cite{tww4} results in a very bad dependency on the parameter in the running time bound of the obtained algorithm.

\medskip
\paragraph{Our explicit algorithm for \textsc{Forest CSP}.}
As a preprocessing step, note that we can assume that no downwards-closed constraint binds two variables of the same component $T$.
Indeed, if $u$ and $v$ is in the same component $T$ and a constraint with a downwards-closed relation $R_{u,v}$ is present, then we can iterate
over all $n_T$ assignments to the variables of $T$ and delete those that do not satisfy $R_{u,v}$. 
(Deleting a value from a domain requires some tedious renumbering of the domains, but does not lead us out of the \textsc{Forest CSP} class of instances.)

Similarly, we can assume that for every downwards-closed constraint binding $u$ and $v$ with relation $R_{u,v}$, for every possible value $x$ of $u$,
there is at least one satisfying value of $v$ (and vice versa), as otherwise one can delete $x$ from the domain of $u$ and propagate.

First, guess whether there is a variable $v$ such that setting $v=1$ extends to a satisfying assignment. If yes, guess such $v$ and simplify the instance, deleting the whole component
of $v$ and restricting the domains of other variables accordingly.

Second, guess whether there is an edge $uv \in E(G)$ such that there is a satisfying assignment where the value of $u$ or the value of $v$
is at the endpoint of a segment of the segment reversion $\pi_{uv}$. If this is the case, guess the edge $uv$, guess whether the value of $u$ or $v$ is at the endpoint,
and guess whether it is the left or right endpoint of the segment. Restrict the domains according to the guess: If, say, we have guessed that the value of $u$
is at the right endpoint of a segment of $\pi_{uv}$, restrict the domain of $u$ to only the right endpoints of segments of $\pi_{uv}$ and propagate the restriction 
through the whole component of $u$. 
The crucial observation now is that, due to this step, $\pi_{uv}$ becomes an identity permutation. Thus, we can contract the edge $uv$, reducing the number of variables by one.

In the remaining case, we assume that for every satisfying assignment, no variable is assigned~$1$ and no variable is assigned a value that is an endpoint
of a segment of an incident segment reversion constraint. 
Pick a variable $a$ and look at a satisfying assignment~$\phi$ that minimizes $\phi(a)$. 
Try changing the value of $a$ to $\phi(a)-1$ (which belongs to the domain, as $\phi(a) \neq 1$) and propagate it through the component $T$ containing~$a$. 
Observe that the assumption that no value is at the endpoint of a segment of an incident segment reversion implies that
for every $b \in T$, the value of $b$ changes from~$\phi(b)$ to either $\phi(b)+1$ or $\phi(b)-1$. 

By the minimality of $\phi$, some constraint is not satisfied if we change the value of $a$ to $\phi(a)-1$ and propagate it through $T$. 
This violated constraint has to be a downwards-closed constraint binding $u$ and $v$ with relation $R_{u,v}$ where $|\{u,v\} \cap T| =1$.
Without loss of generality, assume $v \in T$ and $u \notin T$. Furthermore, to violate a downwards-closed constraint, the change to the value of $v$
has to be from $\phi(v)$ to $\phi(v)+1$. 

Let $S$ be the component of $u$.
Define a function $f:[n_S] \to [n_T]$ as $f(x) = \max \{y \in [n_T]~|~(x,y) \in R_{u,v}\}$. 
Note that $\phi(v) = f(\phi(u))$, as with $u$ set to $\phi(u)$, the value $\phi(v)$ for $v$ satisfies $R_{u,v}$
while the value $\phi(v)+1$ violates $R_{u,v}$.
Thus, the value of $v$ (and, by propagation, the values in the entire component $T$) are a function of the value of $u$ (i.e., the value of the component $S$). 
Hence, in some sense, by guessing the violated constraint $R_{u,v}$ we have reduced the number of components of $G$ (i.e., the number of super-variables).

However, adding a constraint ``$\phi(v) = f(\phi(u))$'' leaves us outside of the class of \textsc{Forest CSP}s. 
Luckily, this can be easily but tediously fixed. 
Thanks to the fact that $f$ is a non-increasing function, one can bind $u$ and~$v$ with a segment-reversion constraint that reverses the whole $[n_S]$, 
replace the domains of all nodes of $T$ with $[n_S]$, and re-engineer all constraints binding the nodes of $T$ to the new domain. 

Hence, in the last branch, after guessing the violated downwards-closed constraint $R_{u,v}$, we have reduced the number of components of $G$ by one. 
This finishes the sketch of the fixed-parameter algorithm for \textsc{Forest CSP}.

\medskip
\paragraph{Reduction to \textsc{Forest CSP}.}
Let us now go back to the problematic constraints and sketch how to cast them into
the \textsc{Forest CSP} setting.
Recall Figure~\ref{fig:over:binary}. We discussed that a fixed position of the left vertical
line at $x_1$ fixes a partition of the red points to the right of this line into blocks. 
If we want to keep the striped area between $x_1$, $y_4$, and $y_5$ empty, we can add
two constraints: one between $x_1$ and $y_4$, keeping, for a fixed position of $x_1$, the
line $y_4$ above the red block below it,
and one between $x_1$ and $y_5$, keeping, again for a fixed position of $x_1$, the line $y_5$
below the red block above it.

To simplify these constraints we perform an additional guessing step.
In Figure~\ref{fig:over:binary}, the rightmost red point in every block is the \emph{leader}
of the block and, similarly, the leftmost blue point in every block is the \emph{leader}
of the block. 
Since the position of $x_1$ fixes the red blocks via the guessed alternation, it also fixes the red leaders.
Furthermore, since, for a fixed position of $x_1$, varying $x_2$ only increases or decreases the
blue blocks without merging or splitting them, the position of $x_1$ also determines
the leaders of the blue blocks.
In the branching step, we guess the left-to-right order of the red leaders
and the left-to-right order of the blue leaders, deleting from the domains of $x_1$
and $x_2$ positions, where the order is different than guessed.

To understand what this branching step gives us, let us move to the larger
example in Figure~\ref{fig:over:sliding2}. 
We think of the left vertical line --- denoted in the figure as $p_1$ ---
as sliding continuously right-to-left
from position $x_1'$ to position $x_1$. 
As we slide, the set of red points to the right of the line grows.
To keep the alternation as guessed, the next vertical line (denoted $p_2$) slides as well, 
on the way shrinking the set of blue points between $p_1$ and $p_2$. 

During this slide, a blue block may disappear, like the third or fourth block with yellow background
in Figure~\ref{fig:over:sliding2}. A disappearing blue block merges the two neighboring red blocks
into one. To keep the alternation as guessed, a new red block needs to appear (e.g., 
the second and third block with green background in Figure~\ref{fig:over:sliding2})
at the same time splitting a blue block into two smaller blocks.

The important observation is as follows: Since the order of the leaders is as guessed,
the first blue block to disappear is the block with the rightmost blue leader,
and, in general, the order of disappearance of blue blocks
is exactly the right-to-left order of their leaders.
Naturally, not all blue blocks need to disappear, but only those that have
their leaders between the two considered positions of $p_2$. 

Similarly, if during the move $s$ blue blocks disappear, then the number of red blocks
also decreased by $s$ and exactly $s$ new red blocks need to appear. These newly appearing
blocks will be exactly the $s$ blocks with the leftmost leaders.

Consequently, for every possible scenario as in Figure~\ref{fig:over:sliding2} with lines
$p_1$ and $p_2$ (but different choices of $x_1$ and $x_1'$), the same blocks will start to merge
and appear, only the number of the merged/appearing blocks $s$ can differ. 
In Figure~\ref{fig:over:sliding2}, the third block with brown background
gets first merged into the fourth block with brown background
and then the second block with brown background gets merged into the resulting block, ending up with the fourth block with green background. 
The above description is general: Whenever we start with the line $p_1$, if the alternation
and order of leaders is as we guessed, then sliding $p_1$ to the right will first merge
the third red block into the fourth (and, also merge the sixth into the fifth) 
and then the second into the resulting block.

This merging order allows us to define a rooted auxiliary tree on the red blocks; in this tree, a child block is supposed to merge to the parent block. Figure~\ref{fig:over:sliding2} depicts an exemplary such tree. 

The root $B$ of the auxiliary tree is the block with the rightmost leader; it never gets merged into another block, its leader stays constant, and the block $B$ only absorbs other blocks. Thus, the $y$-coordinate of its top border only increases and the $y$-coordinate of its bottom border only decreases as one moves $p_1$ from right to left.

Consider now a child $B'$ of the root block $B$ and assume $B'$ is above $B$.
As one moves $p_1$ from right to left, 
once in a while $B'$ gets merged into $B$ and a new block $B'$ appears. 
As $B$ only grows, the new appearance of $B'$ is always above the previous one. 
Meanwhile, between the moments when $B'$ is merged into $B$, $B'$ grows, but its leader stays
constant. Hence, between the merges the $y$-coordinate of the bottom border of $B'$ decreases,
only to jump and increase during each merge. 

As one goes deeper in the auxiliary tree, the above behavior can nest. 
Consider for example a child $B''$ of $B'$ that is below $B'$, that is, is between $B'$ and the root $B$. When $B'$ is merged with $B$, $B''$ is merged as well and $B''$ jumps upwards to a new position
together with $B'$. 
However, between the merges of $B$ and $B'$, $B''$ can merge multiple times with $B'$;
every such merge results in $B''$ jumping --- this time downwards --- to a new position.
If one looks at the $y$-coordinate of the top border of $B''$, then:
\begin{itemize}
\item between the merges of $B'$ and $B''$, it increases;
\item during every merge of $B'$ and $B''$, it decreases;
\item during every merge of $B$ and $B'$, it increases.
\end{itemize}
The above reasoning can be made formal into the following: 
If a block is at depth $d$ in the auxiliary tree, then 
the $y$-coordinate of its top or bottom border, as a function of the position of $p_1$,
can be expressed as a composition of a nondecreasing function and $d+\Oh(1)$ segment reversions.
This is the way we cast the leftover difficult constraints into the \textsc{Forest CSP} world.

We remark that, despite substantial effort, we were not able to significantly simplify the arguments used to model our CSP even using the new machinery of the twin-width arguments~\cite{tww4,multicut-3-term}.

\begin{sidewaysfigure*}
  \begin{center}
    \includegraphics{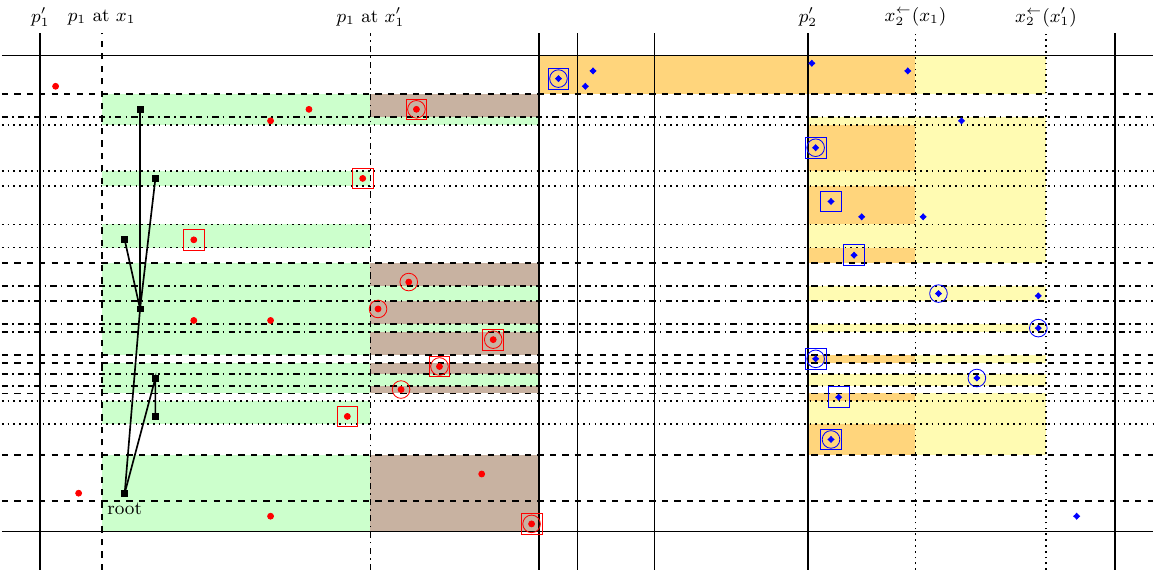}
    \caption{
      More complex example of, where vertical $p_1$ is either at position $x_1$ or at position $x_1'$ for $x_1<x_1'$.
      The horizontal lines for $x_1$ are denoted with dashed and dotted lines.
      The horizontal lines for $x_1'$ are denoted with dashed and dash-dotted lines.
      Blocks given by positioning $p_1$ at $x_1'$ are depicted by brown color with circled leaders and blocks given by positioning $p_1$ at $x_1$ are depicted by the union of green and brown color with squared leaders.
      Let $x_2^\leftarrow(x_1)$ and $x_2^\leftarrow(x_1')$ denote the position of the next vertical line
      $p_2$  that gives the correct alternation if $p_1$ is placed at $x_1$ and $x_1'$, respectively.
      Blocks given by positioning $p_2$ at $x_2^\leftarrow(x_1)$ are depicted by orange color with squared leaders and blocks given by positioning $p_1$ at $x_2^\leftarrow(x_1')$ are depicted by the union of yellow and orange color with circled leaders.
      An auxiliary rooted tree $T$ for red blocks is also visualized.
      There, blocks $B$, $B'$, and $B''$ are marked according to Section~\ref{sec:overview}---Reduction to Forest CSP part.
    }\label{fig:over:sliding2}
  \end{center}
\end{sidewaysfigure*}

\medskip
\paragraph{Final remark.}
We conclude this overview with one remark.
We describe in Section~\ref{sec:redblue} how to reduce \optdis{} to an instance of \textsc{Forest CSP} by a series of color-coding and branching steps.
It may be tempting to backwards-engineer the algorithm for \textsc{Forest CSP}
back to the setting of \optdis{}. However, we think that this is a dead end; in particular, the second branching step, when one contracts an edge, merging two variables,
seems to have no good analog in the \optdis{} setting. 
Furthermore, we think that an important conceptual contribution of this work is the isolation of the \textsc{Forest CSP} problem as an island of tractability
behind the tractability of \optdis{}. 

\section{Segments, segment reversions, and segment representations}\label{sec:segments}
\subsection{Basic definitions and observations}

\begin{definition}
For a finite totally ordered set $(D, \leq)$ and two elements $x, y \in D$, $x \leq y$,
the \emph{segment between $x$ and $y$} is $\segment{D}{x}{y} = \{z \in D~|~x \leq z \leq y\}$.
Elements $x$ and $y$ are the \emph{endpoints} of the segment $\segment{D}{x}{y}$.
\end{definition}
We often write just $[x,y]$ for the segment $\segment{D}{x}{y}$ if the set $(D,\leq)$ is clear
from the context. 

\begin{definition}
Let $(D, \leq)$ be a finite totally ordered set and let $D = \{a(1),a(2),\ldots,a(|D|)\}$ with $a(i) < a(j)$ if and only if $i < j$.

A permutation $\pi : D \to D$ is a \emph{segment reversion} of $D$ 
if there exist integers $1 = i_1 < i_2 < \ldots < i_\ell = |D| + 1$ such that
for every $j \in [\ell]$ and every integer $x$ with $i_j \leq x < i_{j+1}$ we have 
$\pi(a(x)) = a(i_{j+1} - 1 - (x - i_j))$.
In other words, a segment reversion is a permutation that partitions the domain $D$
into segments $[a(i_1), a(i_2-1)], [a(i_2), a(i_3-1)], \ldots, [a(i_\ell), a(i_\ell-1)]$ and reverses
every segment independently. 

A \emph{segment representation of depth $k$} of a permutation $\pi$ of $D$ is a sequence of $k$ segment reversions $\pi_1, \pi_2, \ldots, \pi_k$ of $D$ such that their composition satisfies $\pi = \pi_k \circ \pi_{k - 1} \circ \ldots \circ \pi_1$. 
A permutation $\pi : D \to D$ is of \emph{depth} at most $k$ if $\pi$
admits a segment representation of depth at most $k$.

A \emph{segment representation of depth $k$} of a function $\phi : D \to \mathbb{N}$ is a tuple of $k$ segment reversions $\pi_1, \pi_2, \ldots, \pi_k$ of $D$ and a nondecreasing function $\phi'$ such that their composition satisfies $\phi = \phi' \circ \pi_1 \circ \pi_2 \circ \ldots \circ \pi_k$.
\end{definition}

\begin{definition}
Let $(D,\leq)$ be a finite totally ordered set. A \emph{segment partition} 
is a family $\mathcal{P}$ of segments of $(D,\leq)$ which is a partition of~$D$.
If for two segment partitions $\mathcal{P}_1$ and $\mathcal{P}_2$ we have that
for every $P_1 \in \mathcal{P}_1$ there exists $P_2 \in \mathcal{P}_2$ with $P_1 \subseteq P_2$
then we say that \emph{$\mathcal{P}_1$ is more refined than $\mathcal{P}_2$}
or \emph{$\mathcal{P}_2$ is coarser than $\mathcal{P}_1$}.
The notion of a coarser partition turns the family of all segment partitions into a partially
ordered set with two extremal values, the \emph{most coarse} partition with one segment
and the \emph{most refined} partition with all segments being singletons.
\end{definition}
Note that every segment partition $\mathcal{P}$ induces a segment reversion
 that reverses the segments of $\mathcal{P}$.
 We will denote this segment reversion as $g_\mathcal{P}$.

\begin{definition}
Let $(D_i,\leq_i)$ for $i=1,2$ be two finite totally ordered sets.

A relation $R \subseteq D_1 \times D_2$ is \emph{downwards-closed} if
for every $(a,b) \in R$ and $a' \leq_1 a$, $b' \leq_2 b$ it holds that $(a',b') \in R$.

A relation $R \subseteq D_1 \times D_2$ is of \emph{depth at most $k$} if
there exists a permutation $\pi_1$ of $D_1$ of depth at most~$k_1$,
a permutation $\pi_2$ of $D_2$ of depth at most $k_2$, and a downwards-closed relation $R' \subseteq D_1 \times D_2$
such that $k_1 + k_2 \leq k$ and $(a,b) \in R$ if and only if $(f_1(a),f_2(b)) \in R'$. 
A \emph{segment representation} of $R$ consists of $R'$, a segment representation of $\pi_1$ of depth at most $k_1$
and a segment representation of $\pi_2$ of depth at most $k_2$.
\end{definition}

We make two straightforward observations regarding some relations that are of small depth.

\begin{observation}\label{obs:2sat}
Let $(D_1,\leq_1)$ and $(D_2,\leq_2)$ be two finite totally ordered sets.
For $i=1,2$, let $(a_i^j)_{j=1}^\ell$ be a sequence of elements of $D_i$.
Then a relation $R \subseteq D_1 \times D_2$ defined as $(x_1,x_2) \in R$ if and only if:
\begin{itemize}
\item 
$\bigwedge_{j=1}^\ell (x_1 \leq_1 a_1^j) \vee (x_2 \leq_2 a_2^j)$ is downwards-closed and thus of depth $0$;
\item
$\bigwedge_{j=1}^\ell (x_1 \leq_1 a_1^j) \vee (x_2 \geq_2 a_2^j)$ is of depth $1$, using $k_1=0$
and $k_2=1$ and a segment reversion with one segment reversing the whole $D_2$;
\item
$\bigwedge_{j=1}^\ell (x_1 \geq_1 a_1^j) \vee (x_2 \leq_2 a_2^j)$ is of depth $1$, using $k_1=1$
and $k_2=0$ and a segment reversion with one segment reversing the whole $D_1$;
\item
$\bigwedge_{j=1}^\ell (x_1 \geq_1 a_1^j) \vee (x_2 \geq_2 a_2^j)$ is of depth $2$, using $k_1=1$
and $k_2=1$ and segment reversions each with one segment reversing the whole $D_1$ and the whole $D_2$, respectively.
\end{itemize}
Thus, a conjunction of an arbitrary finite number of the above relations
can be expressed as a conjunction of at most four relations, each of depth at most $2$.
\end{observation}

\begin{observation}\label{obs:ineq}
Let $D_1,D_2 \subseteq D$ for a totally ordered set $(D, \leq)$.
We treat $D_i$ as a totally ordered set with the order inherited from $(D,\leq)$. 
Then a relation $R \subseteq D_1 \times D_2$ defined as $R = \{(x_1,x_2) \in D_1 \times D_2 \mid x_1 < x_2\}$ is of depth at most $1$ and a segment representation of this depth can be computed in polynomial time.\footnote{Throughout, for some relation $\leq$ we use $x < y$ to denote $x \leq y$ and not $x = y$.}  
\end{observation}
\begin{proof}
Let $\pi_2$ be a segment reversion of $D_2$ with one segment, that is, $\pi_2$ reverses the domain~$D_2$.
Observe that $\{(a,\pi_2(b)) \mid a \in D_1 \wedge b \in D_2 \wedge a < b\}$ is a downwards-closed subrelation of $D_1 \times D_2$.
\end{proof}

\subsection{Operating on segment representations}

We will need the following two technical lemmas.
\begin{lemma}\label{lem:seg-swap}
Let $(D_1,\leq_1)$ and $(D_2,\leq_2)$ be two finite totally ordered sets,
    $f: D_1 \to D_2$ be a nondecreasing function\footnote{A function~$f$ on a domain and codomain that are totally ordered by $\leq_1$ and $\leq_2$, respectively, is called \emph{nondecreasing} if for every $x, x'$ in the domain we have that $x \leq_1 x'$ implies $f(x) \leq_2 f(x')$.},
and $g : D_2 \to D_2$ be a segment reversion.
Then there exists a nondecreasing function $f' : D_1 \to D_2$
and a segment reversion $g' : D_1 \to D_1$ such that $g \circ f = f' \circ g'$.
Furthermore, such $f'$ and $g'$ can be computed in polynomial time, given $(D_1,\leq_1)$, $(D_2,\leq_2)$,
  $f$, and $g$.
\end{lemma}
\begin{proof}
Let $(\segment{D_2}{a_i}{b_i})_{i=1}^r$ be the segments of the segment reversion $g$
in increasing order. 
For every $i \in \{1, 2, \ldots, r\}$, let
\begin{align*}
c_i &= \min \{c \in D_1~|~f(c) \geq a_i\},\\
d_i &= \max \{d \in D_1~|~f(d) \leq b_i\}.
\end{align*}
Let $\mathcal{Q}$ be the family of those segments $\segment{D_1}{c_i}{d_i}$ for which 
both $c_i$ and $d_i$ are defined and $c_i \leq_1 d_i$
(which is equivalent to the existence of $x \in D_1$ with $f(x) \in \segment{D_2}{a_i}{b_i}$). 
From the definition of $c_i$s and $d_i$s we obtain that 
$\mathcal{Q}$ is a segment partition of $(D_1,\leq_1)$.
We put $g' = g_{\mathcal{Q}}$
and
$$f' = g \circ f \circ g'.$$
The desired equation $g \circ f = f' \circ g'$ follows directly from the definition of $f'$
and the fact that the segment reversion $g'$ is an involution.\footnote{An \emph{involution} is a function~$\phi$ which is its own inverse, that is, $\phi \circ \phi$ is the identity.}
Clearly, $f'$ and $g'$ are computable in polynomial time.
It remains to check that $f'$ is nondecreasing.

Let $x <_1 y$ be two elements of $D_1$. 
We consider two cases.
In the first case, we assume that $x$ and $y$ belong to the same segment $\segment{D_1}{c_i}{d_i}$
of $\mathcal{Q}$.
Then, $g'(x)$ and $g'(y)$ also lie in $\segment{D_1}{c_i}{d_i}$ and $g'(x) >_1 g'(y)$ by the definition 
of the segment reversion $g' = g_\mathcal{Q}$. 
Since $f$ is nondecreasing, $f(g'(x)) \geq_2 f(g'(y))$.
By the definition of $c_i$ and $d_i$, we have that both $f(g'(x))$ and $f(g'(y))$ lie
in the segment $\segment{D_2}{a_i}{b_i}$. 
Hence, since $\segment{D_2}{a_i}{b_i}$ is a segment of the segment reversion $g$, we have
$g(f(g'(x))) \leq_2 g(f(g'(y)))$, as desired.

In the second case, let $x \in \segment{D_1}{c_i}{d_i}$ and $y \in \segment{D_1}{c_j}{d_j}$
for some $i \neq j$. From the definition of the $c_i$s and $d_i$s we infer
that $x <_1 y$ implies $i < j$.
By the definition of $g' = g_\mathcal{Q}$, we have $g'(x) \in \segment{D_1}{c_i}{d_i}$
and $g'(y) \in \segment{D_1}{c_j}{d_j}$. 
Since $f$ is nondecreasing, $f(g'(x)) \leq_2 f(g'(y))$.
By the definition of the $c_i$s and $d_i$s, we have that
$f(g'(x)) \in \segment{D_2}{a_i}{b_i}$ and
$f(g'(y)) \in \segment{D_2}{a_j}{b_j}$.
Since $\segment{D_2}{a_i}{b_i}$ and $\segment{D_2}{a_j}{b_j}$ are segments
of $g$, we have $g(f(g'(x))) \leq_2 g(f(g'(y)))$, as desired.

This finishes the proof that $f'$ is nondecreasing and concludes the proof of the claim.
\end{proof}

\begin{lemma}\label{lem:dc-comp}
Let $(D_i,\leq_i)$ for $i=1,2,3$ be three finite totally ordered sets,
$f : D_1 \to D_2$ be a nondecreasing function, and $R \subseteq D_2 \times D_3$
be a downwards-closed relation.
Then the relation
$$R' = \{ (x,y) \in D_1 \times D_3~|~(f(x),y) \in R\}$$
is also downwards-closed.
\end{lemma}
\begin{proof}
If $(x,y) \in R'$, $x' \leq_1 x$, and $y' \leq_2 y$, then $f(x') \leq_2 f(x)$ as $f$ is nondecreasing,
$(f(x'),y') \in R$ as $(f(x),y) \in R$ and $R$ is downwards closed, and thus $(x',y') \in R'$ by the definition of $R'$.
\end{proof}

\subsection{Tree of segment partitions}\label{ss:tree}

For a rooted tree $T$, we use the following notation:
\begin{itemize}
\item $\leaves(T)$ is the set of leaves of $T$;
\item $\treeroot(T)$ is the root of $T$;
\item for a non-root node $v$, $\parent(v)$ is the parent of $v$.
\end{itemize}
In this subsection we are interested in the following setting.
A \emph{tree of segment partitions} consists of:
\begin{itemize}
\item a finite totally ordered set $(D,\leq)$;
\item a rooted tree $T$;
\item a segment partition $\mathcal{P}_v$ of $(D,\leq)$ for every $v \in V(T)$ such that:
\begin{itemize}
\item the partition $\mathcal{P}_{\parent(v)}$ is coarser than the partition $\mathcal{P}_v$ for every non-root node $v$;
\item the partition $\mathcal{P}_{\treeroot(v)}$ is the most coarse partition (with one segment);
\item for every leaf $v \in \leaves(T)$ the partition $\mathcal{P}_v$ is the most refined partition (with only singletons);
\end{itemize}
\item an assignment $\flf : V(T) \setminus \{\treeroot(T)\} \to \{\typeinc, \typedec\}$.
\end{itemize}
We say that a non-root node~$w$ is of \emph{increasing type} if $\flf(w) = \typeinc$
and of \emph{decreasing type} if $\flf(w) = \typedec$.

Given a tree of segment partitions $\mathbb{T} = ((D,\leq), T, (\mathcal{P}_v)_{v \in V(T)}, \flf)$, a
\emph{family of leaf functions} is
a family $(f_v)_{v \in \leaves(T)}$ such that
for every $v \in \leaves(T)$ the function $f_v : D \to \mathbb{Z}$ satisfies the following property:
for every non-root element $w$ on the path in $T$ from $v$ to $\treeroot(T)$, 
for every $Q \in \mathcal{P}_{\parent(w)}$, if $Q_1,Q_2,\ldots,Q_a$ are the segments of $\mathcal{P}_w$
contained in $Q$ in increasing order, then 
for every $x_1 \in Q_1$, $x_2 \in Q_2$, \ldots, $x_a \in Q_a$ we have 
\begin{align*}
&f_v(x_1) < f_v(x_2) < \ldots < f_v(x_a)&\mathrm{\ if\ }\flf(w) = \typeinc,\\
&f_v(x_1) > f_v(x_2) > \ldots > f_v(x_a)&\mathrm{\ if\ }\flf(w) = \typedec.
\end{align*}
\begin{lemma}\label{lem:make-seg-rep}
Let $\mathbb{T} = ((D,\leq), T, (\mathcal{P}_v)_{v \in V(T)}, \flf)$ be a tree of segment partitions
and $\mathcal{F} = (f_v)_{v \in \leaves(T)}$ be a family of leaf functions
in $\mathbb{T}$. 
Then there exists a family $\mathcal{G} = (g_v)_{v \in V(T) \setminus \{\treeroot(T)\}}$ of
segment reversions of $D$ and a family $\widehat{\mathcal{F}} = (\hat{f}_v)_{v \in \leaves(T)}$ of stricly increasing functions
with domain $D$ and range $\mathbb{Z}$ such that,
for every $v \in \leaves(T)$,
if $v = v_1, v_2, \ldots, v_b = \treeroot(T)$ are the nodes on the path from $v$ to $\treeroot(T)$ in $T$,
then 
\begin{equation}\label{eq:make-seg-rep}
f_v = \hat{f}_v \circ g_{v_{b-1}} \circ g_{v_{b-2}} \circ \ldots \circ g_{v_{1}}.
\end{equation}
Furthermore, given $\mathbb{T}$ and $\mathcal{F}$, the families
$\mathcal{G}$ and $\widehat{\mathcal{F}}$ can be computed in polynomial time.
\end{lemma}
\begin{proof}
Fix a non-root node $w$.
We say that $w$ is \emph{pivotal} if either
\begin{itemize}
\item $\parent(w) = \treeroot(T)$ and $\flf(w) = \typedec$, or
\item $\parent(w) \neq \treeroot(T)$ and $\flf(w) \neq \flf(\parent(w))$.
\end{itemize}
Let $\mathcal{Q}_w = \mathcal{P}_{\parent(w)}$ if $w$ is pivotal and 
let $\mathcal{Q}_w$ be the most refined partition of $D$ otherwise. 
Let $g_w = g_{\mathcal{Q}_w}$.
That is, $g_w$ is the segment reversion that reverses the segments of $\mathcal{P}_{\parent(w)}$ for
pivotal $w$ and is an identity otherwise.

Fix a leaf $v \in \leaves(T)$ and let $v = v_1, v_2, \ldots, v_b = \treeroot(T)$
be the nodes on the path in $T$ from $v$ to the root $\treeroot(T)$.
Define
$$\hat{f}_v = f_v \circ g_{v_{1}} \circ g_{v_{2}} \circ \ldots \circ g_{v_{b-1}}.$$
Clearly, as a segment reversion is an involution,~\eqref{eq:make-seg-rep} follows.
Hence, to finish the proof of the lemma it suffices to show that $\hat{f}_v$ is strictly increasing. 

Take $x,y \in D$ with $x < y$.
For each $i \in \{1, 2, \ldots, b\}$, let 
\begin{align*}
x_i &= g_{v_{i}} \circ g_{v_{i+1}} \circ \ldots \circ g_{v_{b-1}}(x) \text{, and} \\
y_i &= g_{v_{i}} \circ g_{v_{i+1}} \circ \ldots \circ g_{v_{b-1}}(y) \text{,}
\end{align*}
and let $x_b = x$ and $y_b = y$.
Recall that $\mathcal{P}_{v_b} = \mathcal{P}_{\treeroot(P)}$ is the most coarse partition with only one segment so $x_b,y_b$ lie in the same segment of $\mathcal{P}_{v_b}$.
Let $\ell \leq b$ be the minimum index such that $x_{\ell}$ and $y_{\ell}$ lie in the same segment of $\mathcal{P}_{v_\ell}$.
Note that $\ell > 1$ as $\mathcal{P}_{v_1} = \mathcal{P}_v$ is the most refined partition with 
singletons only.
For each $i \in \{\ell, \ell + 1, \ldots, b\}$, let $Q_i \in \mathcal{P}_{v_i}$ be the segment containing $x_i$ and $y_i$. Observe that, since $\mathcal{Q}_{v_i}$ is a more refined partition than $\mathcal{P}_{v_i}$, for every $i \in \{\ell, \ell + 1, \ldots, b\}$, elements $x_i$ and $y_i$ lie in the same segment of the partition $\mathcal{P}_{v_i}$.

From the definition of being pivotal it follows that the number of indices $j \in \{\ell, \ell + 1, \ldots, b\}$ for which
$v_{j-1}$ is pivotal 
is odd if $\flf(v_{\ell-1}) = \typedec$ 
and even if $\flf(v_{\ell-1}) = \typeinc$.
Recall that $g_{j - 1}$ reverses the segment containing $x_{j - 1}$ and $y_{j - 1}$ if and only if $v_{j - 1}$ is pivotal.
Hence $x_{\ell-1} < y_{\ell-1}$ if $\flf(v_{\ell-1}) = \typeinc$ and
$x_{\ell - 1} > y_{\ell-1}$ if $\flf(v_{\ell-1}) = \typedec$.

Since for every $i \in \{1, 2, \ldots, \ell\}$, we have that $x_i$ and $y_i$ lie in different segments of $\mathcal{P}_{v_i}$,
we have $x_1 < y_1$ if $\flf(v_{\ell-1}) = \typeinc$
 and $x_1 > y_1$ if $\flf(v_{\ell-1}) = \typedec$.
For the same reason, $x_1$ and $y_1$ lie in different segments of 
$\mathcal{P}_{v_{\ell-1}}$.
From the definitions of increasing and decreasing types, we infer that 
if $\flf(v_{\ell-1}) = \typeinc$, then $f_v(x_1) < f_v(y_1)$ as $x_1 < y_1$
and if $\flf(v_{\ell-1}) = \typedec$, then $f_v(x_1) < f_v(y_1)$ as $x_1 > y_1$.
Observe that $\hat{f}_v(x) = f_v(x_1)$ and $\hat{f}_v(y) = f_v(y_1)$.
Thus, in both cases, we obtain that $\hat{f}_v(x) < \hat{f}_v(y)$, as desired.
\end{proof}

\section{Auxiliary CSP}\label{sec:csp}
In this section we will be interested in checking the satisfiability of the following
constraint satisfaction problem (CSP).
\begin{definition}
  An \emph{auxiliary CSP instance} is a tuple $(\mathcal{X}, \mathcal{D}, \mathcal{C})$ consisting of a set $\mathcal{X} = \{x_1,x_2,\ldots,x_k\}$ of $k$ \emph{variables}, a totally ordered finite \emph{domain} $(D_i,\leq_i) \in \mathcal{D}$ for every variable $x_i$, and a set $\mathcal{C}$ of binary \emph{constraints}.
  Each constraint $C \in \mathcal{C}$ is a tuple $(x_{i(C,1)}, x_{i(C,2)}, R_C)$ consisting of two variables $x_{i(C,1)}$ and $x_{i(C,2)}$, and a relation $R_C \subseteq D_{i(C,1)} \times D_{i(C,2)}$ given as a segment representation of some depth.
  We say that constraint $C$ \emph{binds} $x_{i(C,1)}$ and $x_{i(C,2)}$.
  An \emph{assignment} is a function $\phi \colon \mathcal{X} \to \mathcal{D}$ such that for each $x_i \in \mathcal{X}$ we have $\phi(x_i) \in D_i$.
  An assignment $\phi$ is \emph{satisfying} if for each constraint $C = (x_{i(C,1)}, x_{i(C,2)}, R_C) \in \mathcal{C}$ we have $(\phi(x_{i(C, 1)}), \phi(x_{i(C, 2)})) \in R_C$.
\end{definition}

Qualitatively, the main result of this section is the following.
\begin{theorem}\label{thm:csp}
  Checking satisfiability of an auxiliary CSP instance is fixed-parameter tractable when parameterized by the sum of the number of variables, the number of constraints, and the depths of all segment representations of constraints.
\end{theorem}

To prove \cref{thm:csp} we show a more general result stated in \cref{lem:csp} below.
For this and to state precisely the running time bounds of the obtained algorithm, we need a few extra definitions.
For a forest $F$, $\trees(F)$ is the family of trees (connected components) of $F$.
For $y \in V(F)$, $\tree_F(y)$ is the tree of $F$ that contains $y$.
We omit the subscript if it is clear from the context.
\begin{definition}
A \emph{forest-CSP instance} is a tuple consisting of a forest $F$ with its vertex set $V(F)$ being the set of \emph{variables} of the instance,
an ordered finite \emph{domain} $(D_T, \leq_T)$ for every $T \in \trees(F)$ (that is, one domain shared between all vertices of $T$),
for every $e \in E(T)$ and $T \in \trees(F)$ a segment reversion $g_e$ 
that is a segment reversion of~$D_T$,
and a family of \emph{constraints} $\mathcal{C}$.
Each constraint $C \in \mathcal{C}$ is a tuple $(y_1, y_2, R_C)$ where
$y_1,y_2 \in V(F)$ are variables and $R_C \subseteq D_{\tree(y_1)} \times D_{\tree(y_2)}$ is a downwards-closed relation.
We say that $C$ \emph{binds} $y_1$ and~$y_2$.

An \emph{assignment} is a function $\phi \colon V(F) \to \mathcal{D}$ such that for each $y \in V(F)$ we have $\phi(y) \in D_{\tree(y)}$.
An assignment $\phi$ \emph{satisfies} the forest-CSP instance 
if for every edge $yy' \in E(F)$ we have $g_e(\phi(y)) = \phi(y')$ 
and for every constraint $C = (y_1,y_2,R_C)$ we have $(\phi(y_1),\phi(y_2)) \in R_C$.

The \emph{apparent size} of a forest-CSP instance is the sum of the number of variables,
number of trees of $F$, and the number of constraints.
\end{definition}

We will show the following result.

\begin{lemma}\label{lem:csp}
There exists an algorithm that, given a forest-CSP instance $\mathcal{I}$
of apparent size $s$,
in $2^{\Oh(s \log s)} |\mathcal{I}|^{\Oh(1)}$ time
computes a satisfying assignment of $\mathcal{I}$ or correctly concludes that $\mathcal{I}$
is unsatisfiable.
\end{lemma}

To see that \cref{lem:csp} implies Theorem~\ref{thm:csp}, we translate an auxiliary CSP instance~$(\mathcal{X}, \mathcal{D}, \mathcal{C})$ with $k$~variables into an equivalent forest-CSP instance $(F, \mathcal{D}', \mathcal{C}')$.
Start with $\mathcal{D} = \emptyset$, $\mathcal{C}' = \emptyset$, and a forest $F$ consisting of $k$ components $T_1,T_2,\ldots,T_k$ where
$T_i$ is an isolated vertex $x_i \in \mathcal{X}$.
Define the domain \mbox{$(D_{T_i}, \leq_{T_i}) \in \mathcal{D}'$} %
of tree $T_i$ as $(D_{T_i}, \leq_{T_i}) := (D_i, \leq_i) \in \mathcal{D}$.
Recall that for every constraint $C = (x_{i(C,1)}, x_{i(C,2)}, R_C) \in \mathcal{C}$ there is a segment representation, that is, there are $\ell_1, \ell_2 \in \mathbb{N}$, segment reversions $g_1^{1}, g_1^{2}, \ldots, g_1^{\ell_1}$ and $g_2^{1}, g_2^{2}, \ldots, g_2^{\ell_2}$, and a downwards-closed relation $R_C'$ such that
\ifdefined\twocolumnflag%
\begin{align*}
(a_1,a_2) \in R_C \Leftrightarrow 
 (&g_1^{k_1} \circ g_1^{k_1-1} \circ \ldots \circ g_1^{1}(a_1), \\
  &g_2^{k_2} \circ g_2^{k_2-1} \circ \ldots \circ g_2^{1}(a_2)) \in R_C')\text{.}
\end{align*}
\else
$$(a_1,a_2) \in R_C \Leftrightarrow
  (g_1^{k_1} \circ g_1^{k_1-1} \circ \ldots \circ g_1^{1}(a_1), 
  g_2^{k_2} \circ g_2^{k_2-1} \circ \ldots \circ g_2^{1}(a_2)) \in R_C')\text{.}$$
\fi
For each constraint $C \in \mathcal{C}$ as above, proceed as follows:
\begin{enumerate}
\item For both $j=1,2$ attach to $x_{i(C,j)}$ in the tree $T_{i(C,j)}$
a path of length $k_j$ with vertices $x_{i(C,j)} = y_j^0, y_j^1, \ldots, y_j^{k_j}$, wherein $y_j^1, \ldots, y_j^{k_j}$ are new variables,
and label the each edge $y_j^{i-1} y_j^i$ with the segment reversion $g_j^i$.
\item Add a constraint $C' = (y_1^{k_1}, y_2^{k_2}, R_C')$ to $\mathcal{C}'$.
\end{enumerate}
A direct check shows that a natural extension of a satisfying assignment to the input auxiliary CSP instance~$(\mathcal{X}, \mathcal{D}, \mathcal{C})$ satisfies the resulting forest-CSP instance~$(F, \mathcal{D}', \mathcal{C}')$ and, in the other direction, a restriction
to $\{x_1,x_2,\ldots,x_k\}$ of any satisfying assignment to $(F, \mathcal{D}', \mathcal{C}')$ 
is a satisfying assignment to $(\mathcal{X}, \mathcal{D}, \mathcal{C})$.
Furthermore, if the input auxiliary CSP instance has $k$ variables, $c$ constraints, and $p$ is the sum of the depths of all segment representations, then the apparent size
of the resulting forest-CSP instance is $p+2k+c$.
Thus, \cref{thm:csp} follows from \cref{lem:csp}. 

The rest of this section is devoted to the proof of Lemma~\ref{lem:csp}.

\subsection{Fixed-parameter algorithm for forest CSPs}

In what follows, \emph{to solve} a forest-CSP instance means to check its
satisfiability and, in case of a satisfiable instance, produce one satisfying assignment.
The algorithm for Lemma~\ref{lem:csp} is a branching algorithm that at every recursive call
performs a number of preprocessing steps and then branches into a number of subcases.
Every recursive call will be performed in polynomial time and will lead to a number of subcalls that is polynomial
in~$s$. 
Every recursive call will be given a forest-CSP instance $\mathcal{I}$
and will either solve $\mathcal{I}$ directly or produce forest-CSP instances and pass them to recursive subcalls while
ensuring that
(i) the input instance~$\mathcal{I}$ is satisfiable if and only if one of the instances passed to the recursive subcalls
is satisfiable, and (ii) given a satisfying assignment of an instance passed to a recursive subcall,
one can produce a satisfying assignment to $\mathcal{I}$ in polynomial time.
In that case, we say that the recursive call is \emph{correct}.
In every recursive subcall the apparent size $s$ will decrease by at least one, 
bounding the depth of the recursion by $s$.
In that case, we say that the recursive call is \emph{diminishing}.
Observe that these two properties guarantee the correctness of the algorithm and the running time bound
of Lemma~\ref{lem:csp}.

We will often phrase a branching step of a recursive algorithm as \emph{guessing} a property
of a hypothetical satisfying assignment. Formally, at each such step, the algorithm checks
all possibilities iteratively. 

It will be convenient to assume that every domain $(D_T,\leq_T)$ equals $\{1, 2, \ldots, |D_T|\}$ with the order $\leq_T$ inherited from the integers.
(This assumption can be reached by a simple remapping argument and we will maintain it throughout the algorithm.)
Thus, henceforth we always use the integer order $<$ for the domains.

Let us now focus on a single recursive call.
Assume that we are given a forest-CSP instance
$$\mathcal{I} = (F,(D_T)_{T \in \trees(F)}, (g_e)_{e \in E(F)}, \mathcal{C})$$ of size $s$.
For two nodes $y,y' \in V(F)$ in the same tree $T$ of $F$, we denote
$$g_{y \to y'} = g_{e_r} \circ g_{e_{r-1}} \circ \ldots \circ g_{e_1},$$
where $e_1,e_2,\ldots,e_r$ is the unique path from $y$ to $y'$ in $T$.
Thus, if $\phi$ is a satisfying assignment, then $\phi(y') = g_{y \to y'}(\phi(y))$.
(And, moreover, $\phi(y) = g_{y' \to y}(\phi(y'))$ since each segment reversion $g_e$ satisfies $g_e = g_e^{-1}$.)
In other words, a fixed value of one variable in a tree $T$ fixes the values of all variables in that
tree. Thus, there are $|D_T|$ possible assignments of all variables of a tree $T$ and we can
enumerate them in time $\Oh(|T| \cdot |D_T|)$.
We need the following auxiliary operations.

\medskip
\paragraph{Forbidding a value.}
We define the operation of \emph{forbidding value $a \in D_{\tree(y)}$ for variable $y \in V(F)$} as follows. 
Let $T = \tree(y)$.
Intuitively, we would like to delete $a$ from the domain
of $y$ and propagate this deletion to all $y' \in V(T)$ and constraints binding
variables of $T$. 
Formally, we let $D_T' = \{1, 2, \ldots, |D_T|-1\}$. 
For every $y' \in V(T)$, we define $\alpha_{y'} : D_T' \to D_T$
as $\alpha_{y'}(b) = b$ if $b < g_{y \to y'}(a)$ and
$\alpha_{y'}(b) = b+1$ if $b \geq g_{y \to y'}(a)$. 
In every constraint $C = (y_1,y_2,R_C)$ and $j \in \{1,2\}$, if $y_j \in V(T)$, 
   then we replace $R_C$ with $R_C'$ defined as follows,
\ifdefined\twocolumnflag%
\begin{align*}
R_C' &= \{(x_1,x_2) \in D_T' \times D_{\tree(y_2)} \mid (\alpha_{y_1}(x_1),x_2) \in R_C\} \\
    &\qquad\mathrm{if\ }j=1,\\
R_C' &= \{(x_1,x_2) \in D_{\tree(y_1)} \times D_T' \mid (x_1,\alpha_{y_1}(x_2)) \in R_C\} \\
    &\qquad\mathrm{if\ }j=2.
\end{align*}%
\else
\begin{align*}
R_C' &= \{(x_1,x_2) \in D_T' \times D_{\tree(y_2)} \mid (\alpha_{y_1}(x_1),x_2) \in R_C\} &\mathrm{if\ }j=1,\\
R_C' &= \{(x_1,x_2) \in D_{\tree(y_1)} \times D_T' \mid (x_1,\alpha_{y_1}(x_2)) \in R_C\} &\mathrm{if\ }j=2.
\end{align*}%
\fi
(Note that $y_1$ and $y_2$ are not necessarily in different trees.)
Observe that each domain remains of the form $\{0, 1, \ldots, \ell\}$ for some $\ell \in \mathbb{N}$.
It is straightforward to verify that $R_C'$ is downwards-closed as $R_C$ is downwards-closed. 
Furthermore, a direct check shows that:
  \begin{enumerate}
  \item If $\phi$ is a satisfying assignment to the original instance such that $\phi(y) \neq a$,
  then $\phi(y') \neq g_{y \to y'}(a)$ for every $y' \in V(T)$. Moreover, the assignment 
  $\phi'$ defined as $\phi'(y') = \alpha_{y'}^{-1}(\phi(y'))$ for every $y' \in V(T)$
  and $\phi'(y') = \phi(y')$ for every $y' \in V(F) \setminus V(T)$ is a satisfying assignment
  to the resulting instance.
  \item If $\phi'$ is a satisfying assignment to the resulting instance, then
  $\phi$ defined as $\phi(y') = \alpha_{y'}(\phi'(y'))$ for every $y' \in V(T)$, and
  $\phi(y') = \phi'(y')$ for every $y' \in V(F) \setminus V(T)$ is a satisfying assignment
  to the original instance.
\end{enumerate}
\emph{Restricting the domain $D_T$ of a variable $y \in V(T)$ to $A\subseteq D_T$} means forbidding all values
of $D_T \setminus A$ for~$y$.

\bigskip We now describe the steps performed in the recursive call and argue in parallel that the recursive call is correct and diminishing.

\medskip
\paragraph{Preprocessing steps.}
We perform the following preprocessing steps exhaustively. 
\begin{enumerate}
\item If there are either no variables (hence a trivial empty
  satisfying assignment) or a variable with an empty domain (hence an obvious negative answer), solve the instance directly.

  Thus, henceforth we assume $V(F) \neq \emptyset$ and that every domain is nonempty.
  \medskip
  
\item For every constraint $C$ that binds two variables from the same
  tree $T$, we iterate over all $|D_T|$ possible assignments of all variables in $T$ and forbid those that do not satisfy $C$.
  (Recall that fixing the value of one variable of a tree fixes the values of all other variables of that tree.)
  Finally, we delete $C$.
  
  Thus, henceforth we assume that every constraint binds variables from two distinct trees of $F$.
  \medskip
  
\item For every constraint $C$, for both variables $y_j$, $j = 1, 2$, that are bound by~$C$, and for every $a \in D_{\tree(y_j)}$, if there is no $b \in D_{\tree(y_{3-j})}$ such that
  $(a,b)$ satisfies $C$, we forbid $a$ for the variable $y_j$.
  
  Thus, henceforth we assume that for every constraint $C$, every variable it binds, and every
  possible value $a$ of this variable, there is at least one value of the other variable bound
  by $C$ that together with $a$ satisfies $C$.
\end{enumerate}
Clearly, the above preprocessing steps can be performed exhaustively in polynomial time and they do not increase the apparent size of the instance.

\bigskip

We next perform three branching steps.
Ultimately, in each of the subcases we consider we will make a recursive call.
However, the branching steps~1 and~2 both hand one subcase down for treatment in the later branching steps.

For every $T \in \trees(F)$, pick arbitrarily some node $x_T \in V(T)$.
Assume that $\mathcal{I}$ is satisfiable and let $\phi$ be a satisfying assignment that is minimal
in the following sense.
For every $T \in \trees(F)$, we require that either
$\phi(x_T) = 1$
or if we replace the value $\phi(x_T)$ with $\phi(x_T)-1$
and the value $\phi(y)$ with $g_{x_T \to y}(\phi(x_T)-1)$ for every $y \in V(T)$,
we violate some constraint.
Note that if $\mathcal{I}$ is satisfiable then such an assignment exists, because each domain~$D_T$ has the form $\{1, 2, \ldots, |D_T|\}$ and thus $\phi(x_T)-1, g_{x_T \to y}(\phi(x_T)-1) \in D_T$.

\medskip
\paragraph{First branching step.}
We branch into $1+|\trees(T)| \leq s+1$ subcases, guessing whether there exists a tree~$T$ such that the variable $x_T$ satisfies $\phi(x_T) = 1$ and which tree it is precisely.
If we have guessed that no such tree exists, we proceed to the next steps of the algorithm with the assumption that $\phi(x_T) > 1$ for every $T \in \trees(F)$.
The other subcases are labeled by the trees of $F$.
In the subcase for $T \in \trees(F)$, we guess that $\phi(x_T) = 1$.
For every constraint $C = (y_1,y_2,R_C)$ that binds $y_j \in V(T)$ with another variable $y_{3-j} \notin V(T)$, we restrict the domain $D_{\tree(y_{3-j})}$ of $y_{3-j}$ to only values $b$ such that $(g_{x_T \to y_j}(1),b) \in R_C$.
Finally, we delete the tree $T$ and all constraints binding variables of $V(T)$, and invoke a recursive call on the resulting instance.

To see that this step is diminishing, note that, due to the deletion of $T$, the apparent size in the recursive
call is reduced by at least one. 
For correctness, clearly, if $\phi(x_T)=1$, then the resulting instance is satisfiable
and any satisfying assignment to the resulting instance can be extended
to a satisfying assignment of $\mathcal{I}$ by assigning $g_{x_T \to y}(1)$ to $y$
for every~$y \in V(T)$.

\medskip
\paragraph{Second branching step.}
We guess whether there exists an edge $yy' \in E(F)$ such that $\phi(y)$ is an endpoint of a segment of~$g_{yy'}$.
If we have guessed that no such edge $yy'$ exists, we proceed to the next steps of the algorithm.
Otherwise, we guess $yy' \in E(F)$, one endpoint~$y$, and whether $\phi(y)$ is the left or the right endpoint of a segment of $g_{yy'}$, leading to at most $4|E(F)| \leq 4s$ subcases.
(Note that $|E(F)| \leq |V(F)| \leq s$.)
We restrict the domain $D_{\tree(y)}$ of $y$ to only those values~$a$ such that $a$~is the left/right (according to the guess) endpoint of a segment of~$g_{yy'}$.
Observe that now $g_{yy'}$ is an identity, as each of its segment has been reduced to a singleton.
Consequently, we do not change the set of satisfying assignments if we contract the edge $yy'$ in the tree $\tree(y)$ and, for every constraint binding $y$ or $y'$, modify $C$ to bind instead the image of the contraction of the edge~$yy'$.
This decreases $s$ by one and we pass the resulting instance to a recursive subcall.

\medskip
\paragraph{Third branching step.}
We now proceed to the last branching step
with the case where no edge~$yy'$ as in branching step~2 exists.
Recall that also from the first branching step we can assume that $\phi(x_T) > 1$ for every $T \in \trees(F)$.
Pick an arbitrary tree $T \in \trees(F)$. 
Using the minimality of $\phi$, we now guess which constraint $\Gamma = (y_1, y_2, R_\Gamma)$
is violated if we replace $\phi(x_T)$ with $\phi(x_T)-1$
and $\phi(y)$ with $g_{x_T \to y}(\phi(x_T)-1)$ for every $y \in V(T)$.
By symmetry, assume $y_1 \in V(T)$.
Since, due to preprocessing, every constraint binds variables of two distinct trees, $y_2 \notin V(T)$.
Let $S = \tree(y_2)$.
Note that we have at most $s$ subcases in this branching step.

We now aim to show that assigning a value to $y_1$ fixes the value of $y_2$ via constraint~$\Gamma$.
Consequently, we will be able to remove $\Gamma$ and merge the trees $S$ and $T$, resulting in a smaller forest-CSP instance, which we can solve recursively. 

Recall that for every $a \in D_S$
there exists at least one $b \in D_T$ with $(b,a) \in R_\Gamma$, by preprocessing step~3.
Since $R_\Gamma$ is a downwards-closed relation, there exists a nonincreasing
function $f' : D_S \to D_T$ such that
$$R_\Gamma = \{(b,a) \in D_T \times D_S \mid b \leq f'(a)\}.$$
The crucial observation is the following.
\begin{claim}\label{cl:csp}
  Assume that $\phi$ exists and all guesses in the current recursive call have been made correctly. Then, $\phi(y_1) = f'(\phi(y_2))$.
\end{claim}
\begin{proof}
Since we made a correct guess at the second branching step, for every
edge $yy'$ on the path in~$T$ from $x_T$ to $y_1$ (with $y'$ closer than $y$ to $x_T$),
the value $\phi(y) = g_{x_T \to y}(\phi(x_T))$ is not an endpoint of $g_{yy'}$.
Inductively from $x_T$ to $y_1$, we infer that 
for every $y$ on the path from $y_1$ to $x_T$ 
we have that $\phi(y) = g_{x_T \to y}(\phi(x_T))$
and $g_{x_T \to y}(\phi(x_T)-1)$ are two consecutive integers.
In particular, $\phi(y_1)$ and $g_{x_T \to y_1}(\phi(x_T)-1)$ are two consecutive integers.

By choice of $\Gamma$, we have $(g_{x_T \to y_1}(\phi(x_T)-1), \phi(y_2)) \notin R_\Gamma$ but 
$(\phi(y_1), \phi(y_2)) \in R_\Gamma$.
Since $R_\Gamma$ is downwards-closed, 
this is only possible if $g_{x_T \to y_1}(\phi(x_T)-1) = \phi(y_1)+1$
and hence $\phi(y_1) = f'(\phi(y_2))$. 
This concludes the proof of the claim.
\end{proof}

Claim~\ref{cl:csp} implies that by fixing an assignment 
of the tree $S$, we induce an assignment of $T$ via the function $f'$.
We would like to merge the two trees $S$ and $T$ via an edge $y_1y_2$, labeled with $f'$.
However, $f'$ is not a segment reversion, but a nonincreasing function.
Thus, we need to perform some work to get back to a forest-CSP instance representation.
For this, we will leverage Lemma~\ref{lem:seg-swap}.

Let $g^\circ$ be a segment reversion with one segment, reversing the whole $D_S$.
Let $f'' = f' \circ g^\circ$, that is, $f'' : D_S \to D_T$
and $f'' \circ g^\circ = f'$. Observe that since $f'$ is nonincreasing, $f''$ is nondecreasing.

We perform the following operation on $T$ that will result in defining
segment reversions $g_e'$ of $D_T$ for every $e \in E(T)$ 
and nondecreasing functions $f_y: D_S \to D_T$ for every $y \in V(T)$ as follows.
We temporarily root $T$ at $y_1$. 
We initiate $f_{y_1} = f''$.
Then, in a top-to-bottom manner, 
for every edge $yy'$ between a node~$y$ and its parent $y'$ such that $f_{y'}$ is already
defined, we invoke Lemma~\ref{lem:seg-swap}
to $f_{y'}$ and the segment reversion~$g_{yy'}$, obtaining 
a segment reversion $g'_{yy'}$ of $D_S$ and a nondecreasing function $f_y : D_S \to D_T$
such that 
\begin{equation}\label{eq:propagate}
g_{yy'} \circ f_{y'} = f_y \circ g'_{yy'}. 
\end{equation}

We merge the trees $S$ and $T$ into one tree $T'$ by adding an edge $y_1y_2$
and define $g_{y_1y_2}' = g^\circ$. We set $D_{T'} = D_S$; observe
that all $g'_e$ for $e \in E(T)$ as well as $g_{y_1y_2}'$ are segment reversions of $D_S$.
Let $F'$ be the resulting forest.
For every $e \in E(F) \setminus E(T)$, we define $g'_e = g_e$.
Similarly as we defined $g_{y \to y'}$, we define $g'_{y \to y'}$ for every two
vertices $y,y'$ of the same tree of $F'$
as $g'_{e_r} \circ g'_{e_{r-1}} \circ \ldots \circ g'_{e_1}$ where
$e_1,e_2,\ldots,e_r$ are the edges on the path from $y$ to $y'$ in $F'$. 
Note that $g'_{y \to y'} = g_{y \to y'}$ when $y,y' \notin V(T')$
or $y,y' \in V(S)$. 

We now define a modified set of constraints $\mathcal{C}'$ as follows.
Every constraint $C \in \mathcal{C}$ that does not bind any variable of $T$
we insert into $\mathcal{C}'$ without modifications.
For every constraint $C \in \mathcal{C}$ that binds a variable of $T$, we proceed as follows.
By symmetry, assume that $C = (z_1,z_2,R_C)$ with $z_1 \in V(T)$ and $z_2 \notin V(T)$. 
Recall that $R_C \subseteq D_T \times D_{\tree(z_2)}$ and $f_{z_1} : D_S \to D_T$.
We apply Lemma~\ref{lem:dc-comp} to $R_C$ and $f_{z_1}$, obtaining
a downwards-closed relation $R_C' \subseteq D_S \times D_{\tree(z_2)}$ such that
$$(a,b) \in R_C' \Leftrightarrow (f_{z_1}(a), b) \in R_C.$$
We insert $C' := (z_1,z_2,R_C')$ into $\mathcal{C}'$. 

Let $\mathcal{I}' = (F',(D_T)_{T \in \trees(F)}, (g_e')_{e \in E(F')}, \mathcal{C}')$
be the resulting forest-CSP instance. 
Note that $|V(F')| \leq |V(F)|$, $|\mathcal{C}'| \leq |\mathcal{C}|$, while
$|\trees(F')| < |\trees(F)|$. Thus, the apparent size of $\mathcal{I}'$ is smaller
than the apparent size of $\mathcal{I}$.
We pass $\mathcal{I}'$ to a recursive subcall.

To complete the proof of Lemma~\ref{lem:csp}, it remains to show correctness of branching step~3.
This is done in the next two claims.

\begin{claim}\label{cl:csp1}
Let $\zeta'$ be a satisfying assignment to $\mathcal{I}'$.
Define an assignment $\zeta$ to $\mathcal{I}$ as follows.
For every $y \in V(F) \setminus V(T)$, set $\zeta(y) = \zeta'(y)$.
For every $y \in V(T)$, set $\zeta(y) = f_y(\zeta'(y))$.
Then $\zeta$ is a satisfying assignment to $\mathcal{I}$.
\end{claim}
\begin{proof}
To see that $\zeta$ is an assignment, that is, maps each variable into its domain, since every function $f_y$ for $y\in V(T)$ has domain $D_S = D_{T'}$ and codomain $D_T$,
every $y \in V(T)$ satisfies $\zeta(y) \in D_T$. 

To see that $\zeta$ is a satisfying assignment, consider first the condition on the forest edges.
Pick $e = yy' \in E(F)$. If $e \notin E(T)$, then $\zeta(y) = \zeta'(y)$, $\zeta(y') = \zeta'(y')$,
     $g_e = g'_e$,
and obviously $\zeta(y') = g'_e(\zeta(y))$.
Otherwise, assume without loss of generality that $y'$ is closer than $y$ to $y_1$ in $T$.
Then~\eqref{eq:propagate} ensures that
\ifdefined\twocolumnflag%
\begin{align*}
g_{yy'} (\zeta(y')) &=  g_{yy'}(f_{y'}(\zeta'(y'))) = f_y(g'_{yy'}(\zeta'(y')) \\
                    &= f_y(\zeta'(y)) = \zeta(y)
\end{align*}
\else
$$g_{yy'} (\zeta(y')) =  g_{yy'}(f_{y'}(\zeta'(y'))) = f_y(g'_{yy'}(\zeta'(y')) = f_y(\zeta'(y)) = \zeta(y)$$
\fi
as desired.

Now pick a constraint $C \in \mathcal{C}$ and let us show that $\zeta$ satisfies~$C$.
If $C$ does not bind a variable of $T$, then $C \in \mathcal{C}'$ and 
$\zeta$ and $\zeta'$ agree on the variables bound by $C$, hence $\zeta$ satisfies $C$.
Otherwise, without loss of generality, $C = (z_1,z_2,R_C)$ with $z_1 \in V(T)$
and there is the corresponding constraint $C' = (z_1,z_2,R_C')$ in $\mathcal{C}'$ as defined above.
Since $\zeta'$ satisfies $C'$, we have
$(\zeta'(z_1),\zeta'(z_2)) \in R_C'$. 
By the definition of $R_C'$, this is equivalent to
$(f_{z_1}(\zeta'(z_1)),\zeta'(z_1)) \in R_C$. 
Since $\zeta(z_1) = f_{z_1}(\zeta'(z_1))$ (as $z_1 \in V(T)$) and $\zeta(z_2) = \zeta'(z_2)$, this 
is equivalent to $(\zeta(z_1), \zeta(z_2)) \in R_C$. 
Hence, $\zeta$ satisfies the constraint $C$. This finishes the proof of the claim.
\end{proof}

\begin{claim}\label{cl:csp2}
Let $\zeta$ be a satisfying assignment to $\mathcal{I}$
that additionally satisfies $\zeta(y_1) = f'(\zeta(y_2))$. 
Define an assignment $\zeta'$ to $\mathcal{I}'$ as follows.
For every $y \in V(F) \setminus V(T)$, set $\zeta'(y) = \zeta(y)$.
For every $y \in V(T)$, set $\zeta'(y) = g'_{y_2 \to y}(\zeta(y_2))$.
Then $\zeta'$ is a satisfying assignment to $\mathcal{I}'$.
\end{claim}
\begin{proof}
To see that $\zeta'$ is indeed an assignment, it is immediate from the definition of $\mathcal{I}'$ that 
for every tree~$A$ of $F'$ and $y \in V(A)$ we have $\zeta'(y) \in D_A$.
To see that $\zeta'$ is a satisfying assignment, by definition, for every $e = yy' \in E(F')$ we have $\zeta'(y') = g'_e(\zeta'(y))$.
Also, obviously $\zeta'$ satisfies all constraints of $\mathcal{C}'$ that come
unmodified from a constraint of $\mathcal{C}$ that does not bind a variable of $V(T)$.
It remains to show that the remaining constraints are satisfied.

Consider a constraint $C' = (z_1,z_2,R_C') \in \mathcal{C}'$ that comes
from a constraint $C = (z_1,z_2,R_C) \in \mathcal{C}$ binding a variable of $V(T)$.
Without loss of generality, $z_1 \in V(T)$ and $z_2 \notin V(T)$.
By composing~\eqref{eq:propagate} over all edges on the path from $z_1$ to $y_1$ in $T$
we obtain that
$$g_{y_1 \to z_1} \circ f'' = f_{z_1} \circ g'_{y_1 \to z_1}.$$
By composing the above with $g^\circ$ on the right and using $f'' = f' \circ g^\circ$ (hence $f'' \circ g^\circ = f'$) and $g^\circ = g'_{y_1y_2}$, we obtain that
\begin{equation}\label{eq:csp2}
g_{y_1 \to z_1} \circ f' = f_{z_1} \circ g'_{y_2 \to z_1}.
\end{equation}
By the definition of $R_C'$, we have that 
$(\zeta'(z_1), \zeta'(z_2)) \in R_C'$ is equivalent to
$$(f_{z_1}(\zeta'(z_1)), \zeta'(z_2)) \in R_C.$$
By the definition of $\zeta'$, this is equivalent to
$$(f_{z_1} \circ g'_{y_2 \to z_1}(\zeta(y_2)), \zeta(z_2)) \in R_C.$$
By~\eqref{eq:csp2}, this is equivalent to 
$$(g_{y_1 \to z_1} \circ f'(\zeta(y_2)), \zeta(z_2)) \in R_C.$$
Since $f'(\zeta(y_2) = \zeta(y_1)$, this is equivalent to
$$(g_{y_1 \to z_1} (\zeta(y_1)), \zeta(z_2)) \in R_C.$$
By the definition of $g_{y_1 \to z_1}$, this is in turn equivalent to
$$(\zeta(z_1), \zeta(z_2)) \in R_C,$$
which follows as $\zeta$ satisfies $C$. This finishes the proof of the claim.
\end{proof}

Claims~\ref{cl:csp1} and~\ref{cl:csp2} show the correctness of the third branching step,
concluding the proof of Lemma~\ref{lem:csp} and of Theorem~\ref{thm:csp}.

\section{From Optimal Discretization to the auxiliary CSP}\label{sec:redblue}
To prove \cref{thm:main} we give an algorithm that constructs a branching tree.
At each branch, the algorithm tries a limited number of options for some property of the solution.
At the leaves it will then assume that the chosen options are correct and reduce the resulting restricted instance of \optdis{} to the auxiliary CSP from \cref{sec:csp}.
We first give basic notation for the building blocks of the solution in \cref{sec:cells}.
The branching tree is described in \cref{sec:init-branch-steps}.
The reduction to the auxiliary CSP is given in \cref{sec:csp-formulation,sec:simple-steps,sec:alternations,sec:consist-DE,sec:alternating-lines}.
Throughout the description of the algorithm, we directly argue that it satisfies the running time bound and that it is sound, meaning that, if there is a solution, then a solution will be found in some branch of the branching tree.
We argue in the end, in \cref{sec:completeness,sec:wrap-up}, that the algorithm is complete, that is, if it does not return that the input is a no-instance, then the returned object is a solution.

\subsection{Approximate solution and cells}\label[section]{sec:cells}

Let $(W_1,W_2,k)$ be an input to the decision version of \optdis{}.
We assume that $W_1 \cap W_2 = \emptyset$, as otherwise there is no solution.

Using a known factor-$2$ approximation algorithm~\cite{CalinescuDKW05}, we compute in polynomial time a separation
$(\Xapx, \Yapx)$. If $|\Xapx| + |\Yapx| > 2k$, we report that the input instance is a no-instance.
Otherwise, we proceed further as follows.

\paragraph{Discretization.}
Let $n = |W_1| + |W_2|$.
By simple discretization and rescaling, we can assume that 
\begin{itemize}
\item every point in $W_1 \cup W_2$ has both coordinates
being positive integers from $[3n]$ and divisible by $3$, 
\item the sought solution $(X,Y)$ consists of integers from $[3n]$ that are equal to $2$ modulo $3$. 
\item every element of $\Xapx \cup \Yapx$ is an integer from $[3n]$ that is equal to $1$ modulo $3$.
\end{itemize}
Furthermore, we add $1$ and $3n+1$ to both $\Xapx$ and $\Yapx$ (if not already present). Thus, $|\Xapx| + |\Yapx| \leq 2k+4$,
$\Xapx, \Yapx \subseteq \{3i+1 \mid i \in \{0,1,\ldots,n\}\}$ and for every $(x,y) \in W_1 \cup W_2$
we have that $x$ is between the minimum and maximum element of $\Xapx$ and $y$ is between the minimum and maximum element of $\Yapx$.
We henceforth refer to the properties obtained in this paragraph as the \emph{discretization properties}.

\paragraph{Total orders \boldmath $\leqx, \leqy$.}
We will use two total orders on points of $W_1 \cup W_2$:
\begin{itemize}
\item $(x,y) \leqx (x',y')$ if $x < x'$ or both $x = x'$ and $y \leq y'$;
\item $(x,y) \leqy (x',y')$ if $y < y'$ or both $y = y'$ and $x \leq x'$.
\end{itemize}
For a set $W \subseteq W_1 \cup W_2$, the \emph{topmost} point is the $\leqy$-maximum one,
    the \emph{bottommost} is the $\leqy$-minimum,
    the \emph{leftmost} is the $\leqx$-minimum, and
    the \emph{rightmost} is the $\leqx$-maximum one.
    Finally, an \emph{extremal} point in $W$ is the topmost, bottommost, leftmost, or the rightmost point in $W$;
    there are at most four extremal points in a set $W$.

Assume that the input instance is a yes-instance and let $(X,Y)$ be a sought solution: 
a separation for $(W_1,W_2)$ with $|X|+|Y| \leq k$ and $X,Y \subseteq \{3i-1 \mid i \in [n]\}$. 

\paragraph{Cells.}
For two consecutive elements $x_1,x_2$ of $\Xapx \cup X$
and two consecutive elements $y_1,y_2$ of $\Yapx \cup Y$, 
define the set $\cell(x_1,y_1) := \{x_1+1, x_1 + 2, \ldots,x_2-1\} \times \{y_1+1, y_1 + 2, \ldots,y_2-1\}$.
Each such set is called a \emph{cell}.
Note that since we require $x_1,x_2$ to be consecutive elements of $\Xapx \cup X$ and similarly
$y_1,y_2$ to be consecutive elements of $\Yapx \cup Y$, the pair $(x_1,y_1)$ determines the corresponding cell uniquely.
The \emph{points in the cell $\cell(x_1,y_1)$} are the points in the set $\cell(x_1,y_1) \cap (W_1 \cup W_2)$.

Similarly, for two consecutive elements $x_1,x_2$ of $\Xapx$
and two consecutive elements $y_1,y_2$ of $\Yapx$, 
an \emph{\apx-supercell} is the set $\apxcell(x_1,y_1) := \{x_1+1, x_1 + 2,\ldots,x_2-1\} \times \{y_1+1, y_1 + 2, \ldots,y_2-1\}$
and the points in this cell are $\apxcell(x_1,y_1) \cap (W_1 \cup W_2)$. 
Also, for two consecutive elements $x_1,x_2$ of $X \cup \{1,3n+1\}$
and two consecutive elements $y_1,y_2$ of $Y \cup \{1,3n+1\}$, 
an \emph{\opt-supercell} is the set $\optcell(x_1,y_1) := \{x_1+1, x_1 + 2, \ldots,x_2-1\} \times \{y_1+1, y_1 + 2,\ldots,y_2-1\}$
and the points in this cell are $\optcell(x_1,y_1) \cap (W_1 \cup W_2)$. 

Clearly, every \apx-supercell or \opt-supercell contains a number of cells and each cell is contained in exactly one
\apx-supercell and exactly one \opt-supercell.
Note that, since $(\Xapx,\Yapx)$ and $(X,Y)$ are separations, all points in one cell, in one \apx-supercell, and in one \opt-supercell
are either from $W_1$ or from $W_2$, or the (super)cell contains no points.

Furthermore, observe that there are $\Oh(k^2)$ cells, \apx-supercells, and \opt-supercells.

We will also need the following general notation. For 
two elements $x_1, x_2 \in \Xapx \cup X$ with $x_1 < x_2$
and two elements $y_1, y_2 \in \Yapx \cup Y$ with $y_1 < y_2$
by $\area(x_1,x_2,y_1,y_2)$ we denote the union of all cells $\cell(x, y)$
that are between $x_1$ and $x_2$ and between $y_1$ and $y_2$, that is,
that satisfy $x_1 \leq x < x_2$ and $y_1 \leq y < y_2$.

\subsection{Branching steps}\label[section]{sec:init-branch-steps}

\begin{figure}[tb]
\begin{center}
\includegraphics{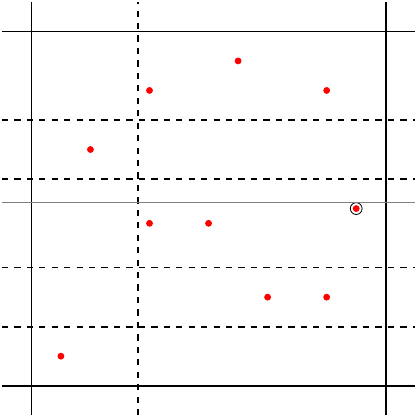}
\caption{Branching Step A. Adding to $\Xapxlines$ (solid) an extra horizontal line (gray) just above the rightmost red point (circled)
  separates some horizontal lines from the solution (dashed) that were not separated before.}\label{fig:branchA}
\end{center}
\end{figure}

In the algorithm we first perform a number of branching steps.
Every step is described in the ``intuitive'' language of \emph{guessing} a property of the solution.
Formally, at every step we are interested in some property of the solution with some (bounded as a function of $k$) number of options and we consider all possible options iteratively.
While considering one of these options, we are interested in finding some solution to the input instance in case $(X,Y)$ satisfies the considered option.
If the solution satisfies the currently considered option, we also say that the corresponding guess is \emph{correct}.

\paragraph{Branching step A: separating elements of the solution.}
We guess whether there exists an \apx-supercell and an extremal point $(x,y)$ in this cell
such that (see Figure~\ref{fig:branchA})
\begin{enumerate}
\item for some $x_1, x_2 \in X$, $x$ is between $x_1$ and $x_2$ while no element of $\Xapx$
is between $x_1$ and $x_2$, or
\item for some $y_1, y_2 \in Y$, $y$ is between $y_1$ and $y_2$ while no element of $\Yapx$
is between $y_1$ and $y_2$.
\end{enumerate}
If we have guessed that this is the case, then we guess $(x, y)$ and, in the first case, we add $x+1$ to $\Xapx$, and in the second case we add $y+1$ to $\Yapx$, and recursively invoke the same branching step.
If we have guessed that no such \apx-cell and an extremal point exist, then we proceed to the next steps of the algorithm.

As $|X|+|Y| \leq k$, the above branching step can be correctly guessed and executed at most $k-1$ times. 
Hence, we limit the depth of the branching tree by $k-1$: at a recursive call at depth $k-1$ we only consider the case where
no such extremal point $(x,y)$ exists.

At every step of the branching process, there are $\Oh(k^2)$ \apx-supercells to choose,
at most four extremal points in every cell, and two options whether the $x$-coordinate of the extremal point separates two elements of $X$
or the $y$-coordinate of the extremal point separates two elements of $Y$.
Thus, the whole branching process
in this step generates $2^{\Oh(k \log k)}$ cases to consider in the remainder of the algorithm, 
where we can assume that no such extremal point $(x,y)$ in any \apx-supercell exists.
Note that the branching does not violate the discretization properties of the elements of $W_1$, $W_2$, $X$, $Y$, $\Xapx$,
and $\Yapx$ and keeps $|\Xapx| + |\Yapx| \leq (2k+4) + (k-1) = 3k+3$.

\paragraph{Branching step B: layout of the solution with regard to the approximate one.}
For every two consecutive elements $x_1,x_2 \in \Xapx$, we guess the number of elements $x \in X$ that are between $x_1$ and $x_2$,
and similarly for every two consecutive elements $y_1,y_2 \in \Yapx$, we guess the number of elements $y \in Y$ that are between $y_1$ and $y_2$.
Recall that $\Xapx \cap X = \emptyset$, $\Yapx \cap Y = \emptyset$, and that $1,3n+1 \in \Xapx \cap \Yapx$, so every element of $X$ and $Y$ is
between two consecutive elements of $\Xapx$ or $\Yapx$, respectively. 
Furthermore, since $|X| + |Y| \leq k$ and $|\Xapx| + |\Yapx| \leq 3k+3$, the above branching leads to $2^{\Oh(k)}$ subcases.

\paragraph{The notions of abstract lines, cells, and their corresponding mappings~\boldmath$\zeta$.} Observe that if we have guessed correctly in Branching Step~B, we know $|X \cup \Xapx|$ and, if we order $X \cup \Xapx$ in the increasing order,
we know which elements of $X \cup \Xapx$ belong to $X$ and which to $\Xapx$; a similar claim holds for $Y \cup \Yapx$. 
We ``only'' do not know the exact values of the elements of $X$ and $Y$, but we have a rough picture of the 
layout of the cells. 
We abstract this information as follows.

We create a totally ordered set $(\Xlines, <)$ of $|X \cup \Xapx|$ elements which we will later refer to as 
\emph{vertical lines}.
Let $\Xsol_X : \Xlines \to \Xapx \cup X$ be a bijection that respects the orders on $\Xlines$ and $\Xapx \cup X \subseteq \mathbb{N}$.\footnote{Respecting the orders means that for each $x, y \in \Xlines$ we have that, if $x \leq y$, then $\Xsol_X(x) \leq \Xsol_X(y)$.}
Let $\Xapxlines = (\Xsol_X)^{-1}(\Xapx)$ be the lines corresponding to the elements of $\Xapx$
and let $\Xoptlines = \Xlines \setminus \Xapxlines$. 
Denote $\Xapxsol = \Xsol_X|_{\Xapxlines}$ and $\Xoptsol_X = \Xsol_X|_{\Xoptlines}$.\footnote{Let $f \colon A \to B$ and $C \subseteq A$. Then $f|_{C}$ is the function resulting from $f$ when removing $A \setminus C$ from the domain of~$f$.}   
Similarly, we define a totally ordered set $(\Ylines, <)$ of $|Y \cup \Yapx|$ \emph{horizontal lines},
  sets $\Yapxlines, \Yoptlines \subseteq \Ylines$ and functions
  $\Ysol_Y$, $\Yapxsol$, and $\Yoptsol_Y$.
Finally, we define $\apxsol = \Xapxsol \cup \Yapxsol$.

Observe that while $\Xsol_X$, $\Xoptsol_X$, $\Ysol_Y$, and $\Yoptsol_Y$ depend on the (unknown to the algorithm) solution
$(X,Y)$, the sets $\Xapxlines$, $\Yapxlines$, $\Xoptlines$, $\Yoptlines$, and functions $\Xapxsol$, $\Yapxsol$, and $\apxsol$
do not depend on $(X,Y)$ and can be computed by the algorithm. This is why we avoid the subscript $X$ or~$Y$ in $\Xapxsol$, $\Yapxsol$, and $\apxsol$.

Our goal can be stated as follows: we want to extend $\Xapxsol$ and $\Yapxsol$ to increasing functions 
$\Xsol : \Xlines \to \mathbb{N}$ and $\Ysol:\Ylines \to \mathbb{N}$ such that
$\{\Xsol(\ell) \mid \ell \in \Xoptlines\}$
and $\{\Ysol(\ell) \mid \ell \in \Yoptlines\}$ is a separation. 

Recall that the notions of cells, \apx-supercells, and \opt-supercells, as well
as the notion $\area()$, have been defined with regard to the
solution $(X,Y)$, but we can also define them with regard to lines $\Xlines$ and $\Ylines$.
That is, for a cell $\cell(x_1,y_1)$, its corresponding 
\emph{abstract cell} is $\cell((\Xsol_X)^{-1}(x_1), (\Ysol_Y)^{-1}(y_1))$.
Let $\Xlines^-$ be the set $\Xlines$ without the maximum element and $\Ylines^-$ be the set $\Ylines$ without the maximum element.
Then we denote the set of abstract cells by $\Cells = \{\cell(\ell_x,\ell_y) \mid \ell_x \in \Xlines^- \wedge \ell_y \in \Ylines^-\}$.
Let $\cell(\ell_x,\ell_y) \in \Cells$ where $\ell_x'$ is the successor of $\ell_x$ in $(\Xlines,<)$
and $\ell_y'$ is the successor of $\ell_y$ in $(\Ylines,<)$. Then we say that $\ell_x$ is the \emph{left} side, $\ell_y$ is the \emph{bottom} side,
$\ell_x'$ is the \emph{right} side, and $\ell_y'$ is the \emph{top} side of~$\cell(\ell_x,\ell_y)$. 

Similarly we define abstract \apx-supercells and abstract \opt-supercells, and
the notion $\area(p_1, p_2, \ell_1, \ell_2)$ for $p_1, p_2 \in \Xlines$, $p_1 < p_2$, $\ell_1, \ell_2 \in \Ylines$,
  $\ell_1 < \ell_2$.
If it does not cause confusion, in what follows we implicitly identify the abstract cell $\cell(\ell_x,\ell_y)$ with its corresponding
cell $\cell(\Xsol_X(\ell_x), \Ysol_Y(\ell_y))$ and similarly for \apx-supercells and \opt-supercells.
Note that for \apx-supercells the distinction between \apx-supercells and abstract \apx-supercells is only in notation
as the functions $\Xapxsol$ and $\Yapxsol$ are known to the algorithm.

\paragraph{Branching step C: contents of the cells and associated mapping~\boldmath$\cont$.}
For every abstract cell $\cell(\ell_x,\ell_y)$, we guess whether the cell $\cell(\Xsol_X(\ell_x), \Ysol_Y(\ell_y))$ contains at least one point of $W_1 \cup W_2$. Since there are $\Oh(k^2)$ cells and two options for each cell, this leads to $2^{\Oh(k^2)}$ subcases. 
Note that if $\cell(\Xsol_X(\ell_x), \Ysol_Y(\ell_y))$ is guessed to contain some points of $W_1 \cup W_2$, we know whether these points are from $W_1$ or from $W_2$: They are from the same set as the points contained in the \apx-supercell containing $\cell(\ell_x,\ell_y)$.
(If the corresponding \apx-supercell does not contain any points of $W_1 \cup W_2$, we discard the cases when $\cell(\Xsol_X(\ell_x),\Ysol_Y(\ell_y))$ is guessed to contain points of $W_1 \cup W_2$.)
Thus, in fact every cell $\cell(\ell_x,\ell_y)$ can be of one of three types: either containing some points of $W_1$ (\emph{type 1}), containing some points of $W_2$ (\emph{type 2}), or not containing any points of $W_1 \cup W_2$ at all (\emph{type 0}).
Let $\cont : \Cells \to \{0,1,2\}$ be the guessed function assigning to every cell its type.

Upon this step, we discard a guess if there are two cells $\cell(\ell_x,\ell_y)$ and $\cell(\ell_x',\ell_y')$ such that we have guessed one to contain some points of $W_1$ and the other to contain some points of $W_2$ that are contained in the same \opt-supercell, as such a situation would contradict the fact that $(X,Y)$ is a separation.
Consequently, we can extend the function $\cont$ to the set of \opt-supercells, indicating for every \opt-supercell whether at least one cell contains a point of $W_1$, a point of $W_2$, or whether the entire \opt-supercell is empty.

For notational convenience, we also extend the function $\cont$ to \apx-supercells in the natural manner.
Here, we also discard the current guess if there is an \apx-supercell that contains some points of $W_1 \cup W_2$, but all abstract cells inside this \apx-supercell are of type $0$.

\paragraph{Branching step D: cells of the extremal points and associated mapping~\boldmath$\phi$.}
We would like now to guess a function $\lead : W_1 \cup W_2 \to \Cells$ that, for every point $(x,y) \in W_1 \cup W_2$ that is extremal in its cell, assigns to $(x,y)$ the abstract cell $\cell(\ell_x,\ell_y)$ such that $\cell(\Xsol_X(\ell_x), \Ysol_Y(\ell_y))$ contains $(x,y)$.
(And we have no requirement on $\lead$ for points that are not extremal in their cell.)

Consider first a random procedure that for every $(x,y) \in W_1 \cup W_2$ samples $\lead(x,y) \in \Cells$ uniformly at random.
Since there are $\Oh(k^2)$ cells and at most four extremal points in one cell, the success probability of this procedure is
$2^{-\Oh(k^2 \log k)}$. 

This random process can be derandomized in a standard manner using the notion of 
\emph{splitters}~\cite{AlonYZ95} (see e.g.\ Cygan et al.~\cite{CyganFKLMPPS15} for an exposition).
For integers $n$, $a$, and $b$, a $(n,a,b)$-splitter is a family $\mathcal{F}$ of functions from $[n]$ to $[b]$
such that for every $A \subseteq [n]$ of size at most $a$ there exists $f \in \mathcal{F}$ that is injective on $A$.
Given integers~$n$ and $r$, one can construct in time polynomial in $n$ and $r$ an $(n,r,r^2)$-splitter of size
$r^{\Oh(1)} \log n$~\cite{AlonYZ95}.
We set $n = |W_1 \cup W_2|$ and $r = 4|\Cells| = \Oh(k^2)$ and construct an $(n,r,r^2)$-splitter $\mathcal{F}_1$
where we treat every function $f_1 \in \mathcal{F}_1$ as a function with domain $W_1 \cup W_2$.
We construct a set $\mathcal{F}_2$ of functions from $r^2$ to $\Cells$ as follows: for every set
$A \subseteq [r^2]$ of size at most $r = 4|\Cells|$ and every function $f_2'$ from $A$
to $\Cells$, we extend $f_2'$ to a function $f_2 : [r^2] \to \Cells$ arbitrarily (e.g., by assigning to every element of
    $[r^2] \setminus A$ one fixed element of $\Cells$) and insert $f_2$ into $\mathcal{F}_2$.
Finally, we define $\mathcal{F} = \{f_2 \circ f_1 \mid (f_1,f_2) \in \mathcal{F}_1 \times \mathcal{F}_2\}$.

Note that $|\mathcal{F}| = 2^{\Oh(k^2 \log k)} \log n$ as $\mathcal{F}_1$ is of size $k^{\Oh(1)} \log n$
while $\mathcal{F}_2$ is of size $2^{\Oh(k^2 \log k)}$ as there are $2^{\Oh(k^2 \log k)}$ choices of the set $A$
and $2^{\Oh(k^2 \log k)}$ choices for the function $f_2'$ from $A$ to $\Cells$. 

We claim that there exists a desired element $\lead \in \mathcal{F}$ as defined above.
By the definition of a splitter and our choice of $r$, there exists $f_1 \in \mathcal{F}_1$ that is injective
on the extremal points. When defining $\mathcal{F}_1$, the algorithm considers at some point the image of the extremal points under $f_1$ as the set~$A$
and hence constructs a function $f_2'$ that, for every extremal point $(x,y)$, assigns to $f_1(x,y)$ the cell that contains~$(x,y)$.
Consequently, $\lead := f_2 \circ f_1$ and hence $\lead$ belongs to $\mathcal{F}$ and satisfies the desired properties.

Our algorithm constructs the family $\mathcal{F}$ as above and tries every $\lead \in \mathcal{F}$ separately.
As discussed, this leads to $2^{\Oh(k^2 \log k)} \log n$ subcases.

Recall that we want to  extend $\Xapxsol$ and $\Yapxsol$ to increasing functions 
$\Xsol : \Xlines \to \mathbb{N}$ and $\Ysol:\Ylines \to \mathbb{N}$ such that
$\{\Xsol(\ell) \mid \ell \in \Xoptlines\}$
and $\{\Ysol(\ell) \mid \ell \in \Yoptlines\}$ is a separation. 
For fixed $\lead \in \mathcal{F}$, we want to ensure that we succeed if for every $(x,y) \in W_1 \cup W_2$ that is extremal in its cell we have that, if $\lead(x,y) = \cell(\ell_x,\ell_y)$, then $(x,y) \in \cell(\Xsol_X(\ell_x), \Ysol_Y(\ell_y))$.

\paragraph{Branching Step E: order of the extremal points.}
For every two abstract cells $\cell, \cell' \in \Cells$
and every two directions $\Delta,\Delta' \in \{\mathsf{top}, \mathsf{bottom}, \mathsf{left}, \mathsf{right}\}$,
we guess how 
the $\Delta$-most point in $\cell$ (the extremal point in $\cell$ in the direction $\Delta$) and
the $\Delta'$-most point in $\cell'$ 
relate in the orders $\leqx$ and $\leqy$.
Since $\leqx$ and $\leqy$ are total orders and there are $\Oh(k^2)$ extremal points in total, 
this branching step leads to $2^{\Oh(k^2 \log k)}$ subcases.

Two remarks are in order. First, in this branching step we in particular guess whenever for some cell one point is the extremal
point in more than one directions, as then the extremal points corresponding to these directions will be guessed to be equal
both in $\leqx$ and in $\leqy$.
Second, if $\cell$ and $\cell'$ are not between the same two consecutive vertical lines, then
the relation of the extremal points in $\cell$ and $\cell'$ in the order $\leqx$ can be inferred and does not need to be guessed;
similarly if $\cell$ and $\cell'$ are not between the same two consecutive horizontal lines, their relation in the $\leqy$ order can be inferred.

In what follows we will use the information guessed in this step in the following specific scenario (see Figure~\ref{fig:branchE}):
for every \apx-supercell $\apxcell$ and direction 
$\Delta \in \{\mathsf{top}, \mathsf{bottom}, \mathsf{left}, \mathsf{right}\}$,
we will be interested in the relative order in $\leqx$ (if $\Delta \in \{\mathsf{left},\mathsf{right}\}$)
or $\leqy$ (if $\Delta \in \{\mathsf{top},\mathsf{bottom}\}$)
of the $\Delta$-most extremal points of the cells in $\apxcell$ that share the $\Delta$ border with $\apxcell$.

\begin{figure}[tb]
\begin{center}
\includegraphics{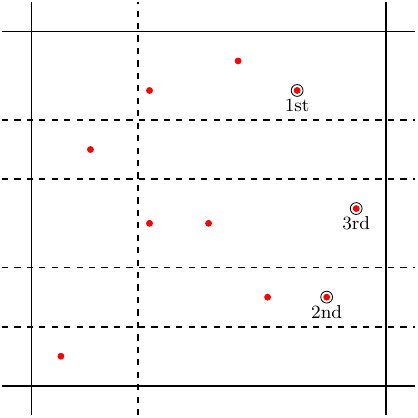}
\caption{Branching Step E and its typical later usage. 
  The step guesses the $\leqx$-order of the rightmost elements of the cells
    in one column.}\label{fig:branchE}
\end{center}
\end{figure}

\bigskip 

All the above branching steps lead to $2^{\Oh(k^2 \log k)} \log n$ subcases in total.
With each subcase, we proceed to the next steps of the algorithm.

\subsection{CSP formulation}\label[section]{sec:csp-formulation}

Recall that $\Xapxsol$ and $\Yapxsol$ map the abstract vertical line set $\Xapxlines$ and horizontal line set $\Yapxlines$, respectively, to the concrete integer coordinates and that we want to extend these functions to increasing functions $\Xsol \colon \Xlines \to \mathbb{N}$ and $\Ysol \colon \Ylines \to \mathbb{N}$ such that $\{\Xsol(\ell) \mid \ell \in \Xoptlines\}$ and $\{\Ysol(\ell) \mid \ell \in \Yoptlines\}$ is a separation.
We phrase this task as a CSP instance with binary constraints and variable set $\Xoptlines \cup \Yoptlines$, where we shall assign to each variable $\ell \in \Xoptlines$ value of $\Xsol(\ell)$ and analogous for~$\Yoptlines$.
The domains are initially defined as follows.
Let $\ell \in \Xoptlines$.
Let $\ell_1$ be the maximum element of $\Xapxlines$
with $\ell_1 < \ell$ and let $\ell_2$ be the minimum element of $\Xapxlines$
with $\ell < \ell_2$.
(Recall that here $<$ is the order of lines determined and defined after Branching Step~B.)
We define the domain $D_\ell$ of $\ell$ to be 
$$D_\ell := \{a \in \mathbb{N} \mid \apxsol(\ell_1) < a < \apxsol(\ell_2) \wedge a \equiv 2 \pmod 3\}.$$
We define the domain $D_\ell$ for each $\ell \in \Yoptlines$ analogously.
Note that, by the discretization properties, such domains can be computed in polynomial time.

To define the final CSP instance, we will in the following do two operations: introduce constraints and do filtering steps.
We will introduce constraints in five different categories: monotonicity, corner, alternations, correct order of extremal points, and alternating lines.
The filtering steps remove values from variable's domains that represent situations that we know or have guessed to be impossible.
To show correctness of the so-constructed reduction to CSP, observe that it suffices to define the constraints and conduct the filtering steps so to ensure the following two properties:
\begin{description}
\item[Soundness]--- if in the current branch we have guessed all the information about $(X,Y)$ correctly, then the pair $(\Xoptsol_X, \Yoptsol_Y)$ is a satisfying assignment to the constructed CSP instance (that is, the values of $(\Xoptsol_X, \Yoptsol_Y)$ are never removed from the corresponding domains
    in the filtering steps and $(\Xoptsol_X, \Yoptsol_Y)$ satisfies all introduced constraints).
\item[Completeness] --- for a satisfying assignment $(\Xoptsol,\Yoptsol)$ to the final CSP instance, the pair $(\{\Xoptsol(\ell) \mid \ell \in \Xoptlines\}, \{\Yoptsol(\ell) \mid \ell \in \Yoptlines\})$ is a separation.
\end{description}

We now proceed to define the five categories of constraints and a number of filtering steps.
For every introduced constraint and conducted filtering step,
    the soundess property will be straightforward. 
A tedious but relatively natural check will ensure that all introduced  constraints of the five categories
together with the filtering steps ensure the completeness property.
While introducing constraints, we will be careful to limit their number to a polynomial in~$k$ and to ensure that every introduced constraint has a segment representation of constant or $\Oh(k)$ depth.
This, together with the results of Section~\ref{sec:csp}, prove Theorem~\ref{thm:main}.

\subsection{Simple filtering steps and constraints}\label[section]{sec:simple-steps}
We start with two simple categories of constraints.

\paragraph{Monotonicity constraints.}
For every two consecutive $\ell_1,\ell_2 \in \Xoptlines$
or two consecutive $\ell_1,\ell_2 \in \Yoptlines$, we add a constraint
that the value of $\ell_1$ is smaller than the value of $\ell_2$. 

It is clear that the above constraints maintain soundness.
By Observation~\ref{obs:ineq}, every such constraint is of depth $1$ and its segment representation can be computed in polynomial time.
Furthermore, there are $\Oh(k)$ monotonicity constraints.

\paragraph{Corner filtering and corner constraints.}
Recall that $\cont \colon \Cells \to \{0, 1, 2\}$ is the function guessed in Branching Step~C that assigns to each cell the type in $\{0, 1, 2\}$ according to whether it contains points of $W_1$ (type~1), points of $W_2$ (type~1), or no points at all (type~0).
We inspect every tuple of two vertical lines $p_1, p_2 \in \Xlines$ with $p_1 < p_2$ and two horizontal lines $\ell_1, \ell_2 \in \Ylines$ with $\ell_1 < \ell_2$ such that
\begin{itemize}
\item there is no line of $\Xapxlines$ between $p_1$ and $p_2$ and
there is no line of $\Yapxlines$ between $\ell_1$ and $\ell_2$;
\item at most two lines of $\{p_1,p_2,\ell_1,\ell_2\}$ belong to $\Xoptlines \cup \Yoptlines$; and
\item according to $\cont$, every cell 
  that lies between $p_1$ and $p_2$ and between $\ell_1$ and $\ell_2$
is of type~$0$, that is, does not contain any point of $W_1 \cup W_2$.
\end{itemize}
A tuple $(p_1,p_2,\ell_1,\ell_2)$ satisfying the conditions above is called an \emph{empty corner}. 

We would like to ensure that in the space
  $\area(p_1, p_2, \ell_1, \ell_2)$ between $p_1$ and $p_2$ and between $\ell_1$ and $\ell_2$
(henceforth called the \emph{area of interest of the tuple $(p_1,p_2,\ell_1,\ell_2)$})
there are no points of $W_1 \cup W_2$. Since at most two lines of $\{p_1,p_2,\ell_1,\ell_2\}$
are from $\Xoptlines \cup \Yoptlines$, we can do it with either restricting domains of some variables
or with a relatively simple binary constraint as described below.
Herein, we distinguish the three cases of how many lines from $\Xapxlines \cup \Yapxlines$ there are in the tuple:

\subparagraph{Corner filtering.}
Observe that the area of interest of the tuple $(p_1,p_2,\ell_1,\ell_2)$ is always contained in a single \apx-supercell.
If all lines of $\{p_1,p_2,\ell_1,\ell_2\}$
are from $\Xapxlines \cup \Yapxlines$,
then the area of interest of $(p_1,p_2,\ell_1,\ell_2)$ is the \apx-supercell $\apxcell(p_1,\ell_1)$.
If this \apx-supercell contains at least one point of $W_1 \cup W_2$, we reject the current branch.

If exactly one line of $\{p_1,p_2,\ell_1,\ell_2\}$ is not from $\Xapxlines \cup \Yapxlines$,
say $\ell$, then we inspect all the values
of $D_\ell$ and delete those values for which there is some point of $W_1 \cup W_2$
in the area of interest of $(p_1,p_2,\ell_1,\ell_2)$.

It is straightforward to see that, if all guesses were correct, then the above filtering steps do not remove any value of $(\Xoptsol_X,\Yoptsol_Y)$ from the corresponding domains, that is, they preserve soundness.

\subparagraph{Corner constraints.}
If exactly two lines of $\{p_1,p_2,\ell_1,\ell_2\}$ are not from $\Xapxlines \cup \Yapxlines$, say $\ell$ and $\ell'$, then we add a constraint binding $\ell$ and $\ell'$ that allows only values $x \in D_{\ell}$ and $x' \in D_{\ell'}$ that leave the area of interest of $(p_1,p_2,\ell_1,\ell_2)$ empty.

It is straightforward to verify that, if all guesses were correct, the pair $(\Xoptsol_X,\Yoptsol_Y)$ satisfies all introduced corner constraints, that is, soundness is preserved.
We now consider the number of constraints and the running time of adding them. Indeed, as we will see below, some of the constraints above are superfluous and we can omit them.

\begin{lemma}\label[lemma]{obs:corner-2sat}
A corner constraint added for a tuple $(p_1,p_2,\ell_1,\ell_2)$ with exactly two lines from $\Xapxlines \cup \Yapxlines$
is of the form treated in Observation~\ref{obs:2sat}
and, consequently, 
is a conjunction of at most four constraints, each of depth at most $2$, and the segment representations of these constraints can be computed in polynomial time.
\end{lemma}
\begin{proof}
Let $\ell$ and $\ell'$ be the two lines of $\{p_1,p_2,\ell_1,\ell_2\}$ that are not from $\Xapxlines \cup \Yapxlines$ and let 
$\apxcell$ be the (abstract) \apx-supercell containing the area of interest of $(p_1,p_2,\ell_1,\ell_2)$.
The constraint asserting that no point of $(W_1 \cup W_2) \cap \apxcell$ is 
in the area of interest of $(p_1,p_2,\ell_1,\ell_2)$
can be expressed
as a conjunction over all $(x,y) \in (W_1 \cup W_2) \cap \apxcell$ of the constraints~$C_{x, y}$ stating that $(x,y)$ is not in the area of interest.
Constraint~$C_{x, y}$, in turn, can be expressed as $(x < \Xsol(\ell)) \vee (y < \Ysol(\ell'))$ if $\ell = p_1$ and $\ell' = \ell_1$ and similarly if $\ell$ and $\ell'$ represent other lines from $(p_1,p_2,\ell_1,\ell_2)$.
By Observation~\ref{obs:2sat}, a conjunction of such constraints~$C_{x, y}$ is a conjunction of at most four constraints, each of depth at most~$2$ and it follows from the simple form of these constraints that their segment representations can be computed in polynomial time.
\end{proof}
Let us now bound the number of corner constraints that we need to add.

There are $\Oh(k^2)$ tuples $(p_1,p_2,\ell_1,\ell_2)$ for which $p_1 \in \Xoptlines$ and $\ell_1 \in \Yoptlines$, as the choice of $p_1$ and $\ell_1$ already determines $p_2$ and $\ell_2$.
Hence, by symmetry, there are $\Oh(k^2)$ tuples $(p_1,p_2,\ell_1,\ell_2)$ that contain one line of $\Xoptlines$ and one line of $\Yoptlines$. 

Consider now a tuple $(p_1,p_2,\ell_1,\ell_2)$ where $p_1,p_2 \in \Xoptlines$ and $\ell_1,\ell_2 \in \Yapxlines$.
Then $\ell_1$ and $\ell_2$ are two consecutive elements of $\Yapxlines$; there are $\Oh(k)$ choices for them.
If there is also an empty corner $(p_1,p_2',\ell_1,\ell_2)$ with $p_2' \in \Xoptlines$ and $p_2 < p_2'$, then
the corner constaint for $(p_1,p_2',\ell_1,\ell_2)$, together with monotonicity constraints,
implies the corner constraint for $(p_1,p_2,\ell_1,\ell_2)$. Hence, we can add only corner constraints
for empty corners $(p_1,p_2,\ell_1,\ell_2)$ with maximal $p_2$. 
In this manner, we add only $\Oh(k^2)$ corner constraints for empty corners $(p_1,p_2,\ell_1,\ell_2)$ with $p_1,p_2 \in \Xoptlines$.
Similarly, we add only $\Oh(k^2)$ corner constraints for tuples $(p_1,p_2,\ell_1,\ell_2)$ with $\ell_1,\ell_2 \in \Yoptlines$.

To sum up, we add $\Oh(k^2)$ corner constraints, each of depth at most $2$.

\subsection{Alternation of a situation}\label[section]{sec:alternations}

\begin{figure}[tb]
\begin{center}
\includegraphics{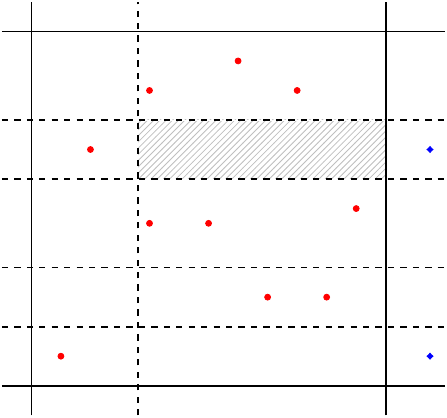}
\caption{Corner constraints are not enough: no empty corner controls the striped area in the figure. Red points are elements of $W_1$ and blue points elements of~$W_2$.
  Solid lines are from $\Xapxlines \cup \Yapxlines$, and dashed ones are from $X \cup Y$.}\label{fig:nocorner}
\end{center}
\end{figure}

\paragraph{Outline.}
Unfortunately, monotonicity and corner constraints are not sufficient to ensure completeness. 
To see this, consider an \apx-supercell $\apxcell(p_1,\ell_1)$ with $p_2$ and $\ell_2$ being the successors of $p_1$ and $\ell_1$ in $\Xapxlines$ and $\Yapxlines$, respectively. 
If there is exactly one line $p \in \Xoptlines$ between $p_1$ and $p_2$ and exactly one line $\ell \in \Yoptlines$ between $\ell_1$ and $\ell_2$,
then any of the cells $\cell(p_1,\ell_1)$, $\cell(p, \ell_1)$, $\cell(p_1,\ell)$, or $\cell(p, \ell)$ that is guessed to be empty by~$\cont$
 is taken care of by the corner constraint for the empty corner $(p_1,p,\ell_1,\ell)$, $(p, p_2, \ell_1,\ell)$, $(p_1,p,\ell,\ell_2)$, and $(p, p_2, \ell, \ell_2)$, respectively. 
 More generally, the corner constraints and other filtering performed above takes care of 
 empty cells contained in $\area(p_1, p_2, \ell_1, \ell_2)$ if there is at most one line of $\Xoptlines$ between $p_1$ and $p_2$
 and at most one line of $\Yoptlines$ between $\ell_1$ and $\ell_2$.
 However, consider a situation in which there are, say, three lines $\ell^1, \ell^2, \ell^3 \in \Yoptlines$ between $\ell_1$ and $\ell_2$ and one line $p \in \Xoptlines$ between $p_1$ and $p_2$ (see Figure~\ref{fig:nocorner}).
 If $\cont(\cell(p, \ell^2)) = 0$ but $\cont(\cell(p, \ell^1)) \neq 0$ and $\cont(\cell(p, \ell^3)) \neq 0$, then the 
 cell $\cell(p, \ell^2)$ is not contained in the area of interest of any of the empty corners and the corner constraints are not sufficient to ensure that $\cell(p, \ell^2)$ is left empty.

The problem in formulating the constraints for such sandwiched cells is that the possible values for the enclosing optimal lines depend not only on the points inside the current \apx-supercell, but also on the way points are to be separated possibly outside of the current \apx-supercell. We begin to disentangle this intricate and non-local relationship by first focusing on lines that ensure correct separation of points within the current \apx-supercell. We will call \opt-lines ensuring such separation \emph{alternating}, and their positions give rise to \emph{alternation constraints}
and \emph{alternating lines constraints}.

\medskip

We perform what follows in both dimensions, left/right and top/bottom. 
For the sake of clarity of description, we present description in the direction ``left/right'' 
(we found introducing an abstract notation of directions too cumbersome, given the complexity of the arguments). 
However, the same steps and arguments apply to the 
and to the ``top/bottom'' directions, when we swap the roles of $x$- and $y$-axes.

\begin{figure}[tb]
\begin{center}
\includegraphics{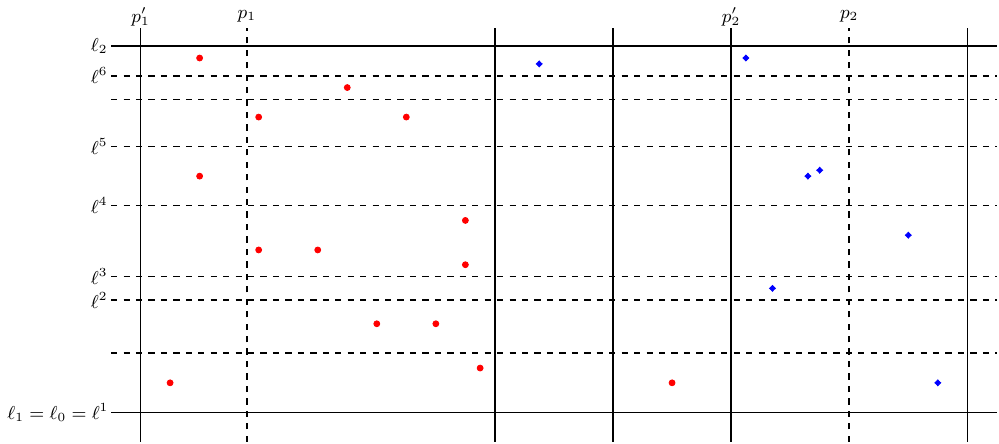}
\caption{A situation of alternation $6$. The lines of
  $\stlines_\situ$ are denoted with $\ell^i$, $1 \leq i \leq 6$.}\label{fig:alt}
\end{center}
\end{figure}

\paragraph{Definitions.}
Let $p_1,p_2$ be two consecutive elements of $\Xoptlines$ that are not consecutive elements of $\Xlines$ (i.e., there is at least one line of $\Xapxlines$ between them).
Let $\ell_1,\ell_2$ be two consecutive elements of $\Yapxlines$. 
The tuple $\situ = (p_1,p_2,\ell_1,\ell_2)$ is called a \emph{situation}.
See \cref{fig:alt} for an example.
Let $\btlines_\situ$ be the set of lines from $\Yoptlines$ that are between $\ell_1$ and $\ell_2$. 
Let $\alllines_\situ := \btlines_\situ \cup \{\ell_1\}$. 
For both $i=1,2$ let $p_i'$ be the maximum element of $\Xapxlines$ that is smaller than $p_i$.

For each $\ell \in \alllines_\situ$, we define $\area(\ell) := \area(p_1, p_2, \ell, \ell')$, where $\ell'$
is the successor of $\ell$ in $\btlines_\situ \cup \{\ell_1,\ell_2\}$.
Note that $\area(\ell) = \optcell(p_1, \ell)$ for every $\ell \in \alllines_\situ$ except for possibly
$\ell = \ell_1$ and $\ell$ being the maximum element of~$\btlines_\situ$.
However, if $\ell = \ell_1$ then $\area(\ell)$ is contained in $\optcell(p_1, \ell_1')$
where $\ell_1'$ is the predecessor of the minimum element of $\btlines_\situ$ in $\Yoptlines \cup \{(\Yapxsol)^{-1}(1)\}$, and if $\ell$ is the maximum element of $\btlines_\situ$, then $\area(\ell)$
is contained in $\optcell(p_1, \ell)$. 

Recall that $\cont \colon \Cells \to \{0, 1, 2\}$ is the function guessed in Branching Step~C that assigns to each cell its content type.
By the above inclusion-property of cells, we can extend the function $\cont$ to $\{\area(\ell) \mid \ell \in \alllines_\situ\}$ 
in the natural manner: Put $\cont(\area(\ell)) = 0$ if every cell $\cell$ contained 
in $\area(\ell)$ satisfies $\cont(\cell) = 0$ and, otherwise, $\cont(\area(\ell))$
is defined as the unique nonzero value attained by $\cont(\cell)$ for $\cell$ contained in $\area(\ell)$.
Note that the values of $\cont(\cell)$ for
cells $\cell$ contained in $\area(\ell)$ cannot attain both values $1$ and $2$, as they are all contained
in one and the same $\opt$-supercell.

An element $\ell \in \alllines_\situ$ is \emph{alternating} 
if $\cont(\area(\ell)) \neq 0$, the maximum element $\ell' \in \alllines_\situ$ with $\ell' < \ell$ and $\cont(\area(\ell')) \neq 0$
exists, and $\cont(\area(\ell)) \neq \cont(\area(\ell'))$. 
Let $\altlines_\situ$ be the set of alternating elements of $\alllines_\situ$
and let $\stlines_\situ = \altlines_\situ \cup \{\ell_0\}$ where $\ell_0$ is the minimum element of $\alllines_\situ$
with $\cont(\area(\ell_0)) \neq 0$. 
We define $\altlines_\situ = \stlines_\situ = \emptyset$ if each element $\ell \in \alllines_\situ$
has $\cont(\area(\ell)) = 0$, that is, every cell $\cell$ contained in $\area(p_1,p_2,\ell_1,\ell_2)$
satisfies $\cont(\cell) = 0$.

Consider the sequence $\seq_\situ$ consisting of values $\cont(\area(\ell))$ for $\ell \in \alllines_\situ$, ordered in the increasing order of the corresponding lines in~$\alllines_\situ$. 
Similarly, $\stseq_\situ$ is a sequence consisting of values $\cont(\area(\ell))$ for $\ell \in \stlines_\situ$, ordered in the increasing order of the corresponding lines in~$\stlines_\situ$. 

The \emph{alternation} of a situation $\situ$ is the length of the sequence $\stseq_\situ$.
Observe that equivalently we can define $\stseq_\situ$
as the maximum length of a subsequence of alternating
$1$s and $2$s in $\seq_\situ$ (the sequence may start either with a $2$ or with a $1$). 

\paragraph{Observations.}
Intuitively, in what follows we focus on alternating lines as they are the ones that separate $W_1$ from $W_2$ within the area bounded by $p_1$, $p_2$, $\ell_1$, and $\ell_2$. 
The introduced constraints are not meant to exactly focus that the content of every cell is as guessed by the function~$\cont$, but only that the alternating lines are placed correctly.
See Figure~\ref{fig:alt}.
 
We now make use of the branching steps to limit possible alternations.
\begin{lemma}\label{lem:alternations}
Assume that all guesses in the branching steps were correct regarding the solution $(X,Y)$.
For each situation $\situ = (p_1,p_2,\ell_1,\ell_2)$, the alternation equals $0$, $1$, or it is an even positive integer. 
If the alternation is at least four, then $\cont(\apxcell(p_1', \ell_1))$ and $\cont(\apxcell(p_2', \ell_1))$ are different and both are nonzero.
\end{lemma}
\begin{proof}
Let $\ell$ and $\ell'$ be the minimum and maximum elements of $\altlines_\situ$, respectively.

Suppose that the contrary of the first statement holds.
Due to symmetry between $W_1$ and~$W_2$, we may assume without loss of generality that 
$\stseq_\situ = (12)^r 1$ for some $r \geq 1$.
This in particular implies $|\altlines_\situ| = 2r \geq 2$, so $|\btlines_\situ| \geq 2$.
Let $p \in \Xapxlines$ with $p_1' \leq p \leq p_2'$, $\cont(\apxcell(p, \ell_1)) = 2$, and such that some point of $W_2$ in $\apxcell(p,\ell_1)$ lies
between $\Xoptsol_X(p_1)$ and $\Xoptsol_X(p_2)$; $p$ exists as $r \geq 1$. 
Let $(x,y) \in \apxcell(p, \ell_1) \cap W_2$ be defined as follows: if $p = p_1'$, then $(x,y)$ is the rightmost
element of $\apxcell(p, \ell_1) \cap W_2$, and if $p > p_1'$, then $(x,y)$ is the leftmost element of $\apxcell(p, \ell_1) \cap W_2$.
The point $(x,y)$ is an extremal point in the \apx-supercell $\apxcell(p, \ell_1)$. 
Observe that, by the choice of $p$, coordinate~$x$ lies between $\Xoptsol_X(p_1)$ and $\Xoptsol_X(p_2)$
while, by the structure of~$\stseq_\situ$, coordinate $y + 1$ lies between $\Yoptsol_Y(\ell)$ and $\Yoptsol_Y(\ell')$.
Since no element of $\Yapxlines$ lies between $\ell$ and $\ell'$,
this contradicts the correctness of the guess at Branching Step A.
Thus the first statement holds.

For the second statement, the reasoning is similar.
For the sake of contradiction, suppose that the contrary of the second statement holds.
Then, due to symmetry between $W_1$ and $W_2$, we may assume without loss of generality that 
\begin{align}
\cont(\apxcell(p_1', \ell_1)), \cont(\apxcell(p_2',\ell_1)) \in \{0,1\}\text{.}\label[equation]{eq:alternations1}  
\end{align}
Since the alternation is at least four, we have $|\btlines_\situ| \geq |\altlines_\situ| \geq 3$. 
Let $W$ be the set of elements of $W_2$ between $p_1$ and $p_2$ and between $\ell_1$ and~$\ell_2$.
By \cref{eq:alternations1} and since $\stseq_\situ$ contains at least one~$2$, we have $W \neq \emptyset$.
If~$\stseq_\situ = (12)^r$ for some $r \geq 2$, then let $(x,y)$ be the bottommost element of $W$
and otherwise, if $\stseq_\situ = (21)^r$, then let $(x,y)$ be the topmost element of $W$.
Observe that $(x,y)$ lies in $\apxcell(p,\ell_1)$ for some $p \in \Xapxlines$ with $p_1' < p < p_2'$
and is the bottommost or topmost, respectively, element of $\apxcell(p, \ell_1)$. 
Furthermore, $y+1$ lies between $\Yoptsol_Y(\ell)$ and $\Yoptsol_Y(\ell')$.
This again contradicts the correctness of the guess at Branching Step~A.
\end{proof}

An astute reader can observe (and it will be proven formally later) than in a situation of alternation at most two, the corner constraints and filtering steps are sufficient to ensure completeness.
Thus, we introduce below alternation and alternating-lines constraints only for situations with alternation at least four.
Henceforth we assume that the studied situation $\situ = (p_1,p_2,\ell_1,\ell_2)$ has alternation at least four. %
By symmetry and Lemma~\ref{lem:alternations}, we can assume that
$\cont(\apxcell(p_1',\ell_1)) = 1$, $\cont(\apxcell(p_2',\ell_1)) = 2$
(otherwise we swap the roles of $W_1$ and $W_2$)
and additionally that 
$\stseq_\situ = (12)^r$ for some $r \geq 2$ (otherwise we reflect the instance on an arbitrary horizontal line).
We remark also that the reflection step above may require adding $+1$ to the depth of the introduced contraints; this will not influence the asymptotic number and total depth of introduced contraints.

Assume now we are given a situation $\situ = (p_1, p_2, \ell_1, \ell_2)$ and fixed values $\Xoptsol(p_1) = x_1$ and $\Xoptsol(p_2) = x_2$. 
With these values, let $\points_\situ(x_1,x_2)$ be the set of points of $W_1 \cup W_2$ in the area bounded by $p_1$, $p_2$, $\ell_1$, and $\ell_2$ (recall that $\ell_1,\ell_2 \in \Yapxlines$). 
Define the sequence $\seq(x_1,x_2) \in \{1, 2\}^*$ as follows.
Let $(w_1, \ldots, w_s)$ be the sequence of all points from $\points_\situ(x_1,x_2)$ such that for each $i \in [s]$ we have $w_i \leqy w_{i + 1}$.
Then, $\seq(x_1,x_2) := (\alpha(w_1), \ldots, \alpha(w_s))$ where $\alpha(w_i) = \beta \in \{1, 2\}$ if $w_i \in W_\beta$.
A \emph{block} is a set of points in $\points_\situ(x_1,x_2)$ that correspond to a maximal block of consecutive equal values in $\seq(x_1,x_2)$. 
The definition of blocks is depicted in Figure~\ref{fig:blocks}.

\begin{figure}[tb]
\begin{center}
\includegraphics{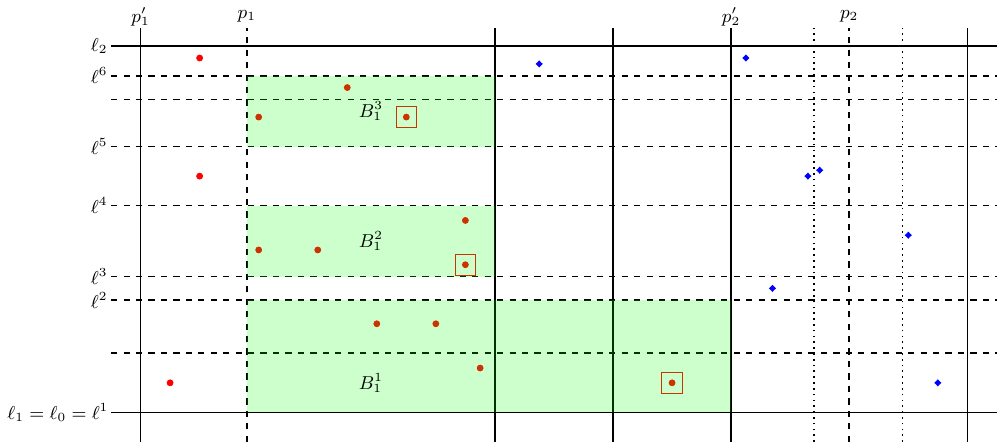}
\caption{A situation of alternation $6$ and its decompostion  to blocks when $p_1$ is positioned at $x_1\in D_{p_1}$ and $p_2$ is postioned at $x_2\in D_{p_2}$.
  The lines of $\stlines_\situ$ are denoted with $\ell^i$, $1 \leq i \leq 6$.
  Blocks of red points are denoted by $B_1^1,B_1^2,B_1^3$ and highlighted in green.
  Observe that blocks of red points do not depend on the exact position of~$x_2$.
  Leaders of red blocks are marked by a red square.
  The dotted lines indicate $x_2^\leftarrow(x_1)$ and $x_2^\rightarrow(x_1)$.
}\label{fig:blocks}
\end{center}
\end{figure}

The sequence $\stseq(x_1,x_2)$ is the subsequence of the sequence $\seq(x_1,x_2)$ that consists of the first element
and all elements whose predecessor is a different element (i.e., $\stseq(x_1,x_2)$ contains one element
for every maximal block of equal elements in $\seq(x_1,x_2)$).

We define the \emph{alternation} of points $\points_\situ(x_1,x_2)$ as follows. 
If there are two points in $\points_\situ(x_1,x_2)$ with the same $y$-coordinate but one from $W_1$ and one from $W_2$, the alternation is $+\infty$.
Otherwise, the alternation of $\points_\situ(x_1,x_2)$ is the number of maximal blocks of consecutive equal values in $\seq(x_1,x_2)$,
that is, the length of~$\stseq(x_1,x_2)$.

We have the following straightforward observation (recall that $p_1$ and $p_2$ are consecutive elements of~$\Xoptlines$).
\begin{observation}\label{obs:altx1x2}
  Let $\situ = (p_1, p_2, \ell_1, \ell_2)$ be a situation.
  If the guesses in all branching steps were correct regarding the solution $(X,Y)$ and $x_i = \Xoptsol_X(p_i)$ for $i=1,2$, then the alternations of $\points_\situ(x_1,x_2)$ and of $\situ$ are equal and $\stseq(x_1,x_2) = \stseq_\situ$.
  In particular, the alternation of $\points_\situ(x_1,x_2)$ is finite.
\end{observation}

We say that the pair $(x_1,x_2) \in D_{p_1} \times D_{p_2}$ \emph{fits} the alternation of the situation $\situ = (p_1,p_2,\ell_1,\ell_2)$ if the conclusion of Observation~\ref{obs:altx1x2} is satisfied
for $\points_\situ(x_1,x_2)$ and the situation $\situ$, that is, $\stseq(x_1,x_2) = \stseq_\situ$ and the alternation
  of $\points_\situ(x_1,x_2)$ is finite. 
Observe the following.
\begin{observation}\label{obs:alttypes}
Let $(x_1,x_2) \in D_{p_1} \times D_{p_2}$ be a pair that does not fit the situation $\situ = (p_1,p_2,\ell_1,\ell_2)$. 
Then, one of the following is true:
\begin{enumerate}[(a)]
\item The alternation of $\points_\situ(x_1,x_2)$ is finite and smaller than the alternation of $\situ$.
That is, $\stseq(x_1,x_2)$ is a proper subsequence of the sequence $\stseq_\situ$.\label{p:alt:small}
\item  The alternation of $\points_\situ(x_1,x_2)$ is infinite or not smaller than the alternation of $\situ$. That is, $\stseq(x_1,x_2)$ is not a subsequence of the sequence~$\stseq_\situ$.\label{p:alt:large}
\end{enumerate}
\end{observation}
For a pair $(x_1,x_2) \in D_{p_1} \times D_{p_2}$ that does not fit the situation $\situ = (p_1,p_2,\ell_1,\ell_2)$, we say that $(x_1,x_2)$ is \emph{of Type~\ref{p:alt:small}} or \emph{Type~\ref{p:alt:large}}, depending on which case
of Observation~\ref{obs:alttypes} it falls into.

Observe that for every $(x_1,x_2), (x_1',x_2') \in D_{p_1} \times D_{p_2}$ with $x_1 \leq x_1'$ and 
$x_2' \leq x_2$ the set $\points_\situ(x_1',x_2')$ is a subset of the set $\points_\situ(x_1,x_2)$,
so $\seq(x_1',x_2')$ is a subsequence of $\seq(x_1,x_2)$ and thus $\stseq(x_1',x_2')$ is a subsequence of $\stseq(x_1,x_2)$.
Hence, we have the following.
\begin{observation}\label{obs:altmon}
Let $(x_1,x_2) \in D_{p_1} \times D_{p_2}$ be a pair that does not fit the situation $\situ = (p_1,p_2,\ell_1,\ell_2)$. 
If $(x_1,x_2)$ is of Type~\ref{p:alt:small} and $(x_1',x_2') \in D_{p_1} \times D_{p_2}$ is such that
$x_1 \leq x_1'$ and $x_2' \leq x_2$, then $(x_1',x_2')$ is of Type~\ref{p:alt:small}, too (in particular, it does not fit $\situ$). 
Similarly, if $(x_1,x_2)$ is of Type~\ref{p:alt:large} and $(x_1',x_2') \in D_{p_1} \times D_{p_2}$ is such that
$x_1' \leq x_1$ and $x_2 \leq x_2'$, then $(x_1',x_2')$ is of Type~\ref{p:alt:large}, too (and, again, it does not fit $\situ$).
\end{observation}

\subsubsection{Filtering for correct alternation}
We exhaustively perform the following filtering operation for each situation $\sigma = (p_1, p_2, \ell_1, \ell_2)$: 
If there exists $x_1 \in D_{p_1}$ such that there is no $x_2 \in D_{p_2}$ such that $(x_1,x_2)$ fits $\situ$,
   we remove $x_1$ from $D_{p_1}$ and symmetrically, 
if there exists $x_2 \in D_{p_2}$ such that there is no $x_1 \in D_{p_1}$ such that $(x_1,x_2)$ fits $\situ$,
   we remove $x_2$ from $D_{p_2}$.
Henceforth we assume that
for every $x_1 \in D_{p_1}$ there is at least one $x_2 \in D_{p_2}$ such that $(x_1,x_2)$ fits $\situ$ and 
for every $x_2 \in D_{p_2}$ there is at least one $x_1 \in D_{p_1}$ such that $(x_1,x_2)$ fits $\situ$.
Clearly, if all the branching steps made a correct guesses, we do not remove 
neither $\Xoptsol_X(p_1)$ from $D_{p_1}$ nor $\Xoptsol_X(p_2)$ from $D_{p_2}$.
That is, this filtering step is sound.
It is not hard to see that it can be carried out in polynomial time.

For the alternating lines constraints that we introduce in \cref{sec:alternating-lines} we need the following observation on the structure of the remaining values.
Consider a value $x_1 \in D_{p_1}$ for the line variable~$p_1 \in \Xoptlines$.
Observation~\ref{obs:altmon} implies that the set of values $x_2 \in D_{p_2}$ such that $(x_1,x_2)$ fit $\situ$
forms a segment in $D_{p_2}$.
Let $x_2^\leftarrow(x_1)$ and $x_2^\rightarrow(x_1)$
be the minimum and maximum values $x_2 \in D_{p_2}$ for which $(x_1,x_2)$ fits $\situ$. 
Similarly, for a value $x_2 \in D_{p_2}$, let $x_1^\leftarrow(x_2)$ and $x_1^\rightarrow(x_2)$
be the minimum and maximum values $x_1 \in D_{p_1}$ for which $(x_1,x_2)$ fits $\situ$.
Observe that $x_1^\leftarrow$ defines a function $D_{p_1} \to D_{p_2}$ and analogously for $x_1^\rightarrow$, $x_2^\leftarrow$, and $x_2^\rightarrow$.
Note that Observation~\ref{obs:altmon} implies the following.
\begin{observation}\label{obs:altmon2}
The functions $x_1^\leftarrow$, $x_1^\rightarrow$, $x_2^\leftarrow$, and $x_2^\rightarrow$
are nondecreasing.
\end{observation}

\subsubsection{Alternation constraints}
Observation~\ref{obs:altx1x2} asserts that the values $(x_1,x_2)$ of variables $p_1$ and $p_2$ in the solution $(X,Y)$ fit the situation~$\situ$.
This motivates adding the following constraints.
For every situation $\situ = (p_1,p_2,\ell_1,\ell_2)$ of alternation at least four
we add a constraint binding $p_1$ and $p_2$ that allows only pairs of values $(x_1,x_2)$ that fit the situation~$\situ$. 

Observation~\ref{obs:altx1x2} asserts that the assignment $\Xoptsol_X \cup \Yoptsol_Y$ satisfies all alternation constraints.
Clearly, there are $\Oh(k^2)$ alternation constraints, as the choice of $p_1$ and $\ell_1$ defines the situation
$\situ = (p_1,p_2,\ell_1,\ell_2)$. 
We now prove that a single alternation constraint is a conjunction of two constraints of bounded depth.
\begin{lemma}\label{lem:alt-depth}
For a situation $\situ = (p_1,p_2,\ell_1,\ell_2)$ of alternation at least four,
the alternation constraint binding $p_1$ and $p_2$ is equivalent to a conjunction of two constraints, each with a segment representation of depth~$1$. Moreover, the latter conjunction of constraints and their segment representations can be computed in polynomial time. 
\end{lemma}
\begin{proof}
By Observation~\ref{obs:alttypes}, 
the discussed alternation constraint is a conjunction of a constraint 
``$(\Xoptsol(p_1),\Xoptsol(p_2))$ is not of Type~\ref{p:alt:small}''
and a constraint ``$(\Xoptsol(p_1), \Xoptsol(p_2))$ is not of Type~\ref{p:alt:large}''.
By Observation~\ref{obs:altmon}, the constraint
``$(\Xoptsol(p_1),\Xoptsol(p_2))$ is not of Type~\ref{p:alt:small}''
is a conjunction, 
over all pairs $(x_1,x_2) \in D_{p_1} \times D_{p_2}$ of Type~\ref{p:alt:small}
of a constraint $(\Xoptsol(p_1) < x_1) \vee (\Xoptsol(p_2) > x_2)$.
Such a conjunction can be represented with a segment representation of depth $1$ due to Observation~\ref{obs:2sat}.
Similarly, by Observation~\ref{obs:altmon}, again
the constraint ``$(\Xoptsol(p_1),\Xoptsol(p_2))$ is not of Type~\ref{p:alt:large}''
is a conjunction, 
over all pairs $(x_1,x_2) \in D_{p_1} \times D_{p_2}$ of Type~\ref{p:alt:large}
of a constraint $(\Xoptsol(p_1) > x_1) \vee (\Xoptsol(p_2) < x_2)$.
Again, such a conjunction can be represented with a segment representation of depth $1$ due to Observation~\ref{obs:2sat}.
This finishes the proof of the lemma.
\end{proof}
Consequently, by adding alternation constraints we add $\Oh(k^2)$ constraints, each of depth~$1$.

\subsection{Filtering for correct orders of extremal points}\label[section]{sec:consist-DE}

Unfortunately, alternation constraints are still not enough to ensure completeness---we need to restrict the places where the alternation occurs further in order to be able to formulate a CSP of the form described in \cref{sec:csp}.
Consider a situation $\sigma = (p_1, p_2, \ell_1, \ell_2)$, fixing a position for $p_1$, and moving the position for~$p_2$ in increasing $\leqx$-order over the positions where the correct alternation is obtained.
This gives sequences of possible positions for the horizontal lines that ensure the correct alternation.
However, intuitively, it is possible for these positions to jump within the $\leqy$-order in a noncontinuous fashion, from top to bottom and back.
We have no direct way of dealing with such discontinuity in the CSP of the form in \cref{sec:csp}. 
To avoid this behaviour, we now make crucial use of the information we have guessed in Branching Step~D and~E.
In combination with Observations~\ref{obs:altx1x2}, \ref{obs:alttypes}, and~\ref{obs:altmon}, we can smooth the admissible positions for the horizontal lines that give the correct alternation. 

Recall that we have assumed without loss of generality that $\cont(\apxcell(p_1',\ell_1)) = 1$ and $\cont(\apxcell(p_2', \ell_1)) = 2$.
Thus, having fixed the value $x_1 \in D_{p_1}$ of $\Xoptsol(p_1)$, the set of points from $W_1$ in $\points_\situ(x_1,x_2)$ is fixed, regardless of the value~$x_2 \in D_{p_2}$ of~$\Xoptsol(p_2)$.
Furthermore:
\begin{observation}\label{obs:altcontents}
Let $x_1 \in D_{p_1}$. For every $x_2 \in D_{p_2}$ with $x_2^\leftarrow(x_1) \leq x_2 \leq x_2^\rightarrow(x_1)$, we have
$$\points_\situ(x_1,x_2^\leftarrow(x_1)) \cap W_2 \subseteq \points_\situ(x_1,x_2) \cap W_2 \subseteq \points_\situ(x_1,x_2^\rightarrow(x_1)) \cap W_2.$$
\end{observation}
\noindent Even further, the following important observation states that the partition of $W_1 \cap \points_\situ(x_1,x_2)$ into blocks does not depend on~$x_2$.
\begin{observation}\label{obs:altblocks}
  Let $x_1 \in D_{p_1}$.
  Then, the partition of the points of $W_1 \cap \points_\situ(x_1,x_2)$ into blocks is the same for any choice of $x_2 \in D_{p_2}$ with $x_2^\leftarrow(x_1) \leq x_2 \leq x_2^\rightarrow(x_1)$.
  Symmetrically, let $x_2 \in D_{p_2}$.
  Then, the partition of the points of $W_2 \cap \points_\situ(x_1,x_2)$ into blocks is the same for any choice of $x_1 \in D_{p_1}$ with $x_1^\leftarrow(x_2) \leq x_1 \leq x_1^\rightarrow(x_2)$.
\end{observation}
\begin{proof}
We prove only the first statement, the second one is symmetrical.
Fix two integers $x_2^\leftarrow(x_1) \leq x_2 \leq x_2' \leq x_2^\rightarrow(x_1)$. 
Then the sequence $\seq(x_1,x_2)$ is a subsequence of $\seq(x_1,x_2')$
that contains the same number of $1$s; they differ only in the number of $2$s.
Since both $(x_1,x_2)$ and $(x_1,x_2')$ fit $\situ$, $\stseq(x_1,x_2) = \stseq(x_1,x_2')$.
Hence, the maximal sequences of consecutive $1$s in $\seq(x_1,x_2)$ and $\seq(x_1,x_2')$ are the same.
Since set of points from $W_1$ in $\points_\situ(x_1,x_2)$ and $\points_\situ(x_1,x_2')$ are the same, the statement follows.
\end{proof}
Observation~\ref{obs:altblocks} allows us to make the following filtering step for situation $\situ$ using the information guessed in Branching Steps D and E.
Informally, in these branching steps we have guessed for each point for which cell it can be an extremal point.
Since the blocks of $W_1 \cap \points_\situ(x_1,x_2)$ are fixed once $x_1$ is fixed, some of the extremal points are fixed, and we can now remove values from $D_{p_1}$ for which this guess would be incorrect.
Similar for~$x_2$.
The formal filtering step works as follows.

Recall that $\stseq(x_1,x_2^\leftarrow(x_1)) = (12)^r$ for some $r \in \mathbb{N}$, $r \geq 2$.
Let $\ell^1,\ell^2, \ldots, \ell^{2r}$ be the elements of $\stlines_\situ$ in increasing order
(i.e., $\altlines_\situ = \{\ell^2, \ell^3, \ldots, \ell^{2r}\}$, cf. Figure~\ref{fig:alt}).

For each $x_1 \in D_{p_1}$, let $B_1^1(x_1), B_1^2(x_1), \ldots, B_1^r(x_1)$ be the partition of $\points_\situ(x_1,x_2^\leftarrow(x_1)) \cap W_1$ into blocks
in the increasing order of~$\leqy$.
Similarly, for each $x_2 \in D_{p_2}$, let $B_2^1(x_2), \ldots, B_2^r(x_2)$ be the partition of $\points_\situ(x_1^\rightarrow(x_2), x_2) \cap W_2$ into blocks
in the increasing order of $\leqy$. 
For each $i \in [r]$, let $\leader_1^i(x_1)$ be the rightmost element of $B_1^i(x_1)$
and let $\leader_2^i(x_2)$ be the leftmost element of $B_2^i(x_2)$.
Below we call these elements \emph{leaders}.

Assume that all branching steps made correct guesses regarding the solution~$(X, Y)$ and consider $x_1^{X,Y} := \Xoptsol_X(p_1)$, $x_2^{X,Y} := \Xoptsol_X(p_2)$. 
Then, for every $i \in [r]$, by the definition of alternating lines, the $y$-coordinate $\Yoptsol_Y(\ell^{2i})$ lies between the $y$-coordinates of the points of $B_1^i(x_1^{X,Y})$
and $B_2^i(x_2^{X,Y})$
and, for every $i \in [r]$ with $i > 1$, the $y$-coordinate $\Yoptsol_Y(\ell^{2i-1})$ lies between the $y$-coordinates of the points of $B_2^{i-1}(x_2^{X,Y})$
and $B_1^i(x_1^{X,Y})$.
Also, for every $i \in [r]$, the element $\leader_1^i(x_1^{X,Y})$ is the rightmost element of its cell and $\leader_2^i(x_2^{X,Y})$ is the leftmost element of its cell.

Fix $i \in [r]$. 
We now observe that we can deduce from the information guessed in the branching steps 
to which cell the element $\leader_1^i(x_1^{X,Y})$ belongs.
Indeed, $B_1^i(x_1^{X,Y})$ consists of the cells $\cell(p, \ell)$ for all $p \in \Xlines$ and $\ell \in \Ylines$ with $p_1 \leq p < p_2$ and $\ell^{2i-1} \leq \ell < \ell^{2i}$.
At Branching Step~C we have guessed which of these cells are empty and which contain some element of $W_1$:
We expect that $\cont(\cell(p,\ell)) \in \{0,1\}$ for every such pair $(p,\ell)$ as above and $\cont(\cell(p,\ell)) = 1$ for at least one such pair; 
we reject the current branch if this is not the case.
The information guessed at Branching Step E allows us to infer 
\begin{itemize}
\item the cell $\cell_1^i \in \{\cell(p, \ell) \mid p_1 \leq p < p_2 \wedge \ell^{2i-1} \leq \ell < \ell^{2i}\}$ that contains $\leader_1^i(x_1^{X,Y})$; and
\item the relative order in $\leqx$ of the elements $\leader_1^i(x_1^{X,Y})$ for $1 \leq i \leq r$. 
\end{itemize}
Observe that, by \cref{obs:altblocks}, given $x_1 \in D_{p_1}$, we can in polynomial time compute whether the above two properties hold.
We remove from $D_{p_1}$ all values $x_1$ for which the order discussed in the second point above is not as expected.
Also, if the information guessed at Branching Step~D is correct, we have $\lead(\leader_1^i(x_1^{X,Y})) = \cell_1^i$. 
We remove from $D_{p_1}$ all values $x_1$ for which there exists $i \in [r]$ with $\lead(\leader_1^i(x_1)) \neq \cell_1^i$.
It is clear that this filtering step is sound and, as mentioned, it can be carried out in polynomial time.

We perform symmetrical analysis with the elements $\leader_2^i(x_2^{X,Y})$. 
That is, the information guessed at Branching Step E allows us to infer 
\begin{itemize}
\item the cell $\cell_2^i \in \{\cell(p, \ell) \mid p_1 \leq p < p_2 \wedge \ell^{2i} \leq \ell < \ell^{2i+1}\}$ (with $\ell^{2r+1} = \ell_2$) that contains $\leader_2^i(x_2^{X,Y})$; and
\item the relative order in $\leqx$ of the elements $\leader_2^i(x_2^{X,Y})$ for $i \in [r]$. 
\end{itemize}
We remove from $D_{p_2}$ all values $x_2$ for which the relative order in $\leqx$ of the elements $\leader_2^i(x_2)$ for $i \in [r]$ is not as expected above.
Also, if the information guessed at Branching Step D is correct, we have $\lead(\leader_2^i(x_2^{X,Y})) = \cell_2^i$. 
We remove from $D_{p_2}$ all values $x_2$ for which there exists $i \in [r]$ with $\lead(\leader_2^i(x_2)) \neq \cell_2^i$.

\subsection{Alternating lines constraints}\label[section]{sec:alternating-lines}
In the previous section we smoothed the possible positions for horizontal lines where the guessed alternation occurs.
This enables us now to introduce constraints that describe these positions and have a form that is suitable for the type of CSP of \cref{sec:csp}.

Fix a situation $\situ = (p_1, p_2, \ell_1, \ell_2)$ of alternation~$a \geq 4$ and let $r = a/2$.
Recall the definitions of the lines $\ell^{1}, \ell^{2}, \ldots, \ell^{2r}$, blocks $B_1^1(\cdot), B_1^2(\cdot), \ldots, B_1^r(\cdot)$, and blocks $B_2^1(\cdot), B_2^2(\cdot) \ldots, B_2^r(\cdot)$ from the previous section.
We introduce the following constraints.
\begin{enumerate}
\item 
For every $i \in [r]$, we introduce a constraint binding $p_1$ and $\ell^{2i}$ that asserts that $\Yoptsol(\ell^{2i})$ is larger than the largest $y$-coordinate
of an element of $B_1^i(\Xoptsol(p_1))$ (i.e., the line $\ell^{2i}$ is above the block $B_1^i$). 
\item 
For every $i \in [r] \setminus \{1\}$, we introduce a constraint binding $p_1$ and $\ell^{2i-1}$ that asserts that $\Yoptsol(\ell^{2i-1})$ is smaller than the smallest $y$-coordinate
of an element of $B_1^i(\Xoptsol(p_1))$ (i.e., the line $\ell^{2i-1}$ is below the block $B_1^i$). 
\item 
For every $i \in [r]$, we introduce a constraint binding $p_2$ and $\ell^{2i}$ that asserts that $\Yoptsol(\ell^{2i})$ is smaller than the smallest $y$-coordinate
of an element of $B_2^i(\Xoptsol(p_2))$ (i.e., the line $\ell^{2i}$ is below the block $B_2^i$). 
\item 
For every $i \in [r - 1]$, we introduce a constraint binding $p_2$ and $\ell^{2i+1}$ that asserts that $\Yoptsol(\ell^{2i+1})$ is larger than the largest $y$-coordinate
of an element of $B_2^i(\Xoptsol(p_2))$ (i.e., the line $\ell^{2i+1}$ is above the block $B_2^i$). 
\end{enumerate}
We call the above constraints \emph{alternating-lines constraints}.
Again, the soundness property of the new constraints is straightforward.
We now prove that the alternating lines constraints are of bounded depth (in the sense of \cref{sec:csp}) and that a corresponding representation can be computed in polynomial time.
This is the intuitive statement behind the following highly nontrivial lemma, whose proof spans the rest of this subsection.
\begin{lemma}\label{lem:imp-depth}
Let $\situ = (p_1,p_2,\ell_1,\ell_2)$ be a situation of alternation $a \geq 4$ and let $r = a/2$.
Assume that $\cont(\apxcell(p_i', \ell_1)) = i$ for $i=1,2$, where $p_i'$ is the predecessor of $p_i$ in~$\Xapxlines$, and that $\stseq_\situ = (12)^{r}$.
Then one can in polynomial time compute two rooted trees $T_j$ for $j=1,2$
with $\leaves(T_j) = \{v_j^1, v_j^2, \ldots, v_j^a, u_j^1, u_j^2, \ldots, u_j^r\}$ and
$|V(T_j)| = \Oh(r)$,
two families of segment reversions $\mathcal{G}_j = (g_{j,v})_{v \in V(T_{j}) \setminus \treeroot(T_j)}$ for $j=1,2$,
and four families of downwards-closed relations $(R_j^i)_{i=1}^r$
for $j=1,2,3,4$ such that the following holds. 
For every $i \in [r]$ and $j=1,2$, let
$v_j^i = w_{j,1}, w_{j,2}, \ldots, w_{j,b_j^i} = \treeroot(T_j)$ be the nodes on the path from $v_j^i$
to $\treeroot(T_j)$ in the tree $T_j$ and let
$u_j^i = z_{j,1}, z_{j,2}, \ldots, z_{j,c_j^i} = \treeroot(T_j)$ be the nodes on the path from $u_j^i$
to $\treeroot(T_j)$ in the tree $T_j$.
Then, for every $i \in [r]$, 
\begin{enumerate}
\item the first alternating-lines constraint for block $B_1^i$ and the line $\ell^{2i}$
is equivalent to 
$$(g_{1,w_{1,b_1^i - 1}} \circ g_{1,w_{1,b_1^i -2}} \circ \ldots \circ g_{1,w_{1,1}}(\Xoptsol(p_1)), g(\Yoptsol(\ell^{2i}))) \in R_1^i,$$
where $g$ is the segment reversion that reverses the whole domain of $\ell^{2i}$;
\item if $i > 1$, then the second alternating-lines constraint for block $B_1^i$ and the line $\ell^{2i-1}$
is equivalent to 
$$(g \circ g_{1,z_{1,c_1^i - 1}} \circ g_{1,z_{1,c_1^i -2}} \circ \ldots \circ g_{1,z_{1,1}}(\Xoptsol(p_1)), \Yoptsol(\ell^{2i})) \in R_2^i,$$
where $g$ is the segment reversion that reverses the whole domain of $p_1$;
\item the third alternating-lines constraint for block $B_2^i$ and the line $\ell^{2i}$
is equivalent to 
$$(g \circ g_{2,w_{2,b_2^i - 1}} \circ g_{2,w_{2,b_2^i -2}} \circ \ldots \circ g_{2,w_{2,1}}(\Xoptsol(p_2)), \Yoptsol(\ell^{2i})) \in R_3^i,$$
where $g$ is the segment reversion that reverses the whole domain of $p_2$; and
\item if $i < r$, then the fourth alternating-lines constraint for block $B_2^i$ and the line $\ell^{2i+1}$
is equivalent to 
$$(g_{2,z_{2,c_2^i - 1}} \circ g_{2,z_{2,c_2^i -2}} \circ \ldots \circ g_{2,z_{2,1}}(\Xoptsol(p_2)), g(\Yoptsol(\ell^{2i+1}))) \in R_4^i,$$
where $g$ is the segment reversion that reverses the whole domain of $\ell^{2i}$.
\end{enumerate}
\end{lemma}
\noindent
In other words, the alternating-lines constraints have segment representations of depth $O(a)$ whose sequences of permutations correspond to root-leaf paths in two trees.

We now proceed to prove \cref{lem:imp-depth}.
Recall that $p_1,p_2 \in \Xoptlines$, $\ell_1,\ell_2 \in \Yapxlines$.
We present the proof for the first two types of alternating lines constraint; the proof of the other types is analogous (i.e., one can consider a center-symmetric image of the instance with the roles of sets $W_1$ and $W_2$ swapped).
That is, we show how to compute the tree $T_1$, the family $\mathcal{G}_1$, and the relations $(R_j^i)$ for $j=1,2$ and $1 \leq i \leq a$.

Let $B_1^1(x_1),B_1^2(x_1),\ldots, B_1^r(x_1)$ be the blocks of $W_1$ in the bottom-to-top order in the situation $\situ = (p_1,p_2,\ell_1,\ell_2)$ when $p_1$ is positioned at $x_1 \in D_{p_1}$.
Recall that, a fixed value $x_1$ for $p_1$ determines the content of $W_1 \cap \points_\situ(x_1,x_2)$ regardless of the choice of the value $x_2$ for $p_2$ (see \cref{obs:altcontents}).
Moreover, by \cref{obs:altblocks}, fixing a value $x_1$ for $p_1$ also determines the partition of $W_1 \cap \points_\situ(x_1,x_2)$ into blocks $B_1^i(x_1)$ (which justifies the notation $B_1^i(x_1)$), see Figure~\ref{fig:blocks}.

Let $\pi_1 : [r] \to [r]$ be a permutation such that $B_1^{\pi_1(1)}(x_1), B_1^{\pi_1(2)}(x_1), \ldots, B_1^{\pi_1(r)}(x_1)$ is the ordering of $B_1^i$s
in the decreasing order with regard to $\leqx$ (i.e., right-to-left) of the leaders (rightmost elements) of $B_1^i(x_1)$. 
That is, we compute a permutation $\pi_1 : [r] \to [r]$ such that for every $x_1 \in D_{p_1}$ we have
$$\leader_1^{\pi_1(r)}(x_1) \leqx \leader_1^{\pi_1(r-1)}(x_1) \leqx \ldots \leqx \leader_1^{\pi_1(1)}(x_1).$$
Observe that $\pi_1$ can be computed in polynomial time using the information guessed in Branching Step~E.

Similarly, let $B_2^1(x_2), \ldots, B_2^r(x_2)$ be the blocks of $W_2$ in the bottom-to-top order with $p_2$ positioned at $x_2 \in D_{p_2}$
and from Branching Step~E we infer a permutation $\pi_2: [r] \to [r]$ such that for every $x_2 \in D_{p_2}$ we have (recall that $\leader_2^i(x_2)$ is the leftmost element of the block $B_2^i(x_2)$)
$$\leader_2^{\pi_2(1)}(x_2) \leqx \leader_2^{\pi_2(2)}(x_2) \leqx \ldots \leqx \leader_2^{\pi_2(r)}(x_2).$$
In what follows, the argument $x_1$ or $x_2$ in $B_1^i(x_1)$ or $B_2^j(x_2)$ will sometimes be superfluous when we only discuss the bottom-to-top order
of these blocks or the left-to-right order of their leaders---these orders are fixed regardless of $x_1$ or $x_2$. In such cases we will omit the argument.

Let $\leadpos_i : D_{p_1} \to \mathbb{N}$ be the function that assigns to $x_1 \in D_{p_1}$ the $y$-coordinate of $\leader_1^i(x_1)$,
    $\leadup_i : D_{p_1} \to \mathbb{N}$ be the function that assigns to $x_1 \in D_{p_1}$ the $y$-coordinate of the topmost element of the block $B_1^i(x_1)$, 
and $\leaddown_i : D_{p_1} \to \mathbb{N}$ be the function that assigns to $x_1 \in D_{p_1}$ the $y$-coordinate of the bottommost element of the block $B_1^i(x_1)$.

The main ingredient in the proof of \cref{lem:imp-depth}, and our main technical result, is the following lemma, which captures the structure of possible placements of vertical lines as a tree-like application of a bounded number of segment reversions.

\begin{lemma}\label{lem:leaders}
In polynomial time, one can compute a rooted tree $T'$
with $\leaves(T') = \{v^1, v^2, \ldots, v^a, u^1, u^2, \ldots, u^a\}$ and $|V(T')| = \Oh(a)$,
a family of segment reversions
$\mathcal{G} = (g_{v})_{v \in V(T') \setminus \{\treeroot(T')\}}$,
  and a family of nondecreasing functions $\widehat{\mathcal{F}} = (\hat{f}_v)_{v \in \leaves(T')}$
such that the following holds.
For every $i \in [a]$, 
if $v^i = v_1, v_2, \ldots, v_b = \treeroot(T')$ is the path from $v^i$ to the root $\treeroot(T')$
and $u^i = u_1, u_2, \ldots, u_c = \treeroot(T')$ is the path from $u^i$ to the root $\treeroot(T')$,
then
\begin{align*}
\leadup_i &= \hat{f}_{v^i} \circ g_{v_{b-1}} \circ g_{v_{b-2}} \circ \ldots \circ g_{v_1},\\
\leaddown_i &= \hat{f}_{u^i} \circ g_{u_{b-1}} \circ g_{u_{b-2}} \circ \ldots \circ g_{u_1}.\\
\end{align*}
\end{lemma}

We now show how Lemma~\ref{lem:leaders} implies Lemma~\ref{lem:imp-depth}.
First, compute the tree $T'$, segment-reversion family $\mathcal{G}$, and family of nondecreasing functions~$\widehat{\mathcal{F}}$
via Lemma~\ref{lem:leaders}.
Let $i \in [r]$ arbitrary.
Note that the first alternating lines constraint is equivalent to:
\begin{equation}\label{eq:lead0:1}
\leadup_i(\Xoptsol(p_1)) < \Yoptsol(\ell^{2i}).
\end{equation}
Let 
$v^i = v_1, v_2, \ldots, v_b = \treeroot(T')$ be the path from $v^i$ to the root $\treeroot(T')$.
By Lemma~\ref{lem:leaders}, \eqref{eq:lead0:1}~is equivalent to
\begin{equation}\label{eq:lead1:1}
\hat{f}_{v^i} \circ g_{v_{b-1}} \circ g_{v_{b-2}} \circ \ldots \circ g_{v_1}(\Xoptsol(p_1)) < \Yoptsol(\ell^{2i})\text{.}
\end{equation}
Hence, if $g$ is the segment reversion reversing $D_{\ell^{2i}}$, then~\eqref{eq:lead1:1} is equivalent to
\begin{equation}\label{eq:lead2:1}
(\hat{f}_{v^i} \circ g_{v_{b-1}} \circ g_{v_{b-2}} \circ \ldots \circ g_{v_1}(\Xoptsol(p_1)), g(\Yoptsol(\ell^{2i}))) \in R
\end{equation}
for some downwards-closed relation~$R$.
For example, we may take $R = \{(x, y) \in \mathbb{N}^2 \mid y \leq y_{\textsf{max}} - x\}$, where $y_{\textsf{max}}$ is the largest $y$-coordinate of any horizontal line.
By Lemma~\ref{lem:dc-comp}, we can compute a downwards-closed relation $R_1^i$ such that~\eqref{eq:lead2:1} is equivalent to
\begin{equation}\label{eq:lead3:1}
(g_{v_{b-1}} \circ g_{v_{b-2}} \circ \ldots \circ g_{v_1}(\Xoptsol(p_1)), g(\Yoptsol(\ell^{2i}))) \in R_1^i\text{.}
\end{equation}
Similarly, if $1 < i \leq r$, the second alternating lines constraint is equivalent to
\begin{equation}\label{eq:lead0:2}
\leaddown_i(\Xoptsol(p_1)) > \Yoptsol(\ell^{2i-1})\text{.}
\end{equation}
Let 
$u^i = u_1, u_2, \ldots, u_b = \treeroot(T')$ be the path from $u^i$ to the root $\treeroot(T')$.
By Lemma~\ref{lem:leaders},~\eqref{eq:lead0:2} is equivalent to
\begin{equation}\label{eq:lead1:2}
\hat{f}_{u^i} \circ g_{u_{b-1}} \circ g_{u_{b-2}} \circ \ldots \circ g_{u_1}(\Xoptsol(p_1)) > \Yoptsol(\ell^{2i-1})\text{.}
\end{equation}
Hence, if $g$ is the function reversing $\{\Yapxlines(\ell_1),\Yapxlines(\ell_1)+1,\ldots, \Yapxlines(\ell_2)\}$,
 then~\eqref{eq:lead1:2} is equivalent to
\begin{equation}\label{eq:lead2:2}
(g \circ \hat{f}_{u^i} \circ g_{u_{b-1}} \circ g_{u_{b-2}} \circ \ldots \circ g_{u_1}(\Xoptsol(p_1)), \Yoptsol(\ell^{2i-1})) \in R
\end{equation}
for some downwards-closed relation~$R$. 
Define $g'$ to be the segment reversion reversing the whole $D_{p_1}$
and $f' = \hat{f}_{u^i} \circ g'$. Then, since $g'$ is an involution,
 \eqref{eq:lead2:2} is equivalent to:
\begin{equation}\label{eq:lead2b:2}
(g \circ f' \circ g' \circ g_{u_{b-1}} \circ g_{u_{b-2}} \circ \ldots \circ g_{u_1}(\Xoptsol(p_1)), \Yoptsol(\ell^{2i-1})) \in R
\end{equation}
Note that $g \circ f' = g \circ \hat{f}_{u^i} \circ g'$ is a nondecreasing function.
By Lemma~\ref{lem:dc-comp} applied to $g \circ f'$ and $R$, one can compute 
a downwards-closed relation $R^i_2$ such that~\eqref{eq:lead2:2}
is equivalent to:
\begin{equation}\label{eq:lead3:2}
(g' \circ g_{u_{b-1}} \circ g_{u_{b-2}} \circ \ldots \circ g_{u_1}(\Xoptsol(p_1)), \Yoptsol(\ell^{2i-1})) \in R_2^i.
\end{equation}
Using~\eqref{eq:lead3:1} and~\eqref{eq:lead3:2}, the following satisfy the conditions of Lemma~\ref{lem:imp-depth}:
\begin{itemize}
\item the tree $T_1$ derived from $T'$ by adding an extra child $u^i_0$ to every node $u^i$, 
\item the family $\mathcal{G}$ derived from $\mathcal{G}_1$ by adding $g_{u^i_0}$, defined as the segment reversion reversing the whole $D_{p_1}$, and
\item the relations $R_j^i$.
\end{itemize}
Thus, it remains to prove Lemma~\ref{lem:leaders}.

\begin{proof}[Proof of Lemma~\ref{lem:leaders}.]
For two blocks $B_1^d$ and $B_1^e$, we say that a block $B_2^j$ is \emph{between $B_1^d$ and $B_1^e$}
if it is between $B_1^d$ and $B_1^e$ in the bottom-to-top order, that is,
if $d < e$ and $d \leq j < e$ or $e < d$ and $e \leq j < d$.

Recall that $r$ is the number of blocks of $W_1$ (and of $W_2$) and recall
the definition of the permutation~$\pi_1$ that permutes the sequence $B_1^1, B_1^2, \ldots B_1^r$ of blocks so that their leaders are increasing in the $\leqx$ order.
We define an auxiliary rooted tree $T$ with $V(T) = [r]$ as follows. The root of $T$ is $\pi_1(1)$.
For every $i \in [r] \setminus \{\pi_1(1)\}$, we define the parent of $i$ as follows. 
Let $i_1$ be the maximum index $i_1 < i$ with $\pi_1^{-1}(i_1) < \pi_1^{-1}(i)$ (i.e., the leader of $B_1^{i_1}$ being to the right of the leader of $B_1^i$). Similarly 
let $i_2$ be the minimum index $i_2 > i$ with $\pi_1^{-1}(i_2) < \pi_1^{-1}(i)$.
These indices are undefined if the maximization or minimization is chosen over an empty set; however note that, due to the presence of $B_1^{\pi_1(1)}$, at least one of these indices
is defined. If exactly one is defined, we take this index to be the parent of $i$ in $T$.
Otherwise, we look at the leftmost of all leaders of all blocks $B_2^j$ between $B_1^{i_1}$ and $B_1^i$ (i.e., $i_1 \leq j < i$)
and at the leftmost of all leaders of all blocks $B_2^j$ between $B_1^i$ and $B_1^{i_2}$ (i.e., $i \leq j < i_2$)
and choose as parent of~$i$ the index $i_\alpha$, $\alpha=1,2$, for which the aforementioned leader is more to the right (i.e., its block is later in the permutation $\pi_2$).
Note that $T$ can be constructed from the information guessed in Branching Step~E.
See Figure~\ref{fig:sliding2} for an example and Figure~\ref{fig:sliding2advanced} for a more involved example.
In the following, the parent of a node $i$ in $T$ is denoted $\parent(i)$.
Furthermore, for each $i \in [r]$, we let $T_i$ be the subtree of $T$ rooted at $i$, 
let $\widehat{B}_i$ be the union of all blocks $B_1^j$ for $j \in V(T_i)$.

We will use tree~$T$ below to define a tree of segment partitions to which we can apply the tools from \cref{ss:tree}, yielding the required family of segment reversions.
The segment partitions associated with the vertices of~$T$ will be defined based on the nested behavior of blocks when moving $p_1$ in increasing $\leqx$-order.
Before we can define the partitions associated with the vertices of $T$, we need to establish a few properties of blocks.

First, no two blocks of $W_1$ share leaders.
\begin{claim}\label{cl1}
Let $e \in W_1$ and assume $e$ is the leader of some $B_1^j(x_1)$. 
Then, $e$ is not a leader of any block $B_1^{j'}(x_1')$ with $j' \neq j$.
\end{claim}
\begin{proof}
The claim follows directly from the filtering for correct orders of extremal points (\cref{sec:consist-DE}): If $e$ is the leader of $B_1^j(x_1)$, then $\lead(e) = \cell_1^j$. (Recall that $\cell_1^j$ is the cell that is expected
to contain the leader of $B_1^j(x_1)$ and is inferred from the information guessed in Branching Step~E.)
\cqed\end{proof}
Next, increasing the position of $p_1$ can only shrink blocks of $W_1$:
\begin{claim}\label{cl2}
Let $x_1, x_1'$ be two elements of $D_{p_1}$ with $x_1 < x'_1$.
Then for every block $B_1^{j'}(x_1')$ there exists a block $B_1^j(x_1)$ such that
$B_1^{j'}(x_1') \subseteq B_1^j(x_1)$. 
\end{claim}
\begin{proof}
By Observation~\ref{obs:altmon2}, $x_2^\leftarrow(x_1) \leq x_2^\leftarrow(x_1')$, that is,
$$W_2 \cap \points_\situ(x_1, x_2^\leftarrow(x_1)) \subseteq W_2 \cap \points_\situ(x_1', x_2^\leftarrow(x_1')).$$
This immediately implies that every block $B_1^{j'}(x_1')$ is contained in some block
$B_1^j(x_1)$, as desired.
\cqed\end{proof}
Next, each leader has some well-defined interval of positions of $p_1$ during which it is the leader of its block.
We first state the boundaries of this interval and then prove that they are well-defined.

Let $e \in W_1$ be the leader of some block, that is, there exist $j \in [r]$ and $x_1 \in D_{p_1}$ such that $e$ is the leader of~$B_1^j(x_1)$.
Define $\leadact_1^\rightarrow(e) \in D_{p_1}$ to be the maximum element of $D_{p_1}$ that is smaller than the $x$-coordinate of~$e$.
Define $\leadact_1^\leftarrow(e)$ to be the element in $D_{p_1}$ that satisfies that $e$ is a leader of $B_1^j(x_1)$ if and only if $\leadact_1^\leftarrow(e) \leq x_1 \leq \leadact_1^\rightarrow(e)$.

Note that $\leadact_1^\rightarrow(e)$ is well-defined since, for $e$ to be leader of $B_1^j(x_1)$, value $x_1 \in D_{p_1}$ needs to be smaller than the $x$-coordinate of $e$, showing that $\leadact_1^\rightarrow(e)$ exists.
\begin{claim}\label{cl3}
  $\leadact_1^\leftarrow(e)$ is well-defined.
\end{claim}
\begin{proof}
It suffices to show that, if
$e$ is the leader of $B_1^j(x_1)$ for some $x_1 \in D_{p_1}$, then 
it is also the leader of $B_1^j(x_1')$ for every $x_1' \in D_{p_1}$ with $x_1 \leq x_1' \leq \leadact_1^\rightarrow(e)$
(note that $x_1 \leq \leadact_1^\rightarrow(e)$ by the definition of $\leadact_1^\rightarrow(e)$). 

Let $B_1^{j'}(x_1')$ be the block containing $e$ and let $e'$ be the leader of this block. 
By Claim~\ref{cl2} and the fact that $B_1^{j'}(x_1')$ and $B_1^{j}(x_1)$ share $e$ we have $B_1^{j'}(x_1') \subseteq B_1^j(x_1)$. 
Thus, the fact that $e'$ is the leader of $B_1^{j'}(x_1')$ implies $e \leqx e'$
while the fact that $e$ is the leader of $B_1^j(x_1)$ implies $e' \leqx e$. 
Hence, $e = e'$. By Claim~\ref{cl1}, $j = j'$ and we are done.
\cqed\end{proof}

Intuitively, there are two things that can happen to a block with some index~$j$ when moving $p_1$ to the right: It can shrink, or it can disappear and reappear elsewhere.
Now, if increasing the position of $p_1$ shrinks a block but it does not disappear, then the leader stays the same:
\begin{claim}\label{cl4}
If for some $x_1,x_1' \in D_{p_1}$ with $x_1 < x_1'$ and an index $j \in [r]$
we have $B_1^j(x_1') \subseteq B_1^j(x_1)$, then $\leader_1^j(x_1) = \leader_1^j(x_1')$.
\end{claim}
\begin{proof}
Since $B_1^j(x_1') \subseteq B_1^j(x_1)$, the leader $\leader_1^j(x_1)$ is to the right
of the coordinate $x_1'$. Thus, $x_1' \leq \leadact_1^\rightarrow(\leader_1^j(x_1))$ by the definition of $\leadact_1^\rightarrow$.
Thus, $\leader_1^j(x_1)$ is also a leader of $B_1^j(x_1')$.
\cqed\end{proof}

\begin{figure}[tb]
\begin{center}
\includegraphics{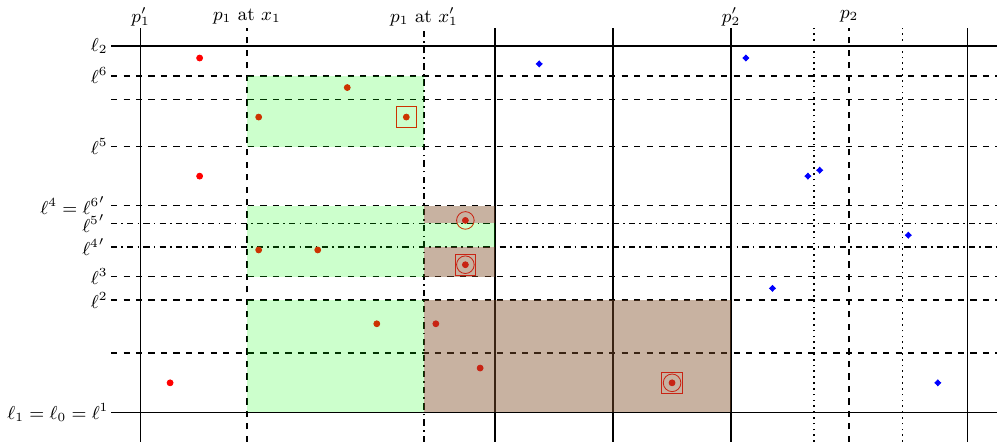}
\caption{
  Situation of Claim~\ref{cl4} and \ref{cl5}, where $p_1$ is either at position $x_1\in D_{p_1}$ or at position $x_1'\in D_{p_1}$ for $x<x'$.
  The lines of $\stlines_\situ$ for $x_1$ are denoted with $\ell^i$, $1 \leq i \leq 6$.
  The lines of $\stlines_\situ$ for $x_1'$ are denoted with $\ell'^{i}$, $1 \leq i \leq 6$.
  Blocks given by positioning $p_1$ at $x_1'$ ($B_1^{1}(x_1'),B_1^{2}(x_1'),B_1^{3}(x_1')$) are depicted by purple color with circled leaders and blocks given by positioning $p_1$ at $x_1$ ($B_1^{1}(x_1),B_1^{2}(x_1),B_1^{3}(x_1)$) are depicted by the union of green and purple color with squared leaders.
}\label{fig:sliding}
\end{center}
\end{figure}

Next we observe that, when moving $p_1$ to the right from one position $x_1$ to another position $x_1'$, then, when ordering the blocks of $W_1$ according to~$\pi_1$, that is, increasing $x$-coordinates of leaders, then there is a unique block index $i_{x_1 \leftarrow x_1'}$ such that blocks before $i_{x_1 \leftarrow x_1'}$ disappear and reappear elsewhere, and blocks after $i_{x_1 \leftarrow x_1'}$ may shrink but do not disappear.

  Let $x_1, x_1'$ be two elements of $D_{p_1}$ with $x_1 < x_1'$.
  Define $i_{x_1 \leftarrow x_1'}$ as the unique index $i_{x_1 \leftarrow x_1'} \in [r]$ that satisfies that, for every $j \in [r]$, we have $B_1^j(x_1') \subseteq B_1^j(x_1)$ if and only if $\pi_1^{-1}(j) \leq i_{x_1 \leftarrow x_1'}$.

\begin{claim}\label{cl5}
  Index $i_{x_1 \leftarrow x_1'}$ is well-defined.
\end{claim}
\begin{proof}
Let $j \in [r]$ be such that $B_1^j(x_1') \subseteq B_1^{j'}(x_1)$ for some $j' \neq j$
and let $j'' \in [r]$ be such that $\pi_1^{-1}(j) < \pi_1^{-1}(j'')$.
Note that $B_1^j(x_1') \cap B_1^j(x_1) = \emptyset$ and that it suffices to show that also $B_1^{j''}(x_1') \cap B_1^{j''}(x_1) = \emptyset$.

Since $B_1^j(x_1') \subseteq B_1^{j'}(x_1)$, by \cref{cl1}, $\leadact_1^\rightarrow(\leader_1^j(x_1)) < x_1'$.
By definition of $\leadact_1^\rightarrow$ and the discretization properties thus $\leader_1^j(x_1)$ is to the left of $x_1'$.
Since $\pi_1^{-1}(j) < \pi_1^{-1}(j'')$, we have $\leader_1^{j''}(x_1) \leqx \leader_1^j(x_1)$.
Thus $\leadact_1^\rightarrow(\leader_1^{j''}(x_1)) < x_1'$.
Hence, $B_1^{j''}(x_1') \cap B_1^{j''}(x_1) = \emptyset$ as desired.
\cqed\end{proof}

For every $j \in [r]$ and two
elements $x_1, x_1' \in D_{p_1}$ with $x_1 < x'_1$
we define $\alpha_{x_1 \leftarrow x_1'}(j) \in [r]$ as follows. 
Let $\alpha_{x_1 \leftarrow x_1'}(j)$ be the ancestor of $j$ in $T$ that is closest\footnote{That is, this ancestor has the shortest path to~$j$ in $T$. A node is an ancestor of itself, that is, it can happen that $\alpha_{x_1 \leftarrow x_1'}(j) = j$.} to $j$ in~$T$ such that
for at least one block $B_2^\iota$ between $B_1^{\alpha_{x_1 \leftarrow x_1'}(j)}$ and $B_1^{\parent(\alpha_{x_1 \leftarrow x_1'}(j))}$
  the leader of $B_2^\iota(x_2^\leftarrow(x_1'))$ is to the left
  of the $x$-coordinate $x_2^\leftarrow(x_1)$. 
  We put $\alpha_{x_1 \leftarrow x_1'}(j)$ to be the root $\pi_1(1)$ if such an ancestor does not exist.

The intuition behind the notion $\alpha_{x_1 \leftarrow x_1'}(j)$ is the following: If we slide $p_1$ from $x_1$ to the right to $x_1'$, then
the $j$-th block $B_1^j$ at $x'_1$ is a subset of $B_1^{\alpha_{x_1 \leftarrow x_1'}(j)}$ at $x_1$. Furthermore, for every descendant $j_\downarrow$
of $j$,\footnote{Node $j$ is its own descendant.}
$B_1^{j_\downarrow}$ is a subset of $B_1^{\alpha_{x_1 \leftarrow x_1'}(j)}$. 
We now prove this intuition in the next three claims. 
We start with the following intermediate step.
\begin{claim}\label{cl-pre-6}
Let $x_1, x_1'$ be two elements of $D_{p_1}$ with $x_1 < x'_1$ and let $j \in [r]$.
Then
\begin{equation}\label{eq:cl6}
B_1^j(x_1')  \cup B_1^{\alpha_{x_1 \leftarrow x_1'}(j)}(x_1') \subseteq B_1^{\alpha_{x_1 \leftarrow x_1'}(j)}(x_1).
\end{equation}
\end{claim}
\begin{proof}
Let $j = j_1, j_2, \ldots, j_b = \alpha_{x_1 \leftarrow x_1'}(j)$ be the vertices on the path in $T$ from $j$
to $\alpha_{x_1 \leftarrow x_1'}(j)$.
By the definition of~$\alpha_{x_1 \leftarrow x_1'}(j)$, for every $i \in [b - 1]$,
the leaders of all blocks $B_2^\iota(x_2^\leftarrow(x_1'))$ between
$B_1^{j_i}$ and $B_1^{j_{i+1}}$ are to the right of $x_2^\leftarrow(x_1)$,
that is, the blocks $B_2^\iota(x_2^\leftarrow(x_1'))$ are disjoint
with $\points_\situ(x_1, x_2^\leftarrow(x_1))$. 
Thus, there is no point of $W_2$ in the area bounded by $x_2^\leftarrow(x_1)$, the predecessor~$p'_2$ of $p_2$ in $\Xlines$, and the two lines given by the $y$-coordinates of the topmost and bottommost point, respectively, in the blocks $B_1^{j_i}(x_1')$.
This implies that all blocks $B_1^{j_i}(x_1')$, $i \in [b]$, are contained in the 
same block~$B_1^{j^\circ}(x_1)$.

We now show that $\alpha_{x_1 \leftarrow x_1'}(j)$ is the first index~$j'$ in the sequence $\pi_1(1), \pi_1(2), \ldots, \pi_1(r)$ such that $B_1^{j'}(x_1')$ is a subset of $B_1^{j^\circ}(x_1)$.
The claim is immediate if $\alpha_{x_1 \leftarrow x_1'}(j) = \pi_1(1)$, so assume otherwise.
Then, $\alpha_{x_1 \leftarrow x_1'}(j)$ is not the root of $T$ and thus $\parent(\alpha_{x_1 \leftarrow x_1'}(j))$ is defined.
Assume that there exists an index~$j_0$ with $\pi_1^{-1}(j_0) < \pi_1^{-1}(\alpha_{x_1 \leftarrow x_1'}(j))$ such that $B_1^{j_0}(x_1') \subseteq B_1^{j^\circ}(x_1)$.
This implies that for every $B_2^\iota$ between $B_1^{\alpha_{x_1 \leftarrow x_1'}(j)}$ and $B_1^{j_0}$ the leader of $B_2^\iota(x_2^\leftarrow(x_1'))$ is to the left of $x_2^\leftarrow(x_1)$.
If $j_0 = \parent(\alpha_{x_1 \leftarrow x_1'}(j))$, then this is a contradiction to the fact that, by definition of $\alpha_{x_1 \leftarrow x'_1}(j)$, there is a block $B_2^\iota(x_2^\leftarrow(x_1'))$ between $B_1^{j_0}$ and $B_1^{\alpha_{x_1 \leftarrow x_1'(j)}}$ whose leader is to the right of $x_2^\leftarrow(x_1)$.
If $j_0 \neq \parent(\alpha_{x_1 \leftarrow x'_1}(j))$, then it follows that the leftmost of the leaders of blocks of $W_2$ between $B_1^{j_0}$ and $B_1^{\alpha_{x_1 \leftarrow x'_1}(j)}$ is more to the right than the leftmost of the leaders of blocks of $W_2$ between $B_1^{\alpha_{x_1 \leftarrow x'_1}(j)}$ and $B_1^{\parent(\alpha_{x_1 \leftarrow x'_1}(j))}$.
This is a contradiction to the definition of parents.
Thus indeed $\alpha_{x_1 \leftarrow x_1'}(j)$ is the earliest index~$j'$ in the sequence $\pi_1(1), \pi_1(2), \ldots, \pi_1(r)$ such that $B_1^{j'}(x_1')$ is a subset of $B_1^{j^\circ}(x_1)$.

We conclude that the leader (rightmost element) $e$ of $B_1^{j^\circ}(x_1)$
is the leader (rightmost element) of $B_1^{\alpha_{x_1 \leftarrow x_1'}(j)}(x_1')$;
by Claim~\ref{cl1} it implies that $\alpha_{x_1 \leftarrow x_1'}(j) = j^\circ$.
This establishes~\eqref{eq:cl6}.
\cqed\end{proof}

In the next claim, we treat blocks that only shrink (but do not disappear) when $p_1$ slides from $x_1$ to~$x_1'$.
\begin{claim}\label{cl-alpha-stay}
Let $x_1, x_1'$ be two elements of $D_{p_1}$ with $x_1 < x'_1$ and let $j \in [r]$.
Then the following conditions are equivalent:
\begin{enumerate}
\item $B_1^j(x_1') \subseteq B_1^j(x_1)$;
\item $\leader_1^j(x_1) = \leader_1^j(x_1')$;
\item $\alpha_{x_1 \leftarrow x_1'}(j) = j$.
\end{enumerate}
\end{claim}
\begin{proof}
If $B_1^j(x_1') \subseteq B_1^j(x_1)$, then $\leader_1^j(x_1) = \leader_1^j(x_1')$ by Claim~\ref{cl4}.
In the other direction, 
if $\leader_1^j(x_1) = \leader_1^j(x_1')$, then $B_1^j(x_1') \cap B_1^j(x_1) \neq \emptyset$, so Claim~\ref{cl2} implies $B_1^j(x_1') \subseteq B_1^j(x_1)$.

To prove equivalence of the first and third condition, we use~\eqref{eq:cl6} of Claim~\ref{cl-pre-6} which
implies that $B_1^j(x_1') \subseteq B_1^{\alpha_{x_1 \leftarrow x_1'}(j)}(x_1)$.
Since the blocks $B_1^\iota(x_1)$ are disjoint for distinct $\iota \in [r]$, $B_1^j(x_1') \subseteq B_1^j(x_1)$ is equivalent to $\alpha_{x_1 \leftarrow x_1'}(j) = j$.
\cqed\end{proof}

In the last claim we show that, when moving $p_1$ from $x_1$ to $x_1'$, and it is the case that the $j$th block~$B_1^j$ disappears and reappears elsewhere, then $B_1^j(x'_1)$ and all the blocks at position $x'_1$ corresponding to descendants of $j$ in $T$ are contained in $B_1^{\alpha_{x_1 \leftarrow x_1'}(j)}(x_1)$.
The formal statement is as follows.

\begin{figure}[tb]
\begin{center}
\includegraphics{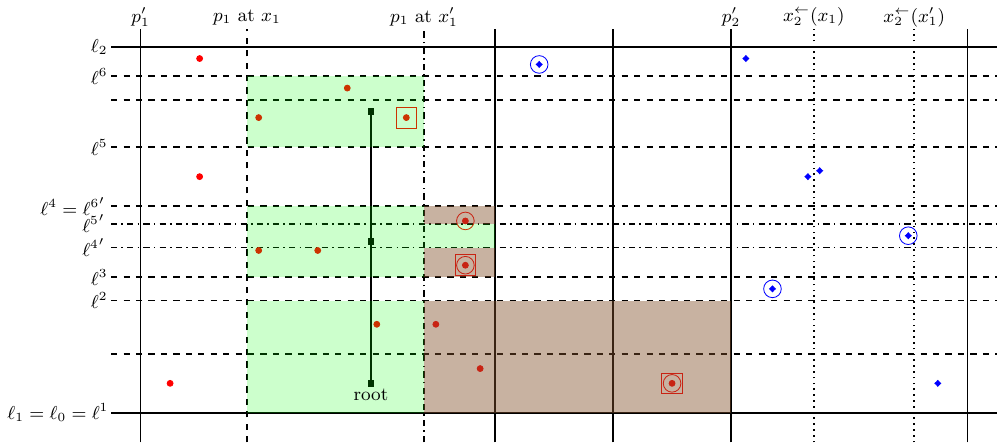}
\caption{
  Situation of Claim~\ref{cl6}, where $p_1$ is either at position $x_1\in D_{p_1}$ or at position $x_1'\in D_{p_1}$ for $x<x'$.
  The lines of $\stlines_\situ$ for $x_1$ are denoted with $\ell^i$, $1 \leq i \leq 6$.
  The lines of $\stlines_\situ$ for $x_1'$ are denoted with $\ell'^{i}$, $1 \leq i \leq 6$.
  Blocks given by positioning $p_1$ at $x_1'$ ($B_1^1(x_1'),B_1^2(x_1'),B_1^3(x_1')$) are depicted by purple color with circled leaders and blocks given by positioning $p_1$ at $x_1$ ($B_1^1(x_1),B_1^2(x_1),B_1^3(x_1)$) are depicted by the union of green and purple color with squared leaders.
  Blocks given by positioning $p_2$ at $x_2^\leftarrow(x_1')$
  with circled leaders.
 Observe that $B_1^3(x_1')\not\subseteq B_1^3(x_1)$ and $\alpha_{x_1 \leftarrow x_1'}(3)=2$.
 Then $B_1^3(x_1')\cup B_1^2(x_1') \subseteq B_1^2(x_1)$.
}\label{fig:sliding2}
\end{center}
\end{figure}

  \begin{sidewaysfigure}
\begin{center}
\includegraphics{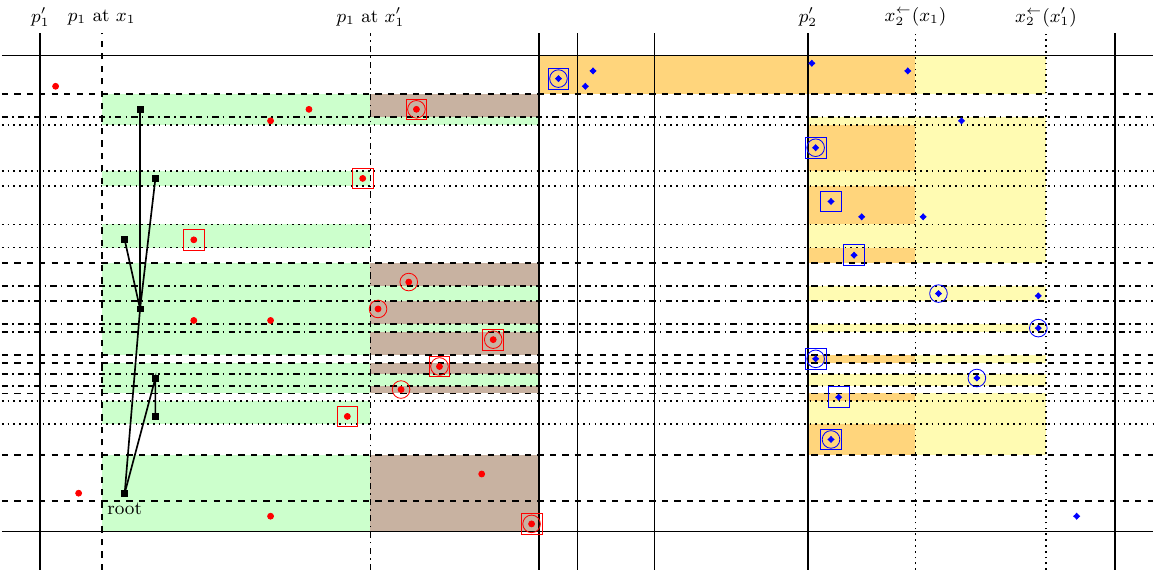}
\caption{
More complex example of situation in Claim~\ref{cl6}, where $p_1$ is either at position $x_1\in D_{p_1}$ or at position $x_1'\in D_{p_1}$ for $x<x'$.
  The lines of $\stlines_\situ$ for $x_1$ are denoted with dashed and dotted lines.
  The lines of $\stlines_\situ$ for $x_1'$ are denoted with dashed and dash-dotted lines.
  Blocks given by positioning $p_1$ at $x_1'$ are depicted by purple color with circled leaders and blocks given by positioning $p_1$ at $x_1$ are depicted by the union of green and purple color with squared leaders.
  Blocks given by positioning $p_2$ at $x_2^\leftarrow(x_1)$ are depicted by orange color with squared leaders and blocks given by positioning $p_1$ at $x_2^\leftarrow(x_1')$ are depicted by the union of yellow and orange color with circled leaders.
  An auxiliary rooted tree $T$ for red blocks is also visualized.
}\label{fig:sliding2advanced}
\end{center}
\end{sidewaysfigure}

\begin{claim}\label{cl6}
Let $x_1, x_1'$ be two elements of $D_{p_1}$ with $x_1 < x'_1$ and let $j \in [r]$
be such that $B_1^{j}(x_1') \not \subseteq B_1^{j}(x_1)$. 
Then for every descendant
$j_\downarrow$ of $j$ in $T$ we have
$B_1^{j_\downarrow}(x_1') \subseteq B_1^{\alpha_{x_1 \leftarrow x_1'}(j)}(x_1)$. 
\end{claim}
\begin{proof}
Since, $B_1^{j}(x'_1) \not\subseteq B_1^{j}(x_1)$, we have $\pi_1^{-1}(j) > i_{x_1 \leftarrow x_1'}$. 
For every $j_\downarrow \in V(T_{j}) \setminus \{j\}$ we have
$\pi_1^{-1}(j_\downarrow) > \pi_1^{-1}(j)$ and thus $\pi_1^{-1}(j_\downarrow) > i_{x_1 \leftarrow x_1'}$ as well.
In particular, $\leader_1^{j_\downarrow}(x_1)$ is to the left of $x_1'$
and hence $B_1^{j_\downarrow}(x_1) \cap B_1^{j_\downarrow}(x_1') = \emptyset$.
Hence, no vertex $j'$ on the path in $T$ between $j_\downarrow$ and $j$ (including $j_\downarrow$ and $j$) satisfies
 $B_1^{j'}(x_1') \subseteq B_1^{j'}(x_1)$.
Thus, applying \eqref{eq:cl6} in Claim~\ref{cl-pre-6} to $j_\downarrow$ (instead of $j$), 
we obtain that $\alpha_{x_1 \leftarrow x_1'}(j_\downarrow) \neq j'$.
Thus, $\alpha_{x_1 \leftarrow x_1'}(j_\downarrow)$ is an ancestor of~$j$.
By definition of $\alpha_{x_1 \leftarrow x_1'}(j)$ we conclude that $\alpha_{x_1 \leftarrow x_1'}(j_\downarrow) = \alpha_{x_1 \leftarrow x_1'}(j)$.
Now applying \eqref{eq:cl6} in Claim~\ref{cl-pre-6} to $j_\downarrow$ again, we have $B_1^{j_\downarrow}(x_1') \subseteq B_1^{\alpha_{x_1 \leftarrow x_1'}(j)}(x_1)$.
This finishes the proof of the claim.
\cqed\end{proof}

The above claims establish the following structure: If we swipe the value of $x_1 \in D_{p_1}$
from right to left, and focus on one block $B_1^j(x_1)$, then a particular element
$e$ is a leader of $B_1^j(x_1)$ between $\leadact_1^\rightarrow(e)$, which is the rightmost
value $x_1$ that is to the left of $e$, and $\leadact_1^\leftarrow(e)$;
for every $x_1 < \leadact_1^\leftarrow(e)$, the element $e$ and the whole block
$B_1^j(\leadact_1^\leftarrow(e))$ is a subset of some other block $B_1^{j'}(x_1)$ for
an ancestor $j'$ of $j$ in the tree $T$.
Furthermore, $B_1^{j_\downarrow}(\leadact_1^\leftarrow(e))$ is also a subset
of $B_1^{j'}(x_1)$ for every $j_\downarrow \in V(T_j)$. 

For a block $B_1^j$ and an element $e$ that is the leader of $B_1^j(x_1)$ for some $x_1 \in D_{p_1}$,
the \emph{epoch} of $B_1^j$ and $e$ is the segment $[\leadact_1^\leftarrow(e), \leadact_1^\rightarrow(e)]$
in $D_{p_1}$. Note that each block $B_1^j$ partitions $D_{p_1}$ into epochs; let $\mathcal{P}_j$ be this
partition.
Note that the epochs one-to-one correspond to the intervals $[\leadact_1^\leftarrow(e), \leadact_1^\rightarrow(e)]$ where $e$ is the leader of some block in $W_1$ for some $x_1$.
Moreover, $e$ is unique to this interval.
Hence, for an epoch~$\epoch$, we may use the notation 
$\epoch = [\leadact_1^\leftarrow(\epoch), \leadact_1^\rightarrow(\epoch)]$ without ambiguity. 

We now make several observations about the structure of epochs.
Let $x_1, x'_1 \in D_{p_1}$ with $x_1 < x_1'$.
Claims~\ref{cl3} and~\ref{cl4} ensure that if $x_1, x_1'$ belong to different epochs of~$B_1^j$, then $B_1^j(x_1) \cap B_1^j(x_1') = \emptyset$, and if $x_1, x_1'$ belong to the same epoch of $B_1^j$, then $B_1^j(x_1') \subseteq B_1^j(x_1)$ and $\leader_1^j(x_1) = \leader_1^j(x_1')$.
Claims~\ref{cl-alpha-stay} and~\ref{cl6} ensure that, if $j'$ is an ancestor of $j$ in $T$, then the epochs of~$B_1^{j'}$ are supersets of the epochs of $B_1^j$, that is, the epochs of $B_1^{j'}$ form a coarser partition of $D_{p_1}$ into segments than the epochs of~$B_1^j$.
To see this, consider two distinct epochs $\epoch, \epoch'$ of $B_1^{j'}$ where $\epoch$ is to the left of $\epoch'$ and observe that the leader of $B_1^{j'}$ is different in these two epochs.
It then suffices to show that also the leader of $B_1^{j}(x_1)$, $x_1 \in \epoch_1$, is different from the leader of $B_1^{j}(x_1')$, $x'_1 \in \epoch_2$.
By Claim~\ref{cl-alpha-stay} we have $B_1^{j'}(x'_1) \not\subseteq B_1^{j'}(x_1)$.
By Claim~\ref{cl6} thus $B_1^j(x'_1) \subseteq B_1^{\alpha_{x_1 \leftarrow x_1'}(j')}(x_1)$, that is, $B_1^j(x'_1) \not\subseteq B_1^j(x_1)$.
By Claim~\ref{cl-alpha-stay} thus the leader $B_1^{j}(x_1)$ is different from $B_1^{j}(x'_1)$.
Finally, observe also that $B_1^{\pi(1)}$ has only one epoch, because $\leader_1^{\pi(1)}(x_1)$ is the rightmost point of $\apxcell(p_1', \ell_1)$ and thus stays constant for all $x_1 \in D_{p_1}$. 

For an epoch $\epoch$ of a block $B_1^j$, we denote by $\leadact_1^\downarrow(\epoch)$
and $\leadact_1^\uparrow(\epoch)$
the minimum and maximum $y$\nobreakdash-coordinate of an element of $B_1^{j_\downarrow}(x_1)$ for $x_1 \in \epoch$
and $j_\downarrow \in V(T_j)$. 
Note that the minimum and maximum values are always attained for 
$x_1 = \leadact_1^\leftarrow(\epoch)$, as the union of all blocks $B_1^{j_\downarrow}(x_1)$
for $j_\downarrow \in V(T_j)$ only grows (in the subset order) as $x_1$ decreases from $\leadact_1^\rightarrow(\epoch)$
to $\leadact_1^\leftarrow(\epoch)$.

By definition, if $j'$ is an ancestor of $j$ in~$T$ and $\epoch'$ is an epoch of $B_1^{j'}$
that contains an epoch $\epoch$ of $j$, then
$$[\leadact_1^\downarrow(\epoch), \leadact_1^\uparrow(\epoch)] \subseteq
  [\leadact_1^\downarrow(\epoch'), \leadact_1^\uparrow(\epoch')].$$
Thus, with an epoch $\epoch$ one can associate a rectangle in $\mathbb{R}^2$:
$$[\leadact_1^\leftarrow(\epoch), \leadact_1^\rightarrow(\epoch)] \times
[\leadact_1^\downarrow(\epoch), \leadact_1^\uparrow(\epoch)],$$
and we have that the rectangle of an epoch of a block $B_1^j$ is contained in the 
rectangle of a corresponding epoch of a block $B_1^{j'}$ for an ancestor $j'$ of $j$.

\begin{figure}[tb]
\begin{center}
\includegraphics[scale=0.5]{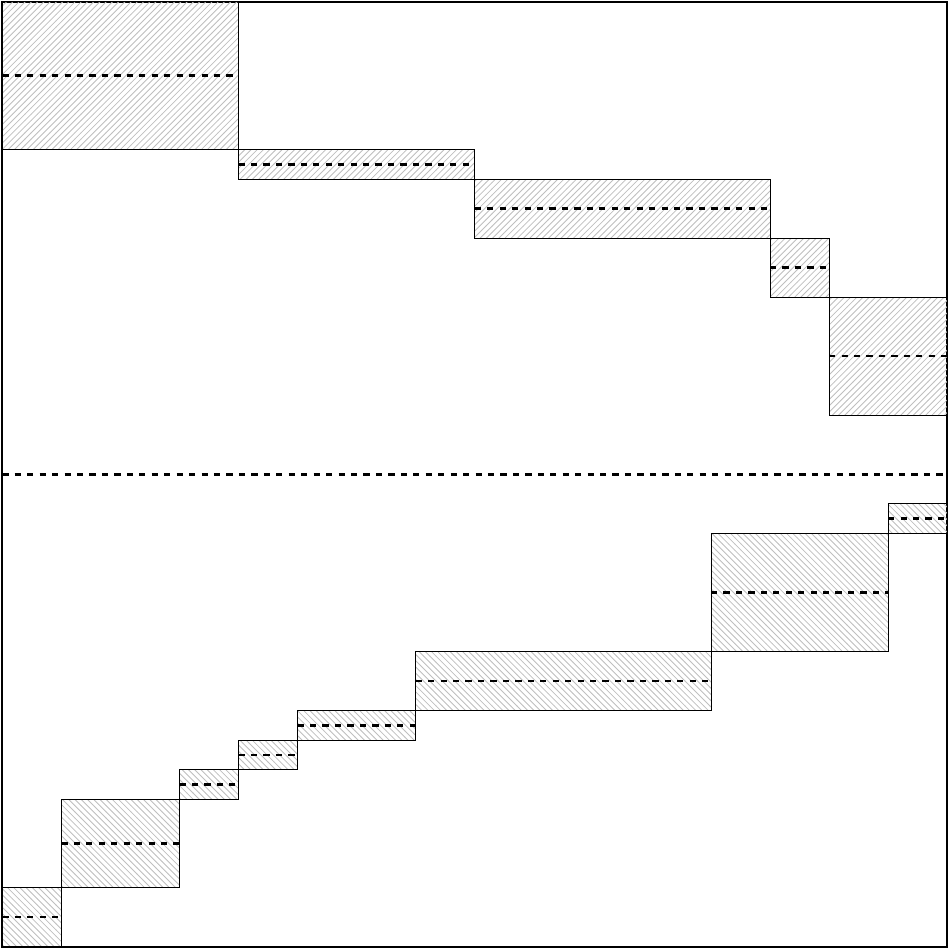}
\caption{Statement of Claim~\ref{cl7}.
  The big rectangle is the box
$[\leadact_1^\leftarrow(\epoch), \leadact_1^\rightarrow(\epoch)] \times
[\leadact_1^\downarrow(\epoch), \leadact_1^\uparrow(\epoch)]$
for one epoch of a block $B_1^j$ with two children $j_1 > j$ and $j_2 < j$.
  The small boxes above with north east lines correspond to epochs of $B_1^{j_1}$
  and the small boxes below with north west lines correspond to epochs of $B_1^{j_2}$.
  The horizontal dashed lines indicate the $y$-coordinate of the leader
  at the corresponding epoch.}\label{fig:epochs}
\end{center}
\end{figure}

Claim~\ref{cl6} implies that, for two different epochs $\epoch$ and $\epoch'$ of the 
same block $B_1^j$, the segments
$[\leadact_1^\downarrow(\epoch), \leadact_1^\uparrow(\epoch)]$ and
$[\leadact_1^\downarrow(\epoch'), \leadact_1^\uparrow(\epoch')]$ are disjoint.
(Note that, since the leader of $B_1^j$ changes between the two epochs, Claim~\ref{cl6} implies that the leader of $B_1^{j_\downarrow}$ changes for each descendant~$j_\downarrow$ of~$j$, that is, the corresponding blocks are disjoint.)
We now observe that, moreover, for a right-to-left sequence of epochs of some block, these segments are ordered top-to-bottom or vice versa (see also Figure~\ref{fig:epochs}):
\begin{claim}\label{cl7}
Let $j'$ be the parent of $j$ in $T$,
let $\epoch'$ be an epoch of $j'$, and let $\epoch_1$, $\epoch_2$, \ldots, $\epoch_b$
be the epochs of $j$ contained in $\epoch'$ in the right-to-left order.
Then the sequence of disjoint segments
$([\leadact_1^\downarrow(\epoch_i), \leadact_1^\uparrow(\epoch_i)])_{i=1}^a$ 
is monotonous, that is, if $j < j'$ then
$$\leadact_1^\downarrow(\epoch_1) \geq \leadact_1^\uparrow(\epoch_1) > \leadact_1^\downarrow(\epoch_2) \geq \leadact_1^\uparrow(\epoch_2) > \ldots > \leadact_1^\downarrow(\epoch_b) \geq \leadact_1^\uparrow(\epoch_b),$$
and if $j' < j$ then
$$\leadact_1^\downarrow(\epoch_1) \leq \leadact_1^\uparrow(\epoch_1) < \leadact_1^\downarrow(\epoch_2) \leq \leadact_1^\uparrow(\epoch_2) < \ldots < \leadact_1^\downarrow(\epoch_b) \leq \leadact_1^\uparrow(\epoch_b).$$
\end{claim}
\begin{proof}
  \looseness=-1
  Recall that through the entire epoch $\epoch'$ the leader of the block $B_1^{j'}$ stays 
the same: if $x_1 < x_1'$ for $x_1,x_1' \in \epoch'$ then $\leader_1^{j'}(x_1) = \leader_1^{j'}(x_1')$
and $B_1^{j'}(x_1') \subseteq B_1^{j'}(x_1)$. 
Claim~\ref{cl-pre-6} implies moreover that $\alpha_{x_1 \leftarrow x_1'}(j') = j'$. 

Fix $\beta \in [b - 1]$, $x_1' \in \epoch_\beta$, and $x_1 \in \epoch_{\beta+1}$. 
We have $B_1^j(x_1') \not \subseteq B_1^j(x_1)$. 
Claim~\ref{cl-alpha-stay} implies that $j \neq \alpha_{x_1 \leftarrow x_1'}(j)$. 
By Claim~\ref{cl6} applied to $B_1^j$, $x_1$, and $x_1'$ we infer
that $\alpha_{x_1 \leftarrow x_1'}(j) = j'$ as $j'$ is the parent of $j$ and $B_1^{j'}(x_1') \subseteq B_1^{j'}(x_1)$. 
Furthermore, we have $B_1^{j_\downarrow}(x_1') \subseteq B_1^{j'}(x_1)$ for every $j_\downarrow \in V(T_j)$.
On the other hand, $B_1^j(x_1)$ is below $B_1^{j'}(x_1)$ if $j < j'$ and
above $B_1^{j'}(x_1)$ if $j > j'$. Hence,
      $\leadact_1^\uparrow(\epoch_{\beta+1}) < \leadact_1^\downarrow(\epoch_\beta)$ if $j < j'$
and   $\leadact_1^\downarrow(\epoch_{\beta+1}) > \leadact_1^\uparrow(\epoch_\beta)$ if $j > j'$.
This finishes the proof of the claim.
\cqed\end{proof}

Claim~\ref{cl7} allows us to conclude the proof of Lemma~\ref{lem:leaders} as follows
using the setting of Section~\ref{ss:tree}.

We construct a tree $T'$ from $T$ by appending to every $j \in V(T) = [r]$
two new children $v^j$ and $u^j$ (which are leaves of $T'$).
With every node $j \in [r]$ we associate the segment partition $\mathcal{P}_j$ of $D_{p_1}$ into epochs of $B_1^j$
and with every leaf of $T'$ we associate the most refined segment partition of $D_{p_1}$ with only singletons.
For every non-root node $j \in V(T)$ we define $\flf(j) = \typeinc$ if
$j < \parent(j)$ and $\flf(j) = \typedec$ if $j > \parent(j)$
We also define $\flf(v^j) = \typedec$ and $\flf(u^j) = \typeinc$. 
This makes $\mathbb{T} = ((D_{p_1}, \leq), T', (\mathcal{P}_v)_{v \in V(T')}, \flf)$
a tree of segment partitions.

Now define for every $j \in [r]$ functions $f_{v^j} = \leadup_j$
and $f_{u^j} = \leaddown_j$. 
Observe that within one epoch of $B_1^j$, $\leadup_j$ is nonincreasing
and $\leaddown_j$ is nondecreasing. Consequently, Claim~\ref{cl7} (together with the fact that $B_1^{\pi(1)}$ has only one epoch)
implies that the family of functions $\mathcal{F} = (f_v)_{v \in \leaves(T')}$ is a family of leaf functions
for the tree of segment partitions $\mathbb{T}$.

We apply Lemma~\ref{lem:make-seg-rep} to $\mathbb{T}$ and $\mathcal{F}$ and
obtain a family $\mathcal{G} = (g_v)_{v \in V(T') \setminus \{\treeroot(T)\}}$
of segment reversions and a family
$\widehat{\mathcal{F}} = (\hat{f}_v)_{v \in \leaves(T)}$ of nondecreasing functions.
By Lemma~\ref{lem:make-seg-rep}, we can return $T'$, $\mathcal{G}$, and $\widehat{\mathcal{F}}$
as outcomes of Lemma~\ref{lem:leaders}.
\end{proof}

\subsection{Completeness}\label[section]{sec:completeness}
We now perform a tedious but rather direct check that shows that all defined constraints and steps where we filtered out domains of some lines guarantee completeness.
\begin{lemma}\label{lem:completeness}
If an assignment $(\Xoptsol,\Yoptsol)$ that assigns to every line $\ell$ an element in $D_\ell$ satisfies all
monotonicity, corner, alternation, and alternating lines constraints, then the pair $(\{\Xoptsol(\ell) \mid \ell \in \Xoptlines\}, \{\Yoptsol(\ell) \mid \ell \in \Yoptlines\})$ is a separation.
\end{lemma}
\begin{proof}
  The proof is by contradiction.
  Assume that there exist two points $(x_1,y_1) \in W_1$ and $(x_2,y_2) \in W_2$ such that
no element of $X' := \{\Xoptsol(\ell) \mid \ell \in \Xoptlines\}$ is between $x_1$ and $x_2$
and no element of $Y' := \{\Yoptsol(\ell) \mid \ell \in \Yoptlines\}$ is between $y_1$ and $y_2$. 
Our goal is to obtain a contradiction by exhibiting either a violated constraint
or an element $\Xoptsol(\ell)$ or $\Yoptsol(\ell)$ of some domain $D_\ell$ that should have been removed in one of the filtering steps.

Let $\Xsol = \Xoptsol \cup \Xapxsol$ and $\Ysol = \Yoptsol \cup \Yapxsol$.
Observe that the choice of the domains and the monotonicity constraints ensure that
$\Xsol$ and $\Ysol$ are both increasing functions.

Let $p_1 \in \Xlines$ be such that 
$\Xsol(p_1)$ is the maximum element of $X' \cup \{1\}$ that is smaller than $x_1$ and $x_2$ and let
$p_2 \in \Ylines$ be such that $\Xsol(p_2)$ is the successor of $\Xsol(p_1)$ in $X' \cup \{1, 3n+1\}$. Note that $\Xsol(p_2) > x_1,x_2$.
Similarly, let $q_1 \in \Ylines$ be such that 
$\Ysol(q_1)$ be the maximum element of $Y' \cup \{1\}$ that is smaller than $y_1$ and $y_2$
and let $q_2 \in \Ylines$ be such that 
$\Ysol(q_2)$ be the successor of $\Ysol(q_1)$ in $Y' \cup \{1,3n+1\}$. Again, $\Ysol(q_2) > y_1,y_2$.

\begin{figure}[tb]
\begin{center}
\includegraphics{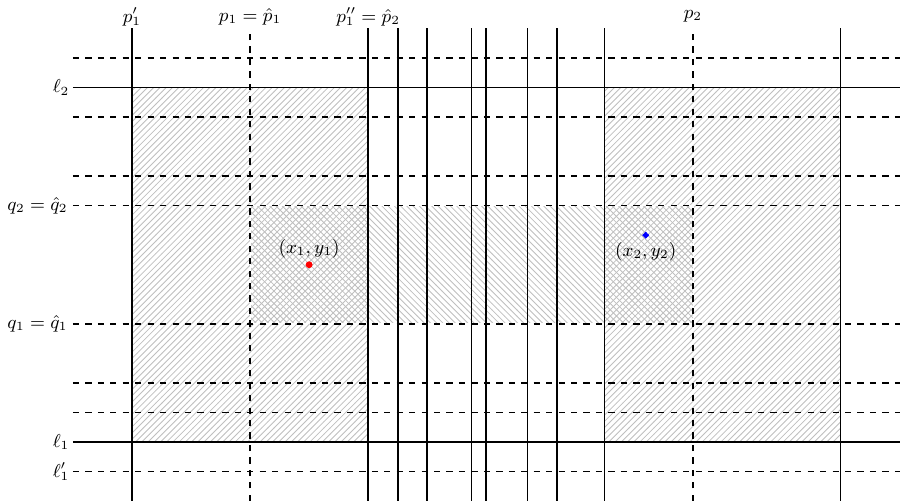}
\caption{Illustration of the proof of Lemma~\ref{lem:completeness}.
Solid lines are from $\Xapxlines \cup \Yapxlines$, dashed lines are from
  $\Xoptlines \cup \Yoptlines$. 
The \apx-supercells $\apxcell(p_1', \ell_1)$ and $\apxcell(p_2', \ell_1)$ and
the \opt-supercell $\optcell(p_1,q_1)$ are highlighted.}\label{fig:completeness}
\end{center}
\end{figure}

By symmetry, we can assume that $\cont(\optcell(p_1, q_1)) \neq 1$.
Let $p_1' \in \Xapxlines$ and $\ell_1 \in \Yapxlines$ such that $\apxcell(p_1', \ell_1)$ is the \apx-supercell containing $(x_1,y_1)$.
Note that $\cont(\apxcell(p_1', \ell_1)) = 1$. 
Let $p_1''$ be the successor of $p_1'$ in $\Xapxlines$ and let $\ell_2$ be the successor of $\ell_1$ in $\Yapxlines$.
Let $\hat{p}_1$ be the maximum of the pair $\{p_1', p_1\}$,
    $\hat{p}_2$ be the minimum of the pair $\{p_1'', p_2\}$, 
    $\hat{\ell}_1$ be the maximum of the pair $\{q_1, \ell_1\}$, and
    $\hat{\ell}_2$ be the minimum of the pair $\{q_2, \ell_2\}$. 
Note that $\Xsol(\hat{p}_1) < x_1 < \Xsol(\hat{p}_2)$ and $\Ysol(\hat{\ell}_1) < y_1 < \Ysol(\hat{\ell}_2)$. 
Consult Figure~\ref{fig:completeness} for an example of such situation.

\begin{claim}\label{cl:corner}
Exactly one of the lines $\hat{p}_1$, $\hat{p}_2$, $\hat{\ell}_1$, and $\hat{\ell}_2$ lies in $\Xapxlines \cup \Yapxlines$.
\end{claim}
\begin{proof}
First, we exclude the case when $\hat{p}_1 = p_1$, $\hat{p}_2 = p_2$, $\hat{\ell}_1 = q_1$, and $\hat{\ell}_2 = q_2$. 
If this were the case, then both $(x_1,y_1)$ and $(x_2,y_2)$ would lie in the \apx-supercell $\apxcell(p_1', \ell_1)$, contradicting the fact that $(\Xapx, \Yapx)$ is a separation.
(Recall that $(\Xapx, \Yapx)$ is the initially computed 2-approximate separation.)

Consider now the case when at least two of the lines $\hat{p}_1$, $\hat{p}_2$, $\hat{\ell}_1$, $\hat{\ell}_2$ belong to $\Xapxlines \cup \Yapxlines$.
Then the area of interest of the tuple $(\hat{p}_1, \hat{p}_2, \hat{\ell}_1, \hat{\ell}_2)$ lies inside the \apx-supercell $\apxcell(p_1',\ell_1)$
and also inside the \opt-supercell $\optcell(p_1,q_1)$.
Consider the abstract cell~$C$ corresponding to $(\hat{p}_1, \hat{p}_2, \hat{\ell}_1, \hat{\ell}_2)$.
Since $\cont(\apxcell(p_1', \ell_1)) = 1$, we have $\delta(C) \in \{0, 1\}$ by definition of $\delta$.
Since $\cont(\optcell(p_1, q_1)) \neq 1$ furthermore $\delta(C) = 0$ (again, by definition of $\delta$).
Thus, the tuple $(\hat{p}_1, \hat{p}_2, \hat{\ell}_1, \hat{\ell}_2)$ is an empty corner. 
However, then the presence of $(x_1,y_1)$ in the area of interest of the tuple $(\hat{p}_1, \hat{p}_2, \hat{\ell}_1, \hat{\ell}_2)$ is a contradiction
as it violates the corner constraint or the corner filtering step for the empty corner in question.
\cqed\end{proof}

Thus, exactly one of the lines $\hat{p}_1$, $\hat{p}_2$, $\hat{\ell}_1$, and $\hat{\ell}_2$ lies in $\Xapxlines \cup \Yapxlines$.
We claim that, by symmetry, we may assume without loss of generality that this is $\hat{p}_2$.
That is, $\hat{p}_1 = p_1$, $\hat{p}_2 = p_1''$, $\hat{\ell}_1 = q_1$, and $\hat{\ell}_2 = q_2$.
In other words, $p_1' < p_1 < p_1'' < p_2$ and $\ell_1 < q_1 < q_2 < \ell_2$. 
Indeed, to see that the above symmetry-breaking assumption is without loss of generality, we may use the fact that the addition of alternation constraints and alternating lines constraints as well as filtering of correct orders of extremal points has been performed
both in top/down and left/right directions, and that in the following we will solely rely on these filtering steps and constraints.
Observe that now $\situ = (p_1,p_2,\ell_1,\ell_2)$ is a \emph{situation}.

Let $p_2'$ be the maximum element of $\Xapxlines$ that is smaller than $p_2$.
Recall that $\btlines_\situ$ is the set of lines of $\Yoptlines$ between $\ell_1$ and $\ell_2$, 
       $\alllines_\situ = \btlines_\situ \cup \{\ell_1\}$,
      and $\area(\ell) = \area(p_1, p_2, \ell, \ell')$ for every $\ell \in \alllines_\situ$ where $\ell'$ is the successor of $\ell$ in $\btlines_\situ \cup \{\ell_1,\ell_2\}$.

\begin{claim}\label{cl:ones-up-down}
There exists an element $\ell \in \Yoptlines$, $q_2 \leq \ell < \ell_2$, such that $\cont(\optcell(p_1, \ell)) = 1$.
Similarly, there exists $\ell \in \Yoptlines$ with $\ell_1 \leq \ell < q_1$ such that $\cont(\optcell(p_1, \ell)) = 1$.
\end{claim}
\begin{proof}
  \looseness=-1
  We show only the first claim, the proof for the second one is analogous.
Assume the contrary. 
Then, as $\cont(\apxcell(p_1', \ell_1)) = 1$
while for every $\ell \in \Yoptlines$ with $q_1 \leq \ell < \ell_2$ we have $\optcell(p_1, \ell)) \neq 1$, 
every cell $\cell$ in the area of interest of the tuple $(p_1, p_1'', q_1, \ell_2)$ satisfies $\cont(\cell) = 0$. 
Hence, the tuple $(p_1,p_1'', q_1,\ell_2)$ is an empty corner.
However, the existence of $(x_1,y_1)$ violates the corner filtering or the corner constraint for that tuple. 
\cqed\end{proof}

\begin{claim}\label{cl:ones-up-down-sym}
Assume $\cont(\optcell(p_1, q_1)) = 0$.
Then, $(x_2,y_2)$ lies in the \apx-supercell $\apxcell(p_2', \ell_1)$ and, consequently, $\cont(\apxcell(p_2', \ell_1)) = 2$. 
Furthermore, 
there exists $\ell \in \Yoptlines$ with $q_1 \leq \ell < \ell_2$ such that $\cont(\optcell(p_1, \ell)) = 2$ and
that there exists $\ell \in \Yoptlines$ with $\ell_1 \leq \ell < q_1$ such that $\cont(\optcell(p_1, \ell)) = 2$.
\end{claim}
\begin{proof}
Recall that $(x_2,y_2)$ lies in the \opt-supercell $\optcell(p_1,q_1)$.
Since $p_1' < p_1 < p_1'' < p_2$ and $\ell_1 < q_1 < q_2 < \ell_2$, the \apx-supercells that 
share cells with $\optcell(p_1,q_1)$ are the cells $\apxcell(r, \ell_1)$ for $p_1' \leq r \leq p_2'$. 
Assume $(x_2,y_2)$ lies in $\apxcell(r, \ell_1)$ for some $p_1' \leq r \leq p_2'$. 

Since $(x_2,y_2) \in W_2$, point $(x_2,y_2)$ does not lie in the \apx-supercell $\apxcell(p_1', \ell_1)$, so $r \neq p_1'$. 
If $r < p_2'$, then consider the tuple $(r,r', q_1,q_2)$ where $r'$ is the successor of $r$ in $\Xapxlines$. 
Observe that the area of interest of that tuple is contained in the \apx-supercell $\apxcell(r,\ell_1)$ and in the \opt-supercell $\optcell(p_1,q_1)$
and contains $(x_2,y_2)$.
Since $(x_2,y_2)$ is in that \apx-supercell, $\cont(\apxcell(r,\ell_1)) = 2$. Since $\cont(\optcell(p_1,q_1)) = 0$, 
for every cell $\cell$ in the area of interest of $(r,r',q_1,q_2)$ we have $\cont(\cell) = 0$. Since $r,r' \in \Xapxlines$, 
it follows that the tuple~$(r,r',q_1,q_2)$ is an empty corner. Since $q_1,q_2 \in \Yoptlines$, a corner constraint has been introduced binding $q_1$ and $q_2$
and the existence of $(x_2,y_2)$ in the area of interest of $(r,r',q_1,q_2)$ violates that constaint. This establishes $r = p_2'$, that is,
$(x_2,y_2)$ lies in $\apxcell(p_2',\ell_1)$ and, consequently, $\cont(\apxcell(p_2',\ell_1)) = 2$.

For the second statement of the claim, we essentially repeat the reasoning of Claim~\ref{cl:ones-up-down}. 
Assume that for every $\ell \in \Yoptlines$ with $q_1 \leq \ell < \ell_2$ we have $\cont(\optcell(p_1,\ell)) \neq 2$ (the second case is analogous).
Consider the tuple $(p_2', q_1, p_2, \ell_2)$.
Observe that its area of interest is contained both in the union of \opt-supercells $\optcell(p_1,\ell)$ for $q_1 \leq \ell < \ell_2$ and
in the \apx-supercell $\apxcell(p_2',\ell_1)$. Since $\cont(\apxcell(p_2', \ell_1)) = 2$, for every cell $\cell$ in the area of interest of $(p_2',q_1,p_2,\ell_2)$
we have $\cont(\cell) = 0$. Since $p_2' \in \Xapxlines$ and $\ell_2 \in \Yapxlines$, $(p_2',q_1,p_2,\ell_2)$ is an empty corner and a corner constraint has
been introduced binding $q_1$ and $p_2$. However, the existence of $(x_2,y_2)$ in the area of interest of this empty corner violates this corner constraint. 
This finishes the proof of the claim.
\cqed\end{proof}

\begin{claim}\label{cl:alt4}
There exists 
\begin{enumerate}
\item a line $\ell \in \Yoptlines$ with $q_2 \leq \ell < \ell_1$ and $\cont(\optcell(p_1, \ell)) = 1$;
\item a line $\ell \in \Yoptlines$ with $q_1 \leq \ell < \ell_1$ and $\cont(\optcell(p_1, \ell)) = 2$;
\item a line $\ell \in \Yoptlines$ with $\ell_1' \leq \ell < q_1$ and $\cont(\optcell(p_1, \ell)) = 1$;
\item a line $\ell \in \Yoptlines$ with $\ell_1 \leq \ell \leq q_1$ and $\cont(\optcell(p_1, \ell)) = 2$.
\end{enumerate}
In particular, the alternation of the situation $\situ = (p_1, p_2, \ell_1, \ell_2)$ is at least four.
\end{claim}
\begin{proof}
The first part of the claim (the existence of the lines) implies  
that the alternation of $\situ$ is at least three.
As an alternation of a situation cannot be an odd integer larger than $1$ (cf.\ Lemma~\ref{lem:alternations}),
we infer that the first part of the claim implies the second one. Thus we are left with proving the first part. 

We first invoke Claim~\ref{cl:ones-up-down}, giving the first and third point.
If $\cont(\optcell(p_1, q_1)) = 2$, then we are done.
In the other case $\cont(\optcell(p_1, q_1)) = 0$, we invoke Claim~\ref{cl:ones-up-down-sym},
   obtaining the lines promised in the second and fourth point.
\cqed\end{proof}

\begin{figure}[tb]
\begin{center}
\includegraphics{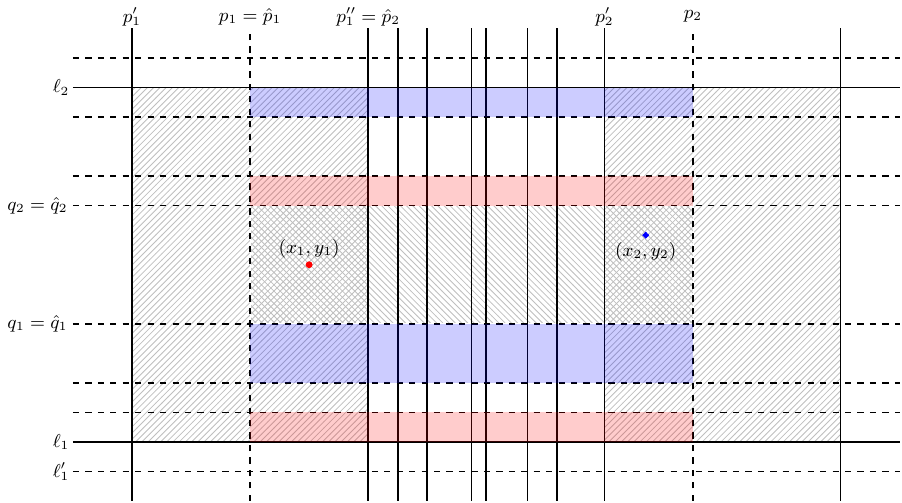}
\caption{Illustration of the proof of Lemma~\ref{lem:completeness}, directly after the proof of Claim~\ref{cl:alt4}.
Solid lines are from $\Xapxlines \cup \Yapxlines$, dashed lines are from
  $\Xoptlines \cup \Yoptlines$. 
The \apx-supercells $\apxcell(p_1', \ell_1)$ and $\apxcell(p_2', \ell_1)$ and
the \opt-supercell $\optcell(p_1,q_1)$ are highlighted.
$\cont(\optcell(p_1, \ell)) = 1$ are highlighted by red background and $\cont(\optcell(p_1, \ell)) = 2$ by blue background.
}\label{fig:completeness2}
\end{center}
\end{figure}

See Figure~\ref{fig:completeness2} which shows the situation in Claim~\ref{cl:alt4}.
By Claim~\ref{cl:alt4}, for the situation $\situ$ an alternation constraint has been added, a filtering step for correct orders of extremal points has been performed,
and a number of alternating lines constraints have been added. 

If the alternation constraint the situation $\situ$ is violated, then we have our desired contradiction.
Otherwise, $(\Xoptsol(p_1), \Xoptsol(p_2))$ fits the situation~$\situ$, that is, $\stseq_\situ = \stseq(\Xoptsol(p_1), \Xoptsol(p_2))$. 

Let $\ell^1,\ell^2, \ldots, \ell^{2r}$ be the elements of $\stlines_\situ$ in the increasing order.
By Claim~\ref{cl:alt4}, $\ell^2 \leq q_1 < \ell^{2r}$. Let $i \in \mathbb{N}$, $1 < i < 2r$, be the maximum index
with $\ell^i \leq q_1$. Observe that $\Yoptsol(\ell^i) < y_1, y_2$ and $y_1, y_2 < \Yoptsol(\ell^{i+1})$
while $\Xoptsol(p_1) < x_1,x_2$ and $x_1,x_2< \Xoptsol(p_2)$. 

We assume that $\cont(\optcell(p_1, \ell^i)) = 2$. The reasoning for $\cont(\optcell(p_1, \ell^i)) = 1$ is analogous, but uses the point $(x_2,y_2)$ instead of $(x_1,y_1)$ and the alternating lines constraint concering $p_2$ instead of $p_1$, while $\cont(\optcell(p_1,\ell^i)) \neq 0$ by the definition of $\stlines_\situ$. 

Recall that $\stseq_\situ = \stseq(\Xoptsol(p_1), \Xoptsol(p_2))$.
Let $j \in [r]$ be the index of the block of $\points_\situ(\Xoptsol(p_1), \Xoptsol(p_2))$
that contains $(x_1,y_1)$.
Since $(x_1,y_1) \in W_1$ but $\cont(\optcell(p_1,\ell^i)) = 2$, we have $i \neq j$.
Observe that $\Yoptsol(q_2) \leq \Yoptsol(\ell^{i+1})$.
Recalling the remaining inequalities, we have $\Yoptsol(\ell^i) \leq \Yoptsol(q_1) < y_1 < \Yoptsol(q_2) \leq \Yoptsol(\ell^{i+1})$. 
If $i > j$ and all monotonicity constraints are satisfied, then the alternating lines
constraint for $p_1$ and line above the $j$-th block is violated. 
Similarly, if $i < j$ and all monotonicity constraints are satisfied, then the alternating lines
constraint for $p_1$ and line below the $j$-th block is violated. 
This is the desired contradiction that finishes the proof of Lemma~\ref{lem:completeness}.
\end{proof}

One remark is in order. A meticulous reader can notice that the proof of Lemma~\ref{lem:completeness} does not in its guts use the filtering step based on Branching Steps~D and~E.
That is, they are not needed to obtain completeness (the conclusion of Lemma~\ref{lem:completeness}).
However, this filtering step has been pivotal in ensuring that alternating lines contraints are sufficiently simple in the proof of Lemma~\ref{lem:imp-depth}.

\subsection{Wrap up}\label[section]{sec:wrap-up}
We now complete the proof of Theorem~\ref{thm:main}.
As already discussed, the branching steps result in $2^{\Oh(k^2 \log k)} \log n$ subcases.
In each subcase, we perform polynomial-time computation that reduces some domains in filtering
steps, possibly discarding the subcase.
If the subcase is not discarded, it produces an auxiliary CSP instance with $k$ variables
and a number of constraints.
There are $\Oh(k)$ monotonicity constraints, $\Oh(k^2)$ corner constraints, 
and $\Oh(k^2)$ alternation constraints, each of constant depth. 
Finally, a situation $\situ = (p_1,p_2,\ell_1,\ell_2)$ of alternation $a \geq 4$
results in $\Oh(a)$ alternating lines constraints
that can be represented as a tree of size $\Oh(a)$ via Lemma~\ref{lem:imp-depth}. 
This tree, if translated directly into a forest CSP instance as discussed in Section~\ref{sec:csp},
yields $\Oh(a)$ variables and constraints. 
There are $\Oh(k)$ choices of the line $p_1$ (which determines $p_2$)
  and, for fixed $p_1$ and $p_2$, the sum of alternations of all situations $(p_1,p_2,\cdot,\cdot)$
  is $\Oh(k)$. 
Hence, adding all alternating lines constraints directly into a forest CSP instance
yields $\Oh(k^2)$ constraints and variables. 
Hence, we obtain a forest-CSP instance of apparent size $\Oh(k^2)$.
This gives fixed-parameter tractability of the \textsc{Optimal Discretization} problem
by Theorem~\ref{thm:csp}
and a running time bound of $2^{\Oh(k^2 \log k)} n^{\Oh(1)}$ by Lemma~\ref{lem:csp}.

\section{Conclusions}\label{sec:conclusions}
We would like to conclude with a number of possible future research directions.

For the \optdis{} problem, the natural direction is to try to improve our running time 
bound $2^{\Oh(k^2 \log k)} n^{\Oh(1)}$. 
Improving the parametric factor to $2^{o(k^2)}$ seems very challenging, as it would require
a significant paradigm shift: In our algorithm, even the most natural branching step (guessing the content of every cell of the solution) yields $2^{\Omega(k^2)}$ subcases. 
A different, but perhaps more accessible, direction would be to analyze and optimize the factor $n^{\Oh(1)}$ in the running time bound.

It would also be interesting to study aspects of effective data reduction.
Does \optdis\ admit a polynomial-size problem kernel\footnote{A problem kernel is a polynomial-time self reduction and its size is an upper bound on the resulting instances in terms of the parameter.}?
Note that a polynomial-size problem kernel for \optdis\ would immediately solve the above question about improving the parametric factor in our running time:
The trivial $n^{k+\Oh(1)}$-time algorithm for \optdis, pipelined with the supposed kernel,
gives a running time of $2^{\Oh(k \log k)} + n^{\Oh(1)}$.
A related and long-open question is whether there is a polynomial-size problem kernel for the fixed-parameter tractable variant of \textsc{Rectangle Stabbing} in which the rectangles are pairwise disjoint~\cite{heggernes_fixedparameter_2013}.
Since working with the geometric representations is difficult, perhaps studying first the kernelization properties of our CSP variant can give some instructive insight.

In a larger scope, the key to our tractability result is the tractability of
the \textsc{Forest CSP} problem.
How large is this isle of tractability in the CSP world?
That is, consider the family of binary CSPs with the number of variables as a parameter. 
Without any restriction on the allowed constraints, this problem is exactly 
\textsc{Multicolored Clique}, the most popular starting point of W[1]-hardness
reductions.
Restricting the allowed constraints to \textsc{Forest CSP} makes the problem fixed-parameter tractable (but still NP-hard, as any permutation can be encoded as a composition of sufficiently many segment reversions), while reducing the allowed constraints to only conjunctions of 
clauses of the form $(x < a) \vee (y > b)$ makes the problem polynomial-time solvable.
New results show a large class of tractable CSP variants~\cite{multicut-3-term} where constraints can be expressed with matrices of bounded grid minor.
Can we further relax the restrictions on the constraints in such \textsc{CSP}s while maintaining fixed-parameter tractability?
It would be interesting to obtain a FPT/W[1]-hard dichotomy here for a wide range of CSPs, but even a good definition of the range of CSPs to consider is unclear to us. 

In the same paper~\cite{multicut-3-term}, we have seen another example of an FPT algorithm for a complex problem via a tractable CSP formulation.
It thus seems that there is further potential in using well-structured CSPs to provide FPT algorithms for difficult tractable problems.

\medskip
\paragraph{Acknowledgments.}
Marcin Pilipczuk would like to thank Brian Lavallee and Blair Sullivan for introducing
him to the problem and sharing some initial thoughts during his visit at NCSU in February 2018.

We acknowledge the comfortable and inspiring atmosphere of Dagstuhl Seminar 19271, where
the details of the algorithm for \textsc{Forest CSP} have been discussed.

\bibliographystyle{plainurl}

\bibliography{references}

\end{document}